\documentclass[11pt]{article}
\usepackage{hyperref}
\usepackage{amsmath}
\usepackage{amssymb}
\usepackage{amsthm}
\usepackage{arydshln}
\usepackage{enumitem}
\usepackage{graphicx}
\usepackage{subcaption}
\usepackage{tikz}
\usepackage{tkz-graph}
\usepackage[maxnames=10, style=alphabetic]{biblatex}
\usepackage[margin=1in]{geometry}
\usepackage{mathdots}
\usepackage{dsfont}
\usepackage{stmaryrd}
\usepackage{leftindex}
\usetikzlibrary{shapes}

\AtBeginBibliography{\small}
\newtheorem{theorem}{Theorem}[section]
\newtheorem*{theorem*}{Theorem}
\newtheorem{proposition}{Proposition}[section]
\newtheorem{lemma}{Lemma}[section]
\newtheorem{corollary}{Corollary}[section]

\theoremstyle{remark}

\newtheorem{claim}{Claim}[section]

\theoremstyle{definition}
\newtheorem{definition}{Definition}[section] 
\newtheorem{example}{Example}[section]

\counterwithin{figure}{section}
\counterwithin{equation}{section}

\DeclareMathOperator{\stab}{Stab}
\DeclareMathOperator{\spn}{span}
\DeclareMathOperator{\tr}{tr}
\DeclareMathOperator{\rr}{\mathbb{R}}
\DeclareMathOperator{\cc}{\mathbb{C}}
\DeclareMathOperator{\kk}{\mathbb{K}}
\DeclareMathOperator{\arity}{arity}
\DeclareMathOperator{\fc}{\mathcal{F}}
\DeclareMathOperator{\gc}{\mathcal{G}}
\DeclareMathOperator{\xc}{\mathcal{X}}
\DeclareMathOperator{\tc}{\mathcal{V}}
\DeclareMathOperator{\ac}{\mathcal{A}}
\DeclareMathOperator{\holant}{Holant}
\DeclareMathOperator{\biholant}{Bi-Holant}
\DeclareMathOperator{\csp}{\#CSP}

\DeclareMathOperator{\eq}{\mathcal{EQ}}

\DeclareMathOperator{\gl}{GL}
\DeclareMathOperator{\va}{\mathbf{a}}
\DeclareMathOperator{\vd}{\mathbf{d}}
\DeclareMathOperator{\vx}{\mathbf{x}}

\DeclareMathOperator{\vk}{\mathbf{K}}
\DeclareMathOperator{\vj}{\mathbf{J}}
\DeclareMathOperator{\sym}{Sym}

\DeclareMathOperator{\qk}{\mathfrak{Q}}
\newcommand{\prop}[1]{\left\langle #1 \right\rangle}
\newcommand{\dprop}[3]{\leftindex_{#2}{\left\langle #1 \right\rangle}_{#3}}
\newcommand{\bdprop}[3]{\leftindex[][|]_{#2}{\big\langle #1 \big\rangle}_{#3}}
\newcommand{\rprop}[2]{\leftindex_{#1}{\mathfrak{R}}_{#2}}
\newcommand{\qprop}[2]{\leftindex_{#1}{(\mathfrak{Q}_{\fc})}_{#2}}

\newcommand\numberthis{\addtocounter{equation}{1}\tag{\theequation}}

\newcommand{\TEdge}[2]{
\draw[thick] (#1) -- (#2) 
    node[draw, fill=white, regular polygon, regular polygon sides=3, minimum size = 3pt, inner sep=1.4pt, pos=0.2] {}
    node[draw, fill=black, regular polygon, regular polygon sides=3, minimum size = 3pt, inner sep=1.4pt, pos=0.8] {};
}
\newcommand{\STEdge}[3]{
\draw[thick] (#1) -- (#2) 
    node[draw, fill=white, regular polygon, regular polygon sides=3, minimum size = 3pt, inner sep=1.4pt, pos=0.44-#3, sloped] {}
    node[draw, fill=black, regular polygon, regular polygon sides=3, minimum size = 3pt, inner sep=1.4pt, pos=0.56-#3, sloped] {};
}

\title{Vanishing Signatures, Orbit Closure, \\ and the Converse of the Holant Theorem}

\author{Jin-Yi Cai\footnote{Department of Computer Sciences, University of Wisconsin-Madison}
        \hspace{4cm} 
        Ben Young{\footnotesize *}\\
    \texttt{\hspace{0.7cm}\href{mailto:jyc@cs.wisc.edu}{jyc@cs.wisc.edu}} \hspace{2.2cm}
\texttt{\href{mailto:benyoung@cs.wisc.edu}{benyoung@cs.wisc.edu}}}
\date{}

\begin{document}
\maketitle

\begin{abstract}
    Valiant's Holant theorem is a powerful tool for algorithms and 
    reductions for
    counting problems.
    It states that if two sets $\fc$ and
    $\gc$ of tensors (a.k.a.~constraint functions or signatures)
    are related by a \emph{holographic transformation},
    then $\fc$ and $\gc$ are \emph{Holant-indistinguishable}, i.e.,
    every 
    tensor network using tensors 
    from $\fc$,  respectively from $\gc$,  contracts to the same value.
    Xia (ICALP 2010) conjectured the converse of the Holant theorem, but 
    a counterexample was found 
    based on \emph{vanishing
    signatures}, those which are Holant-indistinguishable from 0. 
    
    We prove two 
    near-converses of the Holant theorem using techniques from invariant
    theory.
    (I) Holant-indistinguishable $\fc$ and $\gc$ always admit
    two sequences of holographic transformations mapping them
    arbitrarily close to each other, i.e., 
    their $\gl_q$-orbit closures intersect.
    (II) We show that vanishing signatures are
the only true obstacle to a
converse of the Holant theorem.
    As corollaries of the two theorems
    we obtain
    the first characterization of homomorphism-indistinguishability over graphs of bounded
    degree, 
    a long standing open problem, 
    and show that two graphs with invertible adjacency matrices are isomorphic
    if and only if they are homomorphism-indistinguishable over graphs with maximum degree
    at most three. We also show that 
    Holant-indistinguishability is
    complete for a complexity class \textbf{TOCI} introduced by Lysikov and Walter \cite{lysikov}, and hence hard for graph isomorphism.
\end{abstract}

\section{Introduction}
Let $F,G \in \cc^{q\times q}$ be two matrices. If $F$ and $G$ are similar, 
then $\tr(F^k) = \tr(G^k)$ for every $k$ -- that is,
$F$ and $G$ are \emph{indistinguishable} by the function $\tr((\cdot)^k)$. Conversely, if
$\tr(F^k) = \tr(G^k)$ for every $k$, then we may only conclude that $F$ and $G$ have the same
multiset of eigenvalues; $F$ and $G$ are not necessarily similar.
In addition,  what other assumptions on $F$ and $G$ suffice to
obtain similarity? The Holant theorem and questions about its converse are
vast generalizations of this example.

\paragraph{Holant Problems and the Holant Theorem.}
Holant problems, first  introduced 
in~\cite{cai_computational_2011}, are a framework
for expressing counting problems on graphs. Let $\fc$ be a set of tensors over a 
finite-dimensional vector space $\kk^q$ (typically $\kk = \cc$).
A \emph{signature grid}, or a tensor network, is a (multi)graph $\Omega$ with vertices assigned
tensors from $\fc$ and edges act as variables.
Depending on the choice of $\fc$,
one can express many counting problems as the \emph{Holant value} $\holant_{\fc}(\Omega)$,
the 
contraction of $\Omega$ as a tensor network. These  include
the number of matchings, proper vertex or edge-colorings, and Eulerian orientations
of $\Omega$ and the number of homomorphisms from $\Omega$ to a possibly weighted and
directed graph $G$. While Holant is very 
expressive,
it is also restrictive enough to prove sweeping dichotomy theorems.
These
classify $\holant_{\fc}$ as either P-time tractable or \#P-hard for
\emph{any} 
set $\fc$~\cite{cai_computational_2011,cai2008holographic,huang_2016_dichotomy,cai_complete_2016,cai2015holant,lin_complexity_2018,shao,cai2013dichotomy,domain3}.
While most existing work focuses on domain size $q =2$ or 3,
the current work is for all $q$.

Valiant's \emph{Holant theorem} \cite{valiant,valiant_2006_accidental}, the genesis
for Holant problems,
states that:
If two sets $\fc$ and $\gc$ of tensors are related by a \emph{holographic transformation}
-- essentially a basis change by a $T \in \gl_q$ -- then $\fc$ and $\gc$ are 
\emph{Holant-indistinguishable}, meaning
that every signature grid $\Omega$ has the same Holant value whether its vertices are assigned
tensors from $\fc$ or the corresponding transformed tensors in $\gc$. This implies that
$\holant_{\fc}$ and $\holant_{\gc}$ have the same complexity, leading to the notions of
\emph{holographic reductions} between Holant problems and \emph{holographic algorithms}. 
Later work~\cite{cai2007valiant, dimensionality, art_to_science} formalized the Holant theorem and 
holographic reductions
in terms of covariant
and contravariant tensors.
In this form, $\Omega$
is a bipartite graph whose two bipartitions are assigned contravariant
tensors from $\fc$ and covariant tensors from $\fc'$, respectively. The problem is
denoted $\holant_{\fc|\fc'}$. 
Xia \cite{xia} conjectured the converse of the Holant theorem: if $\fc|\fc'$ and $\gc|\gc'$ 
are Holant-indistinguishable,
then there is
a holographic transformation between them.
But 
    a counterexample was found in~\cite{cai_complete_2016}  
    based on \emph{vanishing
    signatures}, those $\fc$ which are Holant-indistinguishable from the set 
of all-0 tensors. 

\paragraph{Homomorphism Indistinguishability.}
The Holant framework is broader than graph homomorphism~\cite{lovasz_operations,Hell-Nesetril}.
The results in this work encompass a long list of other results in this area of research.
Most prominently this includes \emph{homomorphism indistinguishability} of graphs.
Lovász \cite{lovasz_operations} showed that two graphs $F$ and $G$ are isomorphic if and only 
if they admit the same number of homomorphisms from all graphs. This result was later improved
to $F$ and $G$ with edge and vertex weights \cite{lovasz,schrijver,cai-lovasz}.
Another line of work aims to determine the relaxations of isomorphism which must relate
any $F$ and $G$ indistinguishable under homomorphisms \emph{from} restricted classes of 
graphs \cite{dvorak_recognizing_2010,dell,planar,kar2025npa,rattan_weisfeiler,Grohe2025Homomorphism,lasserre}.
One notable graph class whose homomorphism indistinguishability relation had, since the
seminal 2010 work of Dvořák \cite{dvorak_recognizing_2010}, eluded any full
characterization is the graphs of bounded degree. 
Roberson \cite{oddomorphism} showed that
homomorphism indistinguishability from graphs of degree at most $d$ define distinct 
relations strictly weaker than isomorphism
on the set of graphs for distinct $d$, but
did not characterize them further. By expressing 
bounded-degree graph homomorphism as a bipartite 
Holant problem, we obtain as a corollary of our first main theorem the first
characterization of its indistinguishability relation.

Indistinguishability theorems also exist for other subclasses of Holant, including \#CSP and vertex and edge-coloring models
\cite{szegedy_edge_2007,schrijver_graph_2008,regts_edge_reflection,cai_planar_2023,orthogonal}.
The connections developed in this work 
demonstrate the advantage of expressing, via Holant, counting problems such as graph 
homomorphism and \#CSP as tensor networks, which appear in a host of other areas and are
subject to powerful theorems from invariant theory.

\paragraph{Orbit Equality and Orbit Closure Intersection.}
The $\gl_q$-\emph{orbit} of a finite set $\fc$ of tensors is the set
$\{T \cdot \fc \mid T \in \gl_q\}$, where $T$ acts simultaneously on the tensors in $\fc$, in
our setting by holographic transformation. Therefore the converse of the
Holant theorem would state that, if $\fc|\fc'$ and $\gc|\gc'$ are Holant-indistinguishable,
then the $\gl_q$-orbits of $\fc|\fc'$ and $\gc|\gc'$ intersect and hence are equal. A weaker
and often better-behaved notion is that of orbit \emph{closure} intersection
(Euclidean closure, for $\kk = \cc$).
There has been much research in recent years
on the computational complexity of orbit intersection and orbit closure intersection for various actions of a linear-algebraic group $H$
\cite{complexity_isomorphism1,complexity_isomorphism3,complexity_isomorphism4,
orbit_closure_matrix,derksen2020algorithms,operator_scaling,search_problems,minimal_canonical,
lysikov} with connections to geometric complexity theory 
\cite{Landsberg_2017}, including border rank
with applications to matrix multiplication \cite{burgisser2011geometric}, and polynomial identity testing.

Several such works \cite{derksen2020algorithms,search_problems,orbit_closure_matrix,minimal_canonical,lysikov}
apply a theorem 
(\autoref{thm:mumford} below) from geometric
invariant theory which states that the $H$-orbit closures of $\fc$ and $\gc$ intersect if 
and only if $\fc$ and $\gc$ are indistinguishable over all $H$-invariant polynomials
(i.e. every such polynomial takes the same value on inputs $\fc$ and $\gc$).
Acuaviva et al. \cite[Theorem 4.11]{minimal_canonical} prove an orbit-closure indistinguishability theorem
for a family of vertex-regular tensor networks from quantum physics called PEPS 
networks, which admit a variant of holographic transformation called a \emph{gauge transformation}. A PEPS signature set $\fc$ has common arity $2n$, with inputs paired into $n$ dimensions (with possibly distinct domains) and only allows contractions between inputs in the same dimension.
Lysikov and Walter \cite{lysikov} define the complexity class
\textbf{TOCI} of orbit closure intersection problems, showing that it
contains  \textbf{GI} (all problems  reducible
to graph isomorphism).

\paragraph{Our Results.}
We develop new connections between invariant theory and counting problems to prove
two near-converses of the Holant theorem. First, we show
that the converse of the Holant theorem holds for orbit
closure intersection instead of orbit intersection as conjectured in~\cite{xia}.
\begin{theorem*}[first main theorem, \autoref{thm:main1}]
    Finite $\fc|\fc'$ and $\gc|\gc'$ are Holant-indistinguishable if and only if the $\gl_q$-orbit
    closures of $\fc|\fc'$ and $\gc|\gc'$ intersect.
\end{theorem*}
This means, there are two sequences of holographic transformations taking $\fc|\fc'$ and
$\gc|\gc'$ arbitrarily close to each other.
The key idea in the proof is to show that every $\gl_q$-invariant
polynomial is realizable as a sum of the Holant values of indeterminate-valued signature
grids. A special case is a characterization of vanishing sets which applies to any set on
any domain. This greatly generalizes 
the symmetric Boolean-domain characterization of \cite{cai_complete_2016}. 
It also follows that the problem of testing whether $\fc|\fc'$ and 
$\gc|\gc'$ are Holant-indistinguishable is decidable. 

Our second near-converse of the Holant theorem does give a true holographic transformation
between $\fc|\fc'$ and $\gc|\gc'$, but requires that 
$\fc|\fc'$ and $\gc|\gc'$ be \emph{quantum-nonvanishing}. Roughly, this means that $\fc|\fc'$
cannot produce a \emph{quantum gadget} (a linear combination of contractions of tensors
in $\fc|\fc'$) that causes every $\fc|\fc'$-grid containing it to have Holant value 0.
Quantum gadgets generalize several other constructions used 
in counting indistinguishability, including homomorphism tensors/bi-labeled graphs 
\cite{planar, Grohe2025Homomorphism, kar2025npa} and their namesake, quantum
labeled graphs \cite{freedman_reflection,lovasz,dvorak_recognizing_2010} 
(in fact, quantum-vanishing signatures generalize
the concept of the annihilator of the quantum labeled graph algebra
\cite{freedman_reflection}).
\begin{theorem*}[second main theorem, \autoref{thm:main2}]
    If $\fc|\fc'$ and $\gc|\gc'$ are Holant-indistinguishable and quantum-nonvanishing, then
    there is a holographic transformation between $\fc|\fc'$ and $\gc|\gc'$.
\end{theorem*}
The proof of this theorem uses an invariant-theoretic 
characterization due to Derksen and Makam \cite{prop} of the quantum $\fc|\fc'$-gadget 
algebra for quantum-nonvanishing $\fc|\fc'$, analogous to the duality theorems used by
\cite{planar, cai_planar_2023, orthogonal} to prove their indistinguishability results.
However, the quantum-nonvanishing requirement adds new difficulties.
We use Derksen and Makam's theorem
to initially split the problem into two subdomains, then gradually refine these subdomains
by holographic transformations until quantum-nonvanishing forces $\fc|\fc' = \gc|\gc'$.
We use similar techniques to prove \autoref{thm:binary}, a variant of the second main theorem
for quantum-nonvanishing sets 
$\fc$ and $\gc$ of matrices: every product of matrices in $\fc$ has the same trace as the 
corresponding product in $\gc$ if and only if $\fc$ and $\gc$ are simultaneously similar.
The proof of this result is `constructive' in the sense that the recovered transformation 
between $\fc$ and $\gc$ is composed of Jordan decompositions of
quantum-$\fc$-gadget-realizable matrices, and of these matrices themselves
(although the gadgets are obtained nonconstructively). The proof
of the second main theorem is similarly `constructive' except for the application of
Derksen and Makam's theorem.

In \autoref{sec:corollaries}, we use the second main theorem to show that, 
while homomorphism
indistinguishability of graphs $F$ and $G$ over graphs of any bounded degree is not in 
general equivalent to isomorphism, homomorphism indistinguishability over graphs of
maximum degree at most three is equivalent to isomorphism for $F$ and $G$ with invertible
adjacency matrices.
We also apply the first main theorem and results of \cite{lysikov} to show that the problem of Holant-indistinguishability is \textbf{TOCI}-complete and \textbf{GI}-hard.

\section{Background and Preliminaries}
Throughout, let $\kk$ be an algebraically closed field of characteristic 0. 
We work with the finite-dimensional vector space $\kk^q$ and its dual space $(\kk^q)^*$. 
The \emph{mixed tensor algebra} over $\kk^q$ is
\[
    \tc = \tc(\kk^q) := \bigcup_{\ell,r \geq 0} \leftindex_\ell{\tc}_r, \text{ where }
    \leftindex_\ell{\tc}_r = (\kk^q)^{\otimes \ell} \otimes ((\kk^q)^*)^{\otimes r}.
\]
$\tc(\kk^q)$ is bigraded $\kk$-vector space (each grade $\leftindex_\ell{\tc}_r$ is a $\kk$-vector space) and admits the usual tensor product
$\otimes: \leftindex_{\ell_1}{\tc}_{r_1} \times \leftindex_{\ell_2}{\tc}_{r_2} \to \leftindex_{\ell_1+\ell_2}{\tc}_{r_1+r_2}$.
Tensors in 
$\bigcup_{n \geq 1} \leftindex_{n}{\tc}_{0}\subset \tc(\kk^q)$
are called
\emph{contravariant} (or as column vectors lexicographically indexed), and tensors in
$\bigcup_{n \geq 1} \leftindex_{0}{\tc}_{n} $ are called \emph{covariant} (row vectors).
Tensors in $\leftindex_{\ell}{\tc}_r$ for $\ell r > 0$ ($q^\ell \times q^r$ matrices) 
are \emph{mixed}. Note that $\leftindex_0{\tc}_0 = \kk$. 

Given $A = \sum_{i,j = 1}^q a_{i,j} e_i \otimes e_j \in (\kk^q)^{\otimes 2}$, define $A^{1,1} = \sum_{i,j = 1}^q a_{i,j} e_i \otimes e_j^* \in \kk^q \otimes (\kk^q)^*$, also thought of as a
matrix $(a_{i,j})_{i,j=1}^q \in \kk^{q \times q}$. Define $A^{1,1}$ for
binary covariant $A$ similarly.

\subsection{Holant and Bi-Holant}
A \emph{signature} is a function $F: [q]^{n} \to \kk$ on $n = \arity(F)$ inputs from a finite 
\emph{domain} $[q]$. Use $\fc$ to denote a set of signatures sharing a common domain $[q]$,
but possibly with different arities. Given $\fc$, a \emph{signature grid} (or $\fc$-grid) 
$\Omega$ is a multigraph
along with an assignment of an $n$-ary $F_v \in \fc$ to each degree-$n$ 
vertex $v$ in $\Omega$, with an ordering of the $n$ edges $\delta(v)$ incident to $v$ to 
serve as the $n$
inputs to $F$. 
For technical reasons, we also allow $\Omega$ to contain \emph{vertexless loops}
$\bigcirc$ (a loop with one edge  and no vertex).
The goal of $\holant_{\fc}$ is to compute the \emph{Holant value}
\[
    \holant_{\fc}(\Omega) = \sum_{\sigma: E(\Omega) \to [q]} \prod_{v \in V(\Omega)}
    F_v(\sigma(\delta(v)))
\]
of $\Omega$, where $F_v(\sigma(\delta(v)))$ is the evaluation of $F_v$ on the $n$ domain
elements assigned to the incident edges of $v$. Each vertexless loop in $\Omega$ 
contributes a factor  $q$.
Note that the Holant
value of a disconnected signature grid is the product of the Holant values of its connected
components.

For example, suppose $\fc$ consists of, for each $n \geq 1$, the $n$-ary $\{0,1\}$-valued 
signature on the Boolean domain $\{0,1\}$ ($q = 2$) that evaluates to 1 if at most one
(resp. exactly one) of its inputs is 1, and evaluates to 0 otherwise. For any  $\Omega$ without  vertexless
loops,  let $\sigma$ have a nonzero evaluation.
The edges assigned 1 
form a matching (resp.
perfect matching) in $\Omega$, so
$\holant_{\fc}(\Omega)$ equals the number of matchings (resp. perfect matchings) in $\Omega$.

The coefficients of a tensor $F \in \leftindex_\ell{\tc}_r$ 
define an $\ell+r$-arity signature on domain $[q]$. In this work, we will generally think of 
signatures as tensors in this way. We view a single $n$-ary signature as taking different 
shapes (i.e. different choices of $(\ell, r)$: 
$\ell+r = n$) or, as in the case of unrestricted Holant above, ignore this
covariant/contravariant input distinction, depending on the 
context. Viewing signatures as fully contravariant or covariant gives the following
well-studied bipartite setting.
\begin{definition}[$\holant_{\fc|\fc'}$]
    Let $\fc$ and $\fc'$ be sets of contravariant and covariant
    tensors, respectively. An $(\fc|\fc')$-grid $\Omega$ is a bipartite 
    $(\fc \sqcup \fc')$-grid $\Omega$ in which the vertices in the two bipartitions are
    assigned signatures from $\fc$ and $\fc'$, respectively.
\end{definition}
\begin{definition}[$\eq, =_n$]
    Define the set of \emph{equality} signatures
    $\eq = \{=_n \mid n \geq 1\}$, where $=_n$ is the $n$-ary 
    signature defined by $(=_n)(x_1,\ldots,x_n) = 1$ if $x_1 = \ldots = x_n$, and 0 otherwise.
\end{definition}
\begin{proposition}
    \label{prop:shapeless}
    For any $\fc \subset \tc$, $\holant_{\fc}$, $\holant_{=_2|\fc}$, 
    and $\holant_{=_2|\fc,=_2}$ are equivalent.
\end{proposition}
\begin{proof}
    Convert an
    $\fc$-grid $\Omega$ into a $(=_2|\fc)$-grid (which is also a 
    $(=_2|\fc, =_2)$-grid) by placing a degree-2 vertex 
    assigned $=_2$ on each edge. The resulting grid is bipartite between $=_2$ and $\fc$ and,
    since $=_2$ acts identically to an edge, does not change the Holant value.
    Conversely, given an $(=_2|\fc,=_2)$-grid $\Omega$, 
    replace each vertex assigned $=_2$ with an edge. This connects arbitrary inputs of 
    signatures in $\fc$, but this is allowed in $\holant_{\fc}$.
\end{proof}

For a problem (only easily) expressible in the bipartite setting, consider the problem of 
counting homomorphisms from graphs of bounded degree. A \emph{graph homomorphism} from graph 
$X$ to graph $G$ is a map $\varphi: V(X) \to V(G)$ such that, for every edge $uv$ of $X$, 
$\varphi(u)\varphi(v)$ is an edge of $G$. Let $V(G) = [q]$ and
$A_G \in \{0,1\}^{q \times q}$ be the adjacency matrix of $G$, thought of as a binary
signature. Construct
a bijection between left-side graphs $X$ and (vertexless-loop-less)
$\eq|A_G$-grids $\Omega_X$ as shown in \autoref{fig:hom}.
\begin{figure}[ht!]
    \center
    \begin{tikzpicture}[scale=0.75]
    \GraphInit[vstyle=Classic]
    \SetUpEdge[style=-]
    \SetVertexMath

    \node[font=\Large] at (-2.5,0) {$X$};
    \Vertex[x=0,y=0,NoLabel]{a1}
    \Vertex[x=-1.5,y=1,NoLabel]{b1}
    \Vertex[x=-1.5,y=-1,NoLabel]{c1}
    \Vertex[x=1.5,y=0,NoLabel]{d1}

    \begin{scope}[xshift=6cm]
        \node[font=\Large] at (-2.5,0) {$\Omega_X$};
        \Vertex[x=0,y=0,L={=_3},Lpos=90]{a2}
        \Vertex[x=-1.5,y=1,L={=_2},Lpos=90]{b2}
        \Vertex[x=-1.5,y=-1,L={=_2},Lpos=270]{c2}
        \Vertex[x=1.5,y=0,L={=_1},Lpos=0]{d2}

        \tikzset{VertexStyle/.style = {shape=rectangle, fill=black, minimum size=5pt, inner sep=1pt, draw}}
        \Vertex[x=0.75,y=0,NoLabel]{ax1}
        \Vertex[x=-1.5,y=0,NoLabel]{ax2}
        \Vertex[x=-0.75,y=0.5,NoLabel]{ax3}
        \Vertex[x=-0.75,y=-0.5,NoLabel]{ax4}
    \end{scope}

    \foreach \num in {1,2} {
        \Edges(a\num,b\num,c\num,a\num,d\num)
    }

    \Vertex[x=10,y=0.4,L={\in \eq}]{ax8}
    \tikzset{VertexStyle/.style = {shape=rectangle, fill=black, minimum size=5pt, inner sep=1pt, draw}}
    \Vertex[x=10,y=-0.1,L={=A_G}]{ax7}

\end{tikzpicture}
    \caption{A graph $X$ and $\eq|A_G$-grid $\Omega_X$ such that $\holant(\Omega_X) = \hom(X,G)$.}
    \label{fig:hom}
\end{figure}
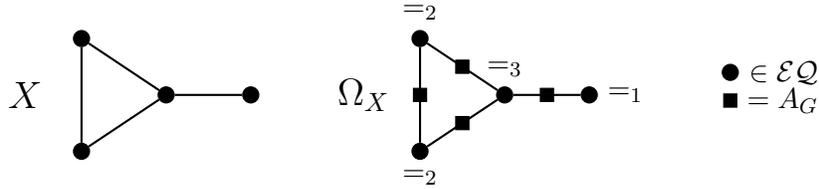
Each equality signature, assigned to an original $X$ vertex, forces all incident edges to 
take the same value in $[q] = V(G)$. Therefore any nonzero edge assignment in $\Omega$ defines
a map $V(X) \to V(G)$. The $A_G$ signatures then enforce that this map is a graph 
homomorphism. Thus $\holant_{\eq|A_G}(\Omega_X) = \hom(X,G)$, the number of 
homomorphisms from $X$ to $G$. By the same construction, defining $\eq_{\leq d} \subset \eq$ 
to be the set of equality signatures of arity at most $d$, 
$\holant_{\eq_{\leq d}|A_G}$ captures the problem of counting homomorphisms from
graphs $X$ of maximum degree at most $d$ to $G$. $\holant_{\eq | A_G}$ is equivalent to
the non-bipartite
$\holant_{\eq\cup\{A_G\}}$ because we can, without affecting the Holant value, 
insert a dummy $=_2$ between any two adjacent $A_G$ vertices and combine 
adjacent $=_a$ and $=_b$ vertices into a single vertex assigned $=_{a+b-2}$.
However, expressing homomorphisms from 
bounded-degree graphs does require bipartiteness, because merging two equality signatures 
of arity $\leq d$ could produce an equality signature of arity $> d$. 

\begin{definition}[Bi-Holant]
    For $\fc \subset \tc(\kk^q)$, a \emph{Bi-Holant} $\fc$-grid $\Omega$ is a Holant
    $\fc$-grid respecting the shapes of its signatures -- that is, the edge between any 
    adjacent $u$ and $v$ must be a contravariant input to $F_u$ and a covariant input to
    $F_v$, or vice-versa.
\end{definition}
In particular, if $\fc$ and $\fc'$ are sets of contravariant and covariant tensors, then
$\biholant_{\fc \cup \fc'}$ is equivalent to $\holant_{\fc|\fc'}$. Therefore, by
\autoref{prop:shapeless}, Bi-Holant generalizes Holant.
\begin{definition}[(Bi-Holant) $\fc$-gadget]
    For $\fc \subset \tc(\kk^q)$, an
    $\fc$-\emph{gadget} is Bi-Holant $\fc$-grid in which zero or more edges are
    \emph{dangling}, with zero or one endpoints not incident to any vertex.
    In an $(\ell,r)$-$\fc$-gadget, $\ell$ dangling edges are contravariant inputs to their
    incident signatures, and $r$ are covariant, drawn with left-facing and right-facing 
    dangling ends, respectively. 
    The dangling ends are ordered from top to bottom on both
    the left and right. 
    A two-sided dangling edge (called a \emph{wire}) always has one 
    contravariant and one covariant dangling end. 

    The \emph{signature} $K \in \leftindex_\ell{\tc}_r$ of an $(\ell,r)$-$\fc$-gadget $\vk$ 
    is defined by letting
    $K(a_1,\ldots,a_\ell,b_1,\ldots,b_r)$ be the Holant value of $\vk$ when the $\ell$ left
    and $r$ right dangling edges are fixed to values $a_1,\ldots,a_\ell$ and $b_1,\ldots,b_r$
    (summing only over assignments $\sigma$ to the internal edges). The signature of a wire
    is $(=_2) = I \in \kk^q \times (\kk^q)^*$, as the inputs to its two ends must match.
\end{definition}
Gadget signatures are defined so that, if $F$ is the signature of an $\fc$-gadget $\vk_F$, then any
$(\fc \cup \{F\})$-grid corresponds to an $\fc$-grid with the same Holant value constructed
by replacing every instance of $F$ and its incident edges with $\vk_F$ (with appropriately
ordered dangling edges). $\biholant_{\fc}(\Omega)$ is the value of the contraction of $\Omega$
as a tensor network with the usual primal/dual tensor contraction -- for example, slicing the
$n$ edges of an $\fc|\fc'$-grid $\Omega$ yields two gadgets with signatures 
$F_1 \in (\kk^q)^{\otimes n}$ and $F_2^* \in ((\kk^q)^*)^{\otimes n}$ such that
$\holant_{\fc}(\Omega) = F_2^*(F_1)$. Similarly, if $F_1,F_2 \in \kk^q \otimes (\kk^q)^*$,
then the signature formed by contracting the right input of $F_1$ with the 
left input of $F_2$ is the (matrix) composition of $F_1$ and $F_2$. 

\begin{figure}[ht!]
    \begin{center}
    \begin{subfigure}{0.35\textwidth}
        \centering
        \begin{tikzpicture}[scale=.4]
\GraphInit[vstyle=Classic]
\SetUpEdge[style=-]
\SetVertexMath

\def\ysh{2}
\def\wlen{2}
\def\wgap{1.5}

\draw[thin, color=gray] (0,0) .. controls
+(0.8,0) and +(0,-1) ..
(\wlen,\ysh/2) .. controls
+(0,1) and +(0.8,0) ..
(0,\ysh);

\begin{scope}[xshift=5.5cm]
\draw[thin, color=gray] (0,0) .. controls
+(-0.8,0) and +(0,-1) ..
(-\wlen,\ysh/2) .. controls
+(0,1) and +(-0.8,0) ..
(0,\ysh);
\end{scope}

\draw[thin, color=gray] (0,-\wgap) -- (2*\wlen + \wgap, -\wgap);


\end{tikzpicture}
        \vspace{0.2cm}
        \caption{$=_2$ in $\leftindex_2{\tc}_0$ (top left), $\leftindex_0{\tc}_2$ (top
        right) and $\leftindex_1{\tc}_1$ (bottom), respectively.}
        \label{fig:wires}
    \end{subfigure}
    \hspace{0.5cm}
    \begin{subfigure}{0.6\textwidth}
        \centering
        \begin{tikzpicture}[scale=0.7]
\GraphInit[vstyle=Classic]
\SetUpEdge[style=-]
\SetVertexMath

\def\wlen{1}

\draw[thin, color=gray] (0,0) -- (-\wlen,0);
\draw[thin, color=gray] (1,1) -- (-\wlen,1);
\draw[thin, color=gray] (2,2) .. controls +(0,0.4) .. (-\wlen,2.4);
\draw[thin, color=gray] (2,0) -- (2+\wlen,0);

\tikzset{VertexStyle/.style = {shape=circle, fill=black, minimum size=6pt, inner sep=1pt, draw}}
\Vertex[x=0,y=0,NoLabel]{u0}
\Vertex[x=1,y=1,NoLabel]{u1}
\Vertex[x=2,y=2,NoLabel]{u2}
\tikzset{VertexStyle/.style = {shape=rectangle, fill=black, minimum size=6pt, inner sep=1pt, draw}}
\Vertex[x=0,y=2,NoLabel]{v0}
\Vertex[x=2,y=0,NoLabel]{v1}
\Edges(u0,v0,u2,v1,u0)
\Edges(v0,u1,v1)

\node at (-\wlen-0.7,1) {\huge 5};

\begin{scope}[xshift=6cm]
    \node at (-\wlen-1.3,1) {\huge +~7};

    \draw[thin, color=gray] (2,0) -- (2+\wlen,0);
    \draw[thin, color=gray] (0,1) .. controls +(0,1.4) .. (-\wlen,2.4);
    \draw[thin, color=gray] (0,1) -- (-\wlen,1);
    \draw[thin, color=gray] (0,1) .. controls +(0,-1) .. (-\wlen,0);

    \Vertex[x=1,y=1.5,NoLabel]{vv0}
    \Vertex[x=2,y=0,NoLabel]{vv1}
\tikzset{VertexStyle/.style = {shape=circle, fill=black, minimum size=6pt, inner sep=1pt, draw}}
    \Vertex[x=0,y=1,NoLabel]{uu0}
    \Vertex[x=2,y=2,NoLabel]{uu1}
    \Edges(uu1,vv0,uu0,vv1)

    \Vertex[x=2+\wlen+0.7,y=1.7,L={\in \fc}]{f}
\tikzset{VertexStyle/.style = {shape=rectangle, fill=black, minimum size=6pt, inner sep=1pt, draw}}
    \Vertex[x=2+\wlen+0.7,y=1.1,L={\in \fc'}]{fp}
\end{scope}

\end{tikzpicture}
        \caption{A $(3,1)$-quantum-$\fc|\fc'$-gadget. Note that left/right dangling edges are
        incident to vertices in $\fc$/$\fc'$, respectively.}
        \label{fig:quantum_gadget}
    \end{subfigure}
    \end{center}
    \caption{Examples of (quantum) gadgets}
\end{figure}
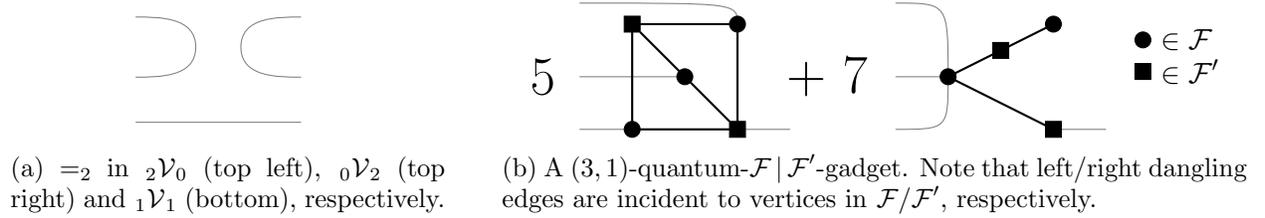

\begin{definition}[$\qk_{\fc}, \prop{\fc}$]
    An \emph{$(\ell,r)$-quantum-$\fc$-gadget} $\vk$ is a formal $\kk$-linear combination of 
    $(\ell,r)$-$\fc$-gadgets. Any component of a term of $\vk$ without any 
    dangling edges evaluates to a scalar and is absorbed into term's coefficient, 
    so assume no term of $\vk$ has any such components.

    Define $\qk_{\fc}$ and $\prop{\fc}$ to be the spaces of all quantum-$\fc$-gadgets and
    quantum-$\fc$-gadget signatures, respectively (extending the gadget signature function 
    linearly), and
    $\dprop{\fc}{\ell}{r} := \prop{\fc} \cap \leftindex_{\ell}{\tc}_r$.
\end{definition}
See \autoref{fig:quantum_gadget}.
Extend left/right dangling edge contraction linearly
to $\qk_{\fc}$. Note that $\qk_{\fc}$ and $\prop{\fc}$ are closed under quantum gadget
construction, as every $\prop{\fc}$-gadget $\vk$ decomposes into an quantum-$\fc$-gadget after
replacing every
$F \in \prop{\fc} \setminus \fc$ in $\vk$ with the quantum-$\fc$-gadget realizing $F$ and 
expanding linearly. 
For every $\ell,r \geq 0$, we have the standard bilinear form 
$\langle\cdot,\cdot\rangle: \leftindex_\ell{\tc}_r \times \leftindex_r{\tc}_\ell \to \kk$. If
$K \in \dprop{\fc}{\ell}{r}$ and $K' \in \dprop{\fc}{r}{\ell}$ are the signatures of
$\prop{\fc}$-gadgets $\vk$ and $\vk'$, then $\langle K, K'\rangle = \holant(\Omega)$, where 
$\Omega$ is constructed by connecting the right inputs of $\vk$ with the left inputs of $\vk'$
and vice-versa (this extends bilinearly to quantum gadgets).

While $\dprop{\fc}{1}{1}$ always contains $I = (=_2)^{1,1}$ as the signature of
a wire, we do not always have $(=_2) \in \dprop{\fc}{0}{2}$ or $(=_2) \in \dprop{\fc}{2}{0}$
(see \autoref{fig:wires});
such a co/contravariant $(=_2)$ is quite powerful, as it allows connecting two left
or two right dangling edges with each other, circumventing bipartiteness (as seen in 
\autoref{prop:shapeless}), and allows
reshaping tensors -- e.g. construct $A^{1,1}$ from 
$A \in (\kk^q)^{\otimes 2}$ by connecting a right-facing $(=_2)$ to the second input of $A$.

\subsection{Transformations, Indistinguishability, and the Holant Theorem}
\label{sec:transform}
Throughout, we treat pairs $\fc,\gc \subset \tc$ of signature sets that are \emph{bijective},
meaning there is an arity-preserving bijection $\leftrightsquigarrow$ between $\fc$ and $\gc$.
Call $\fc \ni F \leftrightsquigarrow G \in \gc$ \emph{corresponding} signatures.
\begin{definition}[$(\cdot)_{\fc\to\gc}$, (Bi-)Holant-indistinguishable]
    Given a $\vk \in \qk_{\fc}$ (possibly with no dangling edges, in which case
    $\vk = \Omega$ is a quantum $\fc$-grid), construct
    $\vk_{\fc\to\gc} \in \qk_{\gc}$ by replacing every $F \in \fc$ in every term of $\vk$ 
    with the corresponding $G \in \gc$.

    Say that $\fc$ and $\gc$ are \emph{Holant-indistinguishable} if 
    $\holant_{\fc}(\Omega) = \holant_{\gc}(\Omega_{\fc\to\gc})$ for every $\fc$-grid
    $\Omega$. Define Bi-Holant-indistinguishable similarly.
\end{definition}
The $(\cdot)_{\fc\to\gc}$ operation induces a bijection between $\prop{\fc}$ and $\prop{\gc}$,
where $K \leftrightsquigarrow \widetilde{K}$ if $K$ and $\widetilde{K}$ are the signatures
of $\vk$ and $\vk_{\fc\to\gc}$ (viewing $\prop{\fc}$ and 
$\prop{\gc}$ as multisets in bijection with $\qk_{\fc}$). Under this bijection, if $\fc$ and
$\gc$ are (Bi-)Holant-indistinguishable, then so are $\prop{\fc}$ and $\prop{\gc}$.

\begin{definition}[$T \cdot F$, $T\fc$]
    \label{def:action}
    For $T \in \gl_q$ and $F \in \leftindex_\ell{\tc}_r$,
    define $T \cdot F := T^{\otimes \ell} F (T^{-1})^{\otimes r}$. Then for 
    $\fc \subset \tc$, define $T\fc = \{T \cdot F \mid F \in \fc\}$.
\end{definition}

\begin{theorem}[The Holant Theorem \cite{valiant}]
    If $\fc|\fc' = T(\gc|\gc')$ for $T \in \gl_q$, then $\fc|\fc'$ and
    $\gc|\gc'$ are Holant-indistinguishable.
    \label{thm:holant}
\end{theorem}
\autoref{thm:holant} follows from the fact that left/right contractions are
$\gl_q$-equivariant for the action of $\gl_q$ in \autoref{def:action}.
See \autoref{fig:holant}
\begin{figure}[ht!]
    \centering
    \begin{tikzpicture}[scale=.76]
    \tikzstyle{every node}=[font=\small]
    \GraphInit[vstyle=Classic]
    \SetUpEdge[style=-]
    \SetVertexMath

    \def\ysh{1.5}
    \def\gap{1.5}

    \foreach \yy in {1,2} {
        \Vertex[x=0,y={-{(\yy+0.5)}*\ysh},L={F'_\yy},Lpos=180]{f\yy}
    }
    \foreach \yy in {1,2,3} {
        \Vertex[x=\gap,y=-\yy*\ysh,L={F_\yy}]{ff\yy}
    }
    \Edges(ff1,f1,ff2,f2,ff3,f1)

    \tikzset{VertexStyle/.style = {shape=rectangle, fill=black, minimum size=5pt, inner sep=1.2pt, draw}}

    \begin{scope}[xshift=4.6cm]
    \node at (-1.5,-2*\gap) {$=$};
    \foreach \yy in {1,2} {
        \Vertex[x=0,y=-{(\yy+0.5)}*\ysh,L={G'_\yy},Lpos=180]{g\yy}
    }
    \foreach \yy in {1,2,3} {
        \Vertex[x=1.3*\gap,y=-\yy*\ysh,L={G_\yy}]{gg\yy}
    }
    \TEdge{g1}{gg1}
    \TEdge{g1}{gg2}
    \TEdge{g1}{gg3}
    \TEdge{g2}{gg2}
    \TEdge{g2}{gg3}
    \end{scope}

    \begin{scope}[xshift=9.8cm]
    \node at (-1.6,-2*\gap) {$=$};
    \foreach \yy in {1,2} {
        \Vertex[x=0,y=-{(\yy+0.5)}*\ysh,L={G'_\yy},Lpos=180]{h\yy}
    }
    \foreach \yy in {1,2,3} {
        \Vertex[x=1.8*\gap,y=-\yy*\ysh,L={G_\yy}]{hh\yy}
    }
    \STEdge{h1}{hh1}{0}
    \STEdge{h1}{hh2}{-0.2}
    \draw[thick] (h1) -- (hh3) 
        node[draw, fill=white, regular polygon, regular polygon sides=3, minimum size = 3pt, inner sep=1.4pt, pos=0.16, sloped] {}
        node[draw, fill=black, regular polygon, regular polygon sides=3, minimum size = 3pt, inner sep=1.4pt, pos=0.24, sloped] {};
    \STEdge{h2}{hh2}{0.2}
    \STEdge{h2}{hh3}{0}
    \end{scope}

    \begin{scope}[xshift=15.5cm]
    \node at (-1.5,-2*\gap) {$=$};
    \foreach \yy in {1,2} {
        \Vertex[x=0,y={-{(\yy+0.5)}*\ysh},L={G'_\yy},Lpos=180]{i\yy}
    }
    \foreach \yy in {1,2,3} {
        \Vertex[x=\gap,y=-\yy*\ysh,L={G_\yy}]{ii\yy}
    }
    \Edges(ii1,i1,ii2,i2,ii3,i1)

    \tikzset{VertexStyle/.style = {shape=regular polygon, regular polygon sides=3, fill=black, minimum size=5pt, inner sep=1.4pt, draw}}
    \Vertex[x=\gap+1.75,y=-2*\ysh+0.3,L={= T}]{t}
    \tikzset{VertexStyle/.style = {shape=regular polygon, regular polygon sides=3, fill=white, minimum size=5pt, inner sep=1.4pt, draw}}
    \Vertex[x=\gap+1.75,y=-2*\ysh-0.3,L={= T^{-1}}]{t}
    \end{scope}


\end{tikzpicture}
    \caption{Illustrating the proof of \autoref{thm:holant}, with 
    $F'_i = G'_i(T^{-1})^{\otimes n_i}$ and $F_i = T^{\otimes n_i} G_i$.}
    \label{fig:holant}
\end{figure}
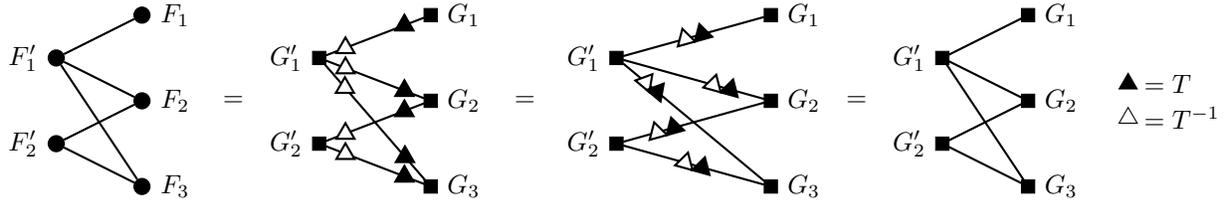

Xia \cite{xia} conjectured the converse of the Holant Theorem:
if $\fc|\fc'$ and $\gc|\gc'$ are Holant-indistinguishable, then there is
a holographic transformation between them.
Cai, Guo, and Williams \cite[Section 4.3]{cai_complete_2016}
discovered the following Boolean-domain counterexample.
\begin{example}
    \label{ex:counter2}
    Let $F' = [f_0,f_1,f_2,f_3,f_4] = [a,b,1,0,0]$, where $f_i$ is the value of $F'$ on
    inputs of Hamming weight $i$ and $a$ and $b$ are not both 0. Define $G' = [0,0,1,0,0]$ and
    $(\neq_2) = [0,1,0]$ similarly.
    Define $\fc|\fc' = (\neq_2)|F'$ and $\gc|\gc' = (\neq_2)|G'$.
    In an $(\neq_2)|F'$-grid
    $\Omega$, the $\neq_2$ signatures in the left bipartition force any nonzero edge 
    assignment $\sigma$ to assign 0 to exactly half of the edges and 1 to the other half. 
    Also, $\sigma$ must provide every $[a,b,1,0,0]$ in the right bipartition no more 1 
    than 0 inputs. If $\sigma$ provides any $[a,b,1,0,0]$ strictly fewer 1 than 0 inputs 
    (to obtain $a$
    or $b$), it must provide a different $[a,b,1,0,0]$ strictly more 1 than 0 inputs to 
    preserve the 0/1 balance, and becomes zero. Hence $(\neq_2)|F'$ is indistinguishable 
    from $(\neq_2)|G'$. However, there is no $T \in \gl_2$ transforming $(\neq_2)|F'$ to $(\neq_2)|G'$.
\end{example}
While there is no invertible matrix transforming $\fc|\fc'$ to $\gc|\gc'$ in 
\autoref{ex:counter2}, observe that
\[
    {\begin{bmatrix} \epsilon^{-1} & 0 \\ 0 & \epsilon \end{bmatrix}}^{\otimes 2} (\neq_2)
    = (\neq_2) \quad\text{and}\quad
    [a,b,1,0,0] {\begin{bmatrix} \epsilon & 0 \\ 0 & \epsilon^{-1} \end{bmatrix}}^{\otimes 4}
    = [a \epsilon^4,b \epsilon^2,1,0,0] \xrightarrow[\epsilon \to 0]{} [0,0,1,0,0],
\]
so $\left[\begin{smallmatrix} \epsilon^{-1} & 0 
\\ 0 & \epsilon\end{smallmatrix}\right] \in \gl_2$ take $\fc|\fc'$ arbitrarily close to $\gc|\gc'$ as $\epsilon
\to 0$. \autoref{thm:main1} below extends this to \emph{any} Bi-Holant-indistinguishable
$\fc$ and $\gc$: the converse of \autoref{thm:holant} holds up to orbit closure.

Cai, Guo and Williams discovered \autoref{ex:counter2} while studying
\emph{vanishing} signature sets, those sets which are Holant-indistinguishable from 0
(more precisely, the appropriate all-0 set).
Reasoning similarly to \autoref{ex:counter2}, $(\neq_2 \mid [a,b,0,0,0])$ is 
vanishing. We will see in \autoref{sec:prop} that the fact that $(\neq_2 \mid [a,b,0,0,0])$ is
vanishing explains why \autoref{ex:counter2} exists, and \autoref{thm:main2} shows
that \emph{any} counterexample $\fc$ to the converse of the Holant theorem
is due to the presence of a signature that vanishes in the context of $\fc$.


Xia proved several subcases of the converse of the Holant theorem for $\fc = \gc = \{=_2\}$,
which by \autoref{prop:shapeless} is the non-bipartite setting.
Young proved that this non-bipartite converse holds if $\kk = \rr$. In this case,
vanishing signatures do not occur (see \autoref{cor:complex} below).
By the following proposition, which follows from the fact that
$(T^{\otimes 2}A)^{1,1} = T A^{1,1} T^{\top}$ for $A \in (\kk^q)^{\otimes 2}$ (or see
\cite[Figure 2.3]{orthogonal}), the transformation must be orthogonal.
\begin{proposition}
    \label{prop:orthogonal}
    $T \in \gl_q$ is orthogonal iff $T \cdot (=_2) = (=_2)$
    for contravariant or covariant $=_2$.
\end{proposition}
\begin{theorem}[{\cite[Theorem 2.3]{orthogonal}}]
    \label{thm:orthogonal}
    Real-valued $\fc$ and $\gc$ are Holant-indistinguishable if
    and only if there is a real orthogonal $T$ such that $T\fc = \gc$.
\end{theorem}

We conclude this section with the following generalization of \autoref{thm:holant}.
\begin{proposition}
    \label{prop:respect}
    For $T \in \gl_q$ and $\fc \subset \tc$, we have $T\prop{\fc} = \prop{T\fc}$.
\end{proposition}
\begin{proof}
    Let $\vk$ be an $\fc$-gadget with signature $K$ and consider $\vk_{\fc\to T\fc}$. 
    The $T$ transformations cancel on every internal edge of $\vk_{\fc\to T\fc}$, 
    (recall \autoref{fig:holant} -- in other words, covariant/contravariant edge
    contractions are $\gl_q$-equivariant), and only survive on the dangling edges.
    Therefore $\vk_{\fc\to T\fc}$ has signature $T \cdot K$. The extension to quantum
    gadgets follows from the linearity of $T$.
\end{proof}
Specializing to $0$-ary gadgets in $\prop{\fc}$ -- that is, (quantum) Bi-Holant $\fc$-grids --
\autoref{prop:respect} says that $\fc$ and $T\fc$ are Bi-Holant indistinguishable, an
extension of \autoref{thm:holant} to Bi-Holant.

\section{The Approximate Converse}
Let $\kk = \cc$. In this section we prove our first main theorem, \autoref{thm:main1}. 
For $H \subset \gl_q$, define 
$\tc^H = \{F \in \tc \mid T \cdot F = F \text{ for every } T \in H\}$ to be
the set of tensors invariant under $H$.
The following restatement of the Tensor First Fundamental Theorem for
$\gl_q$, originally due to Weyl \cite{weyl} (see also \cite[Theorem 5.3.1]{symmetry}), says
that the only tensors invariant under all of $\gl_q$ are the signatures of quantum gadgets
composed only of wires.
\begin{theorem}[Tensor First Fundamental Theorem for $\gl_q$]
    $\tc^{\gl_q} = \prop{\varnothing}$.
    \label{thm:tensor_fft}
\end{theorem}

\begin{definition}[$\gl_q \fc$, $\overline{\gl_q \fc}$]
    \label{def:orbit_closure}
    The $\gl_q$-\emph{orbit} $\gl_q \fc$ of a finite $\fc \subset \tc$ is 
    $\{T\fc \mid T \in \gl_q\}$. If $\fc = \{F_1,\ldots,F_m\}$ with 
    $F_i \in \leftindex_{\ell_i}{\tc}_{r_i}$, then view $\fc$ as an element of the finite-dimensional
    $\cc$-vector space $V := \bigoplus_{i=1}^m \leftindex_{\ell_i}{\tc}_{r_i}$. Then $\gl_q \fc \subset V$, so define the
    $\gl_q$-\emph{orbit closure} $\overline{\gl_q \fc}$ of $\fc$ as the closure of
    $\gl_q \fc$ in the standard Euclidean topology. 
    Equivalently $\gc \in V$ is in
    $\overline{\gl_q \fc}$ if, for every $\epsilon > 0$, there is a $T_\epsilon \in \gl_q$
    such that $\|T_\epsilon \fc - \gc\| < \epsilon$ (using the standard Euclidean norm
    on $V$).
\end{definition}

\begin{definition}[{$\cc[\xc]$}]
    Let $\xc$ be a finite set of domain-$q$ variable-valued signatures. 
    For every $X \in \xc$ of arity $n$ and $\va \in [q]^n$ we introduce a variable $x_{\va}$. Define $\cc[\xc]$ to be the ring of polynomials 
    $ \cc[\{x_{\va}: X \in\xc, \va \in [q]^n\}]$.
\end{definition}
Equivalently, $\cc[\xc] \cong \cc[V]$ is the coordinate ring of the vector space $V$ from
\autoref{def:orbit_closure} (where $\xc$ is bijective with $\fc$).
For variable-valued $\xc$ and $\xc$-grid $\Omega$,
$\biholant_{\xc}(\Omega)$ is a polynomial in the entries of $\xc$.
Evaluating this polynomial at $\fc$ for $\cc$-valued $\fc$ bijective with $\xc$ 
(by substituting $F_{\va}$ for $x_{\va}$ with $\fc \ni F \leftrightsquigarrow X \in \xc$) 
yields $\holant_{\fc}(\Omega) \in \cc$. \autoref{fig:poly_example} shows an example
on the Boolean domain with
$\xc = \{X,Y\}$ for binary covariant $X$ and unary contravariant $Y$.
\begin{figure}[ht!]
    \centering
    \begin{tikzpicture}
    \GraphInit[vstyle=Classic]
    \SetUpEdge[style=-]
    \SetVertexMath
    \tikzset{VertexStyle/.style = {shape=circle, fill=black, minimum size=5pt, inner sep=1pt, draw}}

    \Vertex[x=0,y=0,L=X,Lpos=180]{x}
    \Vertex[x=1,y=0.5,L=Y,Lpos=0]{p1}
    \Vertex[x=1,y=-0.5,L=Y,Lpos=0]{p2}
    \Edge(x)(p1)
    \Edge(x)(p2)

    \node at (-2,0) {\large $\Omega:$};

\end{tikzpicture}
    \caption{
    $\holant(\Omega) = x_{00} y_{0}^2 + x_{01} y_{0} y_{1} + 
    x_{10} y_{1} y_{0} + x_{11} y_{1}^2$, 
    with the four monomials corresponding to the edge assignments $00, 01, 10, 11$,
    respectively.}
    \label{fig:poly_example}
\end{figure}
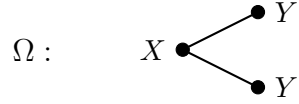

Define an action of $\gl_q$ on $\cc[\xc]$ as follows.
For $T \in \gl_q$ and $p \in \cc[\xc]$, construct $Tp \in \cc[\xc]$ 
by substituting every variable $x_{\va}$ with the $\va$-entry of $T^{-1} \cdot X$. Equivalently,
\begin{equation}
    (Tp)(\fc) = p(T^{-1} \fc)
    \label{eq:tpoly}
\end{equation}
for $\fc\subset \tc(\cc^q)$ bijective with $\xc$. Then define
\[
    \cc[\xc]^{\gl_q} := \{p \in \cc[\xc] \mid Tp = p \text{ for every }
    T \in \gl_q\}
\]
to be the set of polynomials invariant under this action.
The following theorem from geometric invariant theory, stated in
this form in \cite[Theorem 2.3]{derksen2022polystability}, \cite[Corollary 2.3.8]{derksen}, 
shows that the $\gl_q$-orbit closures of $\fc$ and $\gc$ intersect if and only if $\fc$ and
$\gc$ are indistinguishable under all $\gl_q$-invariant polynomials.
\begin{theorem}[Mumford, Fogarty, and Kirwan \cite{mumford}]
    Let $\fc,\gc \subset \tc(\cc^q)$ be bijective with $\xc$. 
    Then $\overline{\gl_q \fc} \cap \overline{\gl_q \gc} \neq \varnothing$
    if and only if $p(\fc) = p(\gc)$ for every $p \in \cc[\xc]^{\gl_q}$.
    \label{thm:mumford}
\end{theorem}
More generally, \autoref{thm:mumford} applies to any
reductive algebraic group in place of $\gl_q$ acting on any vector space $V$ over an 
algebraically closed field (although for fields other than $\cc$ we must define 
$\overline{\gl_q \fc}$ as the Zariski, instead of Euclidean, closure).
Accompanying \autoref{thm:mumford} is a result of Hilbert 
(see \cite{derksen2001polynomial}), which implies that it suffices to check 
finitely many (with the exact number depending on the arity profile of $\xc$)
polynomial invariants to ensure orbit closure intersection.
\begin{theorem}[Hilbert \cite{hilbert}]
    \label{thm:hilbert}
    The $\cc$-algebra $\cc[\xc]^{\gl_q}$ is finitely generated.
\end{theorem}

To convert the condition in \autoref{thm:mumford} from polynomial indistinguishability to
Bi-Holant indistinguishability, we apply the following minor generalization of Weyl's Polynomial 
First Fundamental Theorem for $\gl_q$ \cite{weyl,symmetry} more suited to our purpose.
The proof applies an argument similar to \cite[Theorem 4.23 and Lemma 4.26]{lysikov}.
\begin{theorem}
    \label{thm:poly_fft}
    For variable-valued signature set $\xc = \{X_1,\ldots,X_m\}$ on domain $[q]$,
    \[
        \cc[\xc]^{\gl_q} = 
        \spn\{\biholant_{\xc}(\Omega): \mathcal{X}\text{-grid } \Omega\}
    \]
\end{theorem}
\begin{proof}
    The $\supset$ inclusion follows from \eqref{eq:tpoly},
    the Bi-Holant theorem (\autoref{prop:respect}), and the fact that two polynomials which 
    take the same value on every point must be identical. 

    For the $\subset$ inclusion, let $p \in \cc[\xc]^{\gl_q}$. Split $p$ into a sum 
    $p = \sum_{d_1,\ldots,d_m \geq 0} p_{\vd}$
    of multihomogeneous polynomials, where $d_i$ is the total degree of
    the entries of $X_i$ in $p_{\vd}$ (and only finitely many
    $p_{\vd}$ are nonzero). Since the action of $\gl_q$ replaces each variable
    $x_{i,\va}$ with a linear polynomial in the entries of the same signature $X_i$,
    it preserves the multihomogeneous degree of each $p_{\vd}$. 
    Therefore each $p_{\vd} \in 
    \cc[\xc]^{\gl_q}$, and it suffices to find an $\Omega$ such that
    $\biholant_{\xc}(\Omega) = p_{\vd}$. Let $X_i$ have left-arity $\ell_i$ and right-arity
    $r_i$. Each $p_{\vd}$ is a linear functional on the space
    \[
        \bigotimes_{i=1}^m \sym^{d_i} \big(\leftindex_{\ell_i}{\tc}_{r_i}\big) =
        \bigotimes_{i=1}^m \sym^{d_i}\big((\cc^q)^{\otimes \ell_i}
        \otimes ((\cc^q)^*)^{\otimes r_i} \big)
    \]
    (where $\sym^n(V)$ denotes the space of symmetric tensors in $V^{\otimes n}$) so
    we can, after normalizing by $\left(\prod_{i} d_i! \right)^{-1}$,
    identify $p_{\vd}$ with a tensor
    \begin{equation}
        A_{\vd} \in \bigotimes_{i=1}^m \sym^{d_i}\big(((\cc^q)^*)^{\otimes \ell_i}
        \otimes (\cc^q)^{\otimes r_i} \big)
        \subset ((\cc^q)^*)^{\otimes \sum_i \ell_i d_i} \otimes 
        (\cc^q)^{\otimes \sum_i r_i d_i}.
        \label{eq:factor_reorder}
    \end{equation}
    For example,
    if $q=4$, $\xc = \{X,Y,Z\}$, $(\ell_1,\ell_2,\ell_3) = (0,3,1)$, $(r_1,r_2,r_3) = (2,0,1)$, and
    $p_{1,2,1} = x_{34} y_{013} y_{444} z_{23} 
    = x_{34} y_{444} y_{013} z_{23}$, then
    \begin{align*}
        A_{1,2,1} = \frac{1}{2} 
        &\big((e_3 \otimes e_4) \otimes (e_0^* \otimes e_1^* \otimes e_3^*) \otimes
            (e_4^* \otimes e_4^* \otimes e_4^*) \otimes (e_2^* \otimes e_3) \\
        & + (e_3 \otimes e_4) \otimes (e_4^* \otimes e_4^* \otimes e_4^*) \otimes
             (e_0^* \otimes e_1^* \otimes e_3^*) \otimes (e_2^* \otimes e_3)
        \big).
    \end{align*}
    Now, viewing $\bigotimes_i (X_i)^{\otimes d_i}$
    as a signature with left arity $\sum_i \ell_i d_i$ and right arity 
    $\sum_i r_i d_i$, reconstruct
    \begin{equation}
        p_{\vd} = \Big\langle A_{\vd},~
        \bigotimes_i X_i^{\otimes d_i} \Big\rangle.
        \label{eq:reconstruct}
    \end{equation}
    Furthermore, for any $T \in \gl_q$,
    \[
        Tp_{\vd} = p_{\vd}(T^{-1}\xc) 
        = \Big\langle A_{\vd},~
        \bigotimes_i (T^{-1} \cdot X_i)^{\otimes d_i} \Big\rangle
        = \Big\langle T \cdot A_{\vd},~
            \bigotimes_i (X_i)^{\otimes d_i} \Big\rangle.
    \]
    so the map $p_{\vd} \mapsto A_{\vd}$ is $\gl_q$-equivariant.
    With $p_{\vd} \in \cc[\xc]^{\gl_q}$, it follows that     
    $A_{\vd} \in \tc(\cc^q)^{\gl_q}$ 
    (up to the reordering of factors in \eqref{eq:factor_reorder}, which doesn't affect 
    this invariance), so, by 
    \autoref{thm:tensor_fft}, $A_{\vd} \in \prop{\varnothing}$ is the signature of
    a wire gadget. Now \eqref{eq:reconstruct} says that $p_{\vd}$ is a full
    contraction consisting only of wires and signatures in $\xc$, which is
    $\biholant_{\xc}(\Omega)$ for some $\xc$-grid $\Omega$.
\end{proof}

\begin{theorem}[first main theorem]
    Finite $\fc,\gc \subset \tc(\cc^q)$ are Bi-Holant-indistinguishable if and only if 
    $\overline{\gl_q \fc} \cap \overline{\gl_q \gc} \neq \varnothing$.
    \label{thm:main1}
\end{theorem}
\begin{proof}
    The $(\Rightarrow)$ direction follows from \autoref{thm:mumford} and 
    \autoref{thm:poly_fft}. $(\Leftarrow)$ follows from the Bi-Holant Theorem
    (\autoref{prop:respect}) and the fact that $\biholant_{\fc}(\Omega)$ is
    a polynomial, hence continuous,
    function in $\fc$.
\end{proof}

Combining \autoref{thm:main1}, \autoref{thm:poly_fft}, and \autoref{thm:hilbert} shows that 
\begin{corollary}
    \label{cor:decidable}
    The problem of determining whether any two finite $\fc,\gc \subset \tc(\cc^q)$ are
    Bi-Holant-indistinguishable is decidable.
\end{corollary}
There are algorithms for computing the finite generating set of $\cc[\xc]^{\gl_q}$
guaranteed by \autoref{thm:hilbert} \cite{derksen, derksen1999computation}, and there are
upper bounds on the largest degree of any such generator \cite{derksen2001polynomial}. 
However, in general these upper bounds are exponential in the size of $\xc$ (i.e. the size
of the signature sets in question) and in certain cases there are exponential lower bounds
-- see e.g. \cite[Proposition 4.15]{minimal_canonical}.

Say $\fc \subset \tc(\cc^q)$ is \emph{Bi-Holant-vanishing}
if it is Bi-Holant-indistinguishable from the set of all-0
signatures. By \autoref{prop:shapeless}, this notion captures both
bipartite and general Holant vanishing.
\begin{corollary}
    \label{cor:vanishing}
    Finite $\fc \subset \tc(\cc^q)$ is Bi-Holant-vanishing if and only if
    $0 \in \overline{\gl_q \fc}$.
\end{corollary}

\section{Quantum-Nonvanishing Wheeled PROPs}
Given $(\ell_1,r_1)$-gadget $\vk_1$ and $(\ell_2,r_2)$-gadget $\vk_2$, construct a
$(\ell_1+\ell_2,r_1+r_2)$-gadget $\vk_1 \otimes \vk_2$ as the disjoint union of $\vk_1$
and $\vk_2$, placing $\vk_1$ above $\vk_2$ (so all $\vk_1$ dangling edges precede all
$\vk_2$ dangling edges in the left and right order). This operation extends bilinearly
to quantum gadgets and induces the tensor product on the underlying signatures.

\label{sec:prop}
\begin{definition}[{\cite[Definition 2.1]{prop}}]
    \label{def:prop}
    A \emph{pre-wheeled PROP} is a bigraded $\kk$-vector space 
    $\mathfrak{R} = \bigoplus_{\ell,r \geq 0} \rprop{\ell}{r}$  together with
    \begin{itemize}
        \item a special element $1_{\mathfrak{R}} \in \rprop{0}{0}$,
        \item a special element $I_{\mathfrak{R}} \in \rprop{1}{1}$,
        \item a bilinear map $\otimes: \rprop{\ell_1}{r_1} \times \rprop{\ell_2}{r_2}
            \to \rprop{\ell_1+\ell_2}{r_1+r_2}$, and
        \item a linear map $\leftindex_i{\partial}_j: \rprop{\ell}{r} \to \rprop{\ell-1}{r-1}$
            for every $1 \leq i \leq \ell$ and $1 \leq j \leq r$.
    \end{itemize}
\end{definition}
The mixed tensor algebra $\tc$ is a pre-wheeled PROP, where 
$1_{\tc} = 1_{\kk}$, $I_{\tc} = I$ (the identity map), 
$\otimes$ is the usual tensor
product, and $\leftindex_i{\partial}_j$ contracts the $i$th contravariant input with the $j$th
covariant input. For any $\fc$, the space $\mathfrak{Q}_{\fc}$ of quantum-$\fc$-gadgets
(the formal direct sums of the diagrams
themselves) is also a pre-wheeled PROP, where $\qprop{\ell}{r}$ is the space of
$(\ell,r)$-quantum-$\fc$-gadgets, $1_{\mathfrak{Q}_{\fc}}$ is the empty gadget, 
$I_{\mathfrak{Q}_{\fc}}$
is the wire gadget, $\otimes$ is gadget tensor product, and $\leftindex_i{\partial}_j$ is
the operation of connecting the $i$th left input and $j$th right input.
In fact, 
$\mathfrak{Q}_{\fc}$ is (isomorphic to) the \emph{free wheeled PROP} generated by $\fc$
\cite[Definition 2.16]{prop}. A \emph{wheeled PROP} is a pre-wheeled PROP which is the image
of a free wheeled PROP under a
pre-wheeled PROP homomorphism (a linear map respecting the bigrading and the four
elements/operations listed in \autoref{def:prop}) \cite[Definition 2.17]{prop}. 
Therefore $\prop{\fc} \subset \tc$ is a wheeled PROP, as
it is the image of the free wheeled PROP $\mathfrak{Q}_{\fc}$ under the pre-wheeled
PROP homomorphism mapping a quantum-$\fc$-gadget to its signature.
Specifically, $\prop{\fc}$ is a sub-wheeled PROP of $\tc$ (which is the image of
the free wheeled PROP $\mathfrak{Q}_{\tc}$ under the same signature-evaluation map),
and every sub-wheeled PROP of $\tc$ is $\prop{\fc}$ for some $\fc \subset \tc$.

\begin{definition}[$\fc$-nonvanishing, Quantum-nonvanishing]
    Say $K \in \dprop{\fc}{\ell}{r}$ is $\fc$-\emph{nonvanishing} if
    it satisfies any of the following equivalent conditions.
    \begin{enumerate}[label=(\arabic*)]
        \item There is a $\widehat{K} \in \dprop{\fc}{r}{\ell}$ such that 
        $\langle K,\widehat{K} \rangle \neq 0$, or
        \item there is an $\prop{\fc}$-grid $\Omega$ containing $K$ such that 
            $\holant(\Omega) \neq 0$, or
        \item there is an $\fc \cup \{K\}$-grid $\Omega$ containing $K$ such that
            $\holant(\Omega) \neq 0$.
    \end{enumerate}

    Then say $\fc \subset \tc$ is $(\ell,r)$-\emph{quantum-nonvanishing} if every 
    nonzero $K \in \dprop{\fc}{\ell}{r}$ is $\fc$-nonvanishing
    (equivalently, the bilinear form
    $\langle\cdot,\cdot\rangle$ is nondegenerate on $\dprop{\fc}{\ell}{r}$), and $\fc$ is
    \emph{quantum-nonvanishing} if it is $(\ell,r)$-quantum-nonvanishing for every $(\ell,r)$.
\end{definition}
\begin{proof}
$(1) \implies (2)$ because $\langle K,\widehat{K} \rangle$ is the Holant value of an
$\prop{\fc}$-grid containing $K$, and $(2) \implies (1)$ because, given $\Omega$, let
$\widehat{K}$ be the signature of the $\prop{\fc}$-gadget formed by removing a vertex assigned
$K$ from $\Omega$, leaving its formerly incident edges dangling.
$(3) \implies (2)$ because every $\fc \cup \{K\}$-grid is an $\prop{\fc}$-grid, and
$(2) \implies (3)$ because
expanding as quantum-$\fc$-gadgets the other signatures in the $\prop{\fc}$ grid $\Omega$
containing $K$ yields a quantum $\fc \cup \{K\}$-grid with each term containing $K$,
at least one of which has nonzero Holant value.
\end{proof}

The following theorem of
Derksen and Makam states that, if $\fc$ is quantum-nonvanishing, then there is a subgroup
$\stab(\prop{\fc}) \subset \gl_q$ such that every tensor in $\tc$ invariant under
the action of $\stab(\prop{\fc})$ is realizable as a quantum-$\fc$-gadget signature
(Derksen and Makam use the term ``simple'' instead of ``quantum-nonvanishing''). The theorem generalizes the theorem of Schrijver \cite{schrijver_tensor_2008,regts_rank_2012} used
to prove \autoref{thm:orthogonal}, and is the same type of result
(in the sense of characterizing quantum gadget signatures as invariant tensors) as the
Tannaka-Krien duality used by Man\v{c}inska and Roberson \cite{planar} and Cai and Young
\cite{cai_planar_2023} to prove their results on planar indistinguishability and quantum
isomorphism.
\begin{theorem}[{\cite[Theorem 6.2, Proposition 6.5, Corollary 6.6]{prop}}]
    A signature set $\fc$ is quantum-nonvanishing 
    if and only if $\prop{\fc} = \tc^{\stab(\prop{\fc})}$ for some 
    reductive subgroup $\stab(\prop{\fc}) \subset \gl_q$.
    Furthermore, if these conditions hold, then $\prop{\fc}$ is finitely generated.
    \label{thm:duality}
\end{theorem}

In \autoref{sec:bipartite}, we use \autoref{thm:duality} to prove the following main theorem.
\begin{theorem}[second main theorem]
    \label{thm:main2}
    If $\fc|\fc'$ and $\gc|\gc'$ are quantum-nonvanishing, then $\fc|\fc'$ and $\gc|\gc'$ are
    Holant-indistinguishable if and only if there
    is a $T \in \gl_q$ such that $T(\fc|\fc') = \gc|\gc'$.
\end{theorem}

\autoref{thm:main2} implies that any $\fc|\fc'$ and $\gc|\gc'$ serving as a counterexample
to the converse of the Holant theorem cannot both be quantum-nonvanishing. 
In \autoref{ex:counter2}, $\fc|\fc'$ is quantum-vanishing. To see this, consider the quantum
$\fc|\fc'$-gadget $4 \vk_1 - \vk_2$ shown in \autoref{fig:counter}. 
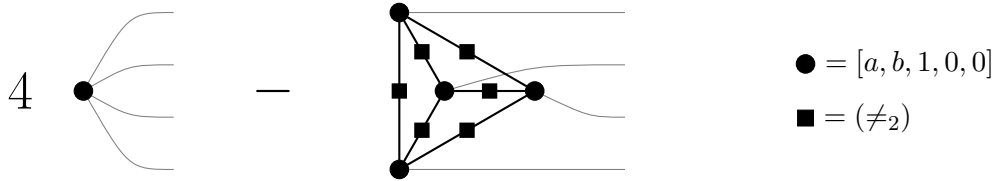
\begin{figure}[ht!]
    \centering
    \begin{tikzpicture}[scale=1.2]
    \GraphInit[vstyle=Classic]
    \SetUpEdge[style=-]
    \SetVertexMath
    \tikzset{VertexStyle/.style = {shape=circle, fill=black, minimum size=7pt, inner sep=1pt, draw}}

    \def\wlen{1}
    \def\ygap{0.87}

    \foreach \ii in {-\ygap,-\ygap/3,\ygap/3,\ygap} {
        \draw[thin, color=gray] (0,0) .. controls +(\wlen/2,\ii) .. (\wlen,\ii);
    }
    \Vertex[x=0,y=0,NoLabel]{a}

    \begin{scope}[xshift=4cm]
        \draw[thin, color=gray] (-0.5,\ygap) -- (1+\wlen, \ygap);
        \draw[thin, color=gray] (-0.5,-\ygap) -- (1+\wlen, -\ygap);
        \draw[thin, color=gray] (0,0) .. controls +(1,\ygap/3) .. (1+\wlen, \ygap/3);
        \draw[thin, color=gray] (1,0) .. controls +(0.6,-\ygap/3) .. (1+\wlen, -\ygap/3);
        \Vertex[x=0,y=0,NoLabel]{c}
        \Vertex[x=-0.5,y=\ygap,NoLabel]{u}
        \Vertex[x=-0.5,y=-\ygap,NoLabel]{v}
        \Vertex[x=1,y=0,NoLabel]{w}
        \foreach \uu/\vv in {c/w,c/u,c/v,u/v,u/w,v/w} {
            \draw[thick] (\uu) -- (\vv) node[draw, rectangle, fill=black, minimum size = 5pt, inner sep = 1pt, pos=0.5] {};
        }
    \end{scope}

    \node at (-0.7,0) {\huge 4};
    \node at (2.1,0) {\huge $-$};

    \Vertex[x=8,y=0.3,L={= [a,b,1,0,0]}]{ax8}
    \tikzset{VertexStyle/.style = {shape=rectangle, fill=black, minimum size=6pt, inner sep=1pt, draw}}
    \Vertex[x=8,y=-0.3,L={= (\neq_2)}]{ax7}

\end{tikzpicture}
    \caption{A quantum gadget $4 \vk_1 - \vk_2$ with signature $K = [4a-p_1(a,b), 4b-p_2(b),0,0,0]$}
    \label{fig:counter}
\end{figure}
Reasoning as in \autoref{ex:counter2}, the symmetric gadget $\vk_2$ has
signature $[p_1(a,b),p_2(b),4,0,0]$ for polynomials $p_1$ in $a$ and $b$ and $p_2$ in $b$,
so the signature of $4 \vk_1 - \vk_2$ is $K := [4a-p_1(a,b), 4b-p_2(b),0,0,0] \in 
\dprop{\fc|\fc'}{0}{4}$. But, in any $(\neq_2) \mid F',K$-grid $\Omega$ 
containing $K$,
every nonzero assignment is forced to assign $K$ strictly fewer 1s than 0s, so must assign
strictly more 1s than 0s to another $[a,b,1,0,0]$ or $K$, which then evaluates to 0. 
Therefore $K$ is $\fc|\fc'$-vanishing (if $a,b$ are such that $K=0$,
\autoref{thm:main2} asserts that some nonzero quantum gadget must be $\fc|\fc'$-vanishing).

Observe that the $\fc|\fc'$-vanishing $K$ corresponds to $0 \in \prop{\gc|\gc'}$. 
This motivates the following.
\begin{definition}
    \label{def:covanishing}
    $\fc$ and $\gc$ are \emph{covanishing} if, for every 
    $\prop{\fc} \ni F \leftrightsquigarrow G \in \prop{\gc}$, $F = 0 \iff G = 0$.
\end{definition}
By \autoref{prop:respect}, if $\fc$ and $\gc$ are not covanishing, then there is no
$T \in \gl_q$ transforming $\fc$ to $\gc$ (such a $T$ would map a nonzero signature
to 0), giving an alternate explanation for \autoref{ex:counter2}.

The covanishing property generalizes indistinguishability in the following sense.
\begin{proposition}
    $\fc$ and $\gc$ are $(0,0)$-covanishing iff $\fc$ and $\gc$ are
    Bi-Holant-indistinguishable.
    \label{prop:zerozero}
\end{proposition}
\begin{proof}
    We have $\dprop{\fc}{0}{0} = \spn\{\biholant_{\fc}(\Omega): \fc\text{-grid } \Omega\} \subset
    \kk$. So if $\fc$ and $\gc$ are indistinguishable, then every scalar in 
    $\dprop{\fc}{0}{0}$ equals the corresponding scalar in $\dprop{\gc}{0}{0}$, 
    hence $\fc$ and $\gc$ are $(0,0)$-covanishing. Conversely,
    suppose there is an $\fc$-grid $\Omega$ such that $\biholant_{\fc}(\Omega) \neq
    \biholant_{\gc}(\Omega_{\fc\to\gc})$. In both $\prop{\fc}$ and $\prop{\gc}$, the vertexless
    loop $\bigcirc$ has Holant value $q \in \kk$. Therefore
    \[
        0 = \biholant_{\fc}\left(\Omega - \frac{\biholant_{\fc}(\Omega)}{q} \cdot \bigcirc\right)
        \leftrightsquigarrow
        \biholant_{\gc}\left(\Omega_{\fc\to\gc} - \frac{\biholant_{\fc}(\Omega)}{q} \cdot \bigcirc\right)
        \neq 0,
    \]
    so $\fc$ and $\gc$ are not $(0,0)$-covanishing.
\end{proof}
\begin{proposition}
    If $\fc$ and $\gc$ are Bi-Holant-indistinguishable and $(\ell,r)$-quantum-nonvanishing,
    then $\fc$ and $\gc$ are $(\ell,r)$-covanishing.
    \label{prop:covanishing}
\end{proposition}
\begin{proof}
    Assume $\fc$ and $\gc$ are not $(\ell,r)$-covanishing, so WLOG there is a
    $\dprop{\fc}{\ell}{r} \ni K \leftrightsquigarrow 0 \in \dprop{\gc}{\ell}{r}$
    with $K \neq 0$. By indistinguishability,
    every $\fc\cup\{K\}$-grid $\Omega$ containing $K$ satisfies 
    $\biholant(\Omega)
    = \biholant(\Omega_{\fc\cup\{K\}\to\gc\cup\{0\}}) = 0$ because 
    $\Omega_{\fc\cup\{K\}\to\gc\cup\{0\}}$ contains 0. 
    Therefore $K$ is $\fc$-vanishing, so $\fc$ is $(\ell,r)$-quantum-vanishing.
\end{proof}


\section{The Conditional Converse}
\label{sec:conditional}
In this section, we prove our second main theorem \autoref{thm:main2},
as well as \autoref{thm:binary}, a similar result for sets of matrices 
$\fc,\gc \subset \kk^q \otimes (\kk^q)^*$. Both proofs make heavy use of the
\emph{subdomain restriction} constructions of the following definition.
\begin{definition}[$F|_X$, $\prop{\fc}_X$]
    For $F \in \tc(\kk^q)$ and $X \subset [q]$, define $F|_X \in \tc(\kk^X)$ to be the
    subsignature of $F$ induced by $X$. For $\fc \subset \tc(\kk^q)$, define
    $\prop{\fc}_X := \{F|_X: F \in \prop{\fc}\} \subset \tc(\kk^X)$, a set on domain 
    $X$ bijective with $\prop{\fc}$.
\end{definition}
Note that, while $\prop{\fc}_X \subset \prop{\prop{\fc}_X}$,
we may not have $\prop{\fc}_X \supset \prop{\prop{\fc}_X}$ (unless $I_X^\shortuparrow
\in \prop{\fc}$ -- see \autoref{prop:uparrow}). For example, if $\prop{\fc}$ contains the 
$(X,\overline{X})$-block matrix $\left[\begin{smallmatrix} A & B \\ C & D \end{smallmatrix}
\right]$, then $A \in \prop{\fc}_X$, so $A^2 \in \prop{\prop{\fc}_X}$, but we may not be able to obtain
$A^2$ as the $X$-block of a matrix in $\prop{\fc}$.

\begin{definition}[$(\cdot)^{\shortuparrow Z}$]
    \label{def:shortuparrowow}
    Let $X \subset Z$ and $F \in \tc(\kk^X)$. Define $F^{\shortuparrow Z} \in \tc(\kk^Z)$ by
    \[
        F^{\shortuparrow Z}(\vx) = \begin{cases} F(\vx) & \vx \in X^n \\ 0 & \text{otherwise}
        \end{cases}
        \text{ for } \vx \in Z^n.
    \]
    That is, $F^{\shortuparrow Z}$ expands the domain of $F$ to all of $Z$ by padding with zeros. 
    We frequently write simply $F^\shortuparrow$ when the ambient domain $Z$ is clear from context.
\end{definition}

The next three results show the utility of realizing the \emph{subdomain restrictor} 
$I_X^\shortuparrow = \left[\begin{smallmatrix} I_X & 0 \\ 0 & 0 \end{smallmatrix}\right]$,
which acts like an edge ($I$) on inputs from $X$ and zeroes out the other subdomains.
\begin{proposition}
    \label{prop:uparrow}
    Suppose $\prop{\fc} \ni I_X^\shortuparrow \leftrightsquigarrow I_X^\shortuparrow \in \prop{\gc}$.
    Then, for any $\prop{\prop{\fc}_X} \ni F \leftrightsquigarrow G \in \prop{\prop{\gc}_X}$, we have
    $\prop{\fc} \ni F^\shortuparrow \leftrightsquigarrow G^\shortuparrow \in \prop{\gc}$.
    Therefore $\prop{\fc}_X = \prop{\prop{\fc}_X}$.
\end{proposition}
\begin{proof}
    By definition, $F$ is the signature of a quantum-$\prop{\fc}_X$-gadget $\vk$ and $G$ is
    the signature of $\vk_{\prop{\fc}_X\to\prop{\gc}_X}$.
    Construct a quantum-$\prop{\fc}$-gadget $\vk^\shortuparrow$ as
    follows. Start with $\vk_{\prop{\fc}_X\to\prop{\fc}}$, constructed by
    replacing each $S|_X \in \prop{\fc}_X$ in $\vk$ with the corresponding $S \in \prop{\fc}$.
    Then replace each dangling and internal edge
    -- which when viewed alone is a $(1,1)$ wire gadget with signature $I$ --
    with $I_X^\shortuparrow \in \dprop{\fc}{1}{1}$. This has the effect of forcing all edges
    in $\vk^\shortuparrow$, including dangling edges, to take values in $X$, so 
    the signature of $\vk^\shortuparrow$ is $F^\shortuparrow$. Similarly,
    the signature of $(\vk_{\prop{\fc}_X\to\prop{\gc}_X})^\shortuparrow$ is 
    $G^\shortuparrow$, and, since
    $(\vk_{\prop{\fc}_X\to\prop{\gc}_X})^\shortuparrow = 
    (\vk^\shortuparrow)_{\prop{\fc}\to\prop{\gc}}$, we have
    $F^\shortuparrow \leftrightsquigarrow G^\shortuparrow$. See \autoref{fig:uparrow_insert}.

    The second claim follows from the first and the fact that $F^\shortuparrow|_X = F$.
\end{proof}

\begin{figure}[ht!]
    \centering
    \begin{tikzpicture}[scale=1]
    \GraphInit[vstyle=Classic]
    \SetUpEdge[style=-]
    \SetVertexMath

    \def\xx{1.2}
    \def\xsh{6}
    \def\wlen{0.8}
    \def\wwlen{1.5}

    \node at (\xx,1.8) {$\vk$};
    \node at (\xx+\xsh,1.8) {$\vk^\shortuparrow$};

    \draw[thin, color=gray] (0,0) -- (-\wlen,0);
    \draw[thin, color=gray] (0,1) -- (-\wlen,1);
    \draw[thin, color=gray] (2*\xx,1) -- (2*\xx+\wlen,1);

    \tikzset{VertexStyle/.style = {shape=circle, minimum size=7pt, inner sep=0pt,fill = black, draw}}
    \Vertex[x=0,y=1,L={(S_1)|_X},Lpos=90]{sx1}
    \Vertex[x=2*\xx,y=1,L={(S_2)|_X},Lpos=90]{sx2}
    \Vertex[x=0,y=0,L={(S_3)|_X},Lpos=270]{sx3}
    \Vertex[x=\xx,y=0,L={(S_4)|_X},Lpos=270]{sx4}
    \Vertex[x=2*\xx,y=0,L={(S_5)|_X},Lpos=270]{sx5}

    \Edges(sx1,sx2,sx4,sx3,sx1,sx4,sx5)

    \begin{scope}[xshift=\xsh*28.5]
        \draw[thin, color=gray] (0,0) -- (-\wwlen,0);
        \draw[thin, color=gray] (0,1) -- (-\wwlen,1);
        \draw[thin, color=gray] (2*\xx,1) -- (2*\xx+\wwlen,1);

        \Vertex[x=0,y=1,L={S_1},Lpos=90]{s1}
        \Vertex[x=2*\xx,y=1,L={S_2},Lpos=90]{s2}
        \Vertex[x=0,y=0,L={S_3},Lpos=270]{s3}
        \Vertex[x=\xx,y=0,L={S_4},Lpos=270]{s4}
        \Vertex[x=2*\xx,y=0,L={S_5},Lpos=270]{s5}

        \tikzset{VertexStyle/.style = {shape=rectangle, fill=black, minimum size=5pt, inner sep=1pt, draw}}
        \Vertex[x=\xx,y=1,NoLabel]{ax1S}
        \Vertex[x=0,y=0.5,NoLabel]{ax2S}
        \Vertex[x=\xx/2,y=0.5,NoLabel]{ax3S}
        \Vertex[x=3*\xx/2,y=0.5,NoLabel]{ax4S}
        \Vertex[x=\xx/2,y=0,NoLabel]{ax5S}
        \Vertex[x=3*\xx/2,y=0,NoLabel]{ax6S}
        \Vertex[x=-\wwlen/2,y=1,NoLabel]{d1}
        \Vertex[x=2*\xx+\wwlen/2,y=1,NoLabel]{d2}
        \Vertex[x=-\wwlen/2,y=0,NoLabel]{d3}

        \Vertex[x=4*\xx,y=0.5,L={=I_X^\shortuparrow}]{ax7S}
    \end{scope}

    \Edges(s1,s2,s4,s3,s1,s4,s5)
    \Edge(d1)(s1)
    \Edge(d2)(s2)
    \Edge(d3)(s3)
\end{tikzpicture}
    \caption{The construction in Propositions \ref{prop:uparrow} and 
    \ref{prop:subdomain_indist}.}
    \label{fig:uparrow_insert}
\end{figure}
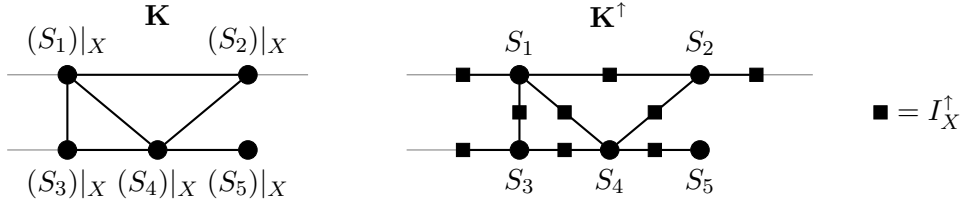

\begin{proposition}[{Bi-Holant version of \cite[Lemma 4.2]{orthogonal}}]
    \label{prop:subdomain_indist}
    If $\fc$ and $\gc$ are Bi-Holant-indistinguishable and $\prop{\fc} \ni I_X^\shortuparrow 
    \leftrightsquigarrow I_X^\shortuparrow \in
    \prop{\gc}$, then $\prop{\fc}_X$ and $\prop{\gc}_X$ are Bi-Holant-indistinguishable.
\end{proposition}
\begin{proof}
    Let $\Omega$ be an $\prop{\fc}_X$-grid. Viewing $\Omega$ as a $(0,0)$-$\prop{\fc}_X$-gadget, construct
    the quantum $\fc$-grid $\Omega^\shortuparrow$ as in the proof of \autoref{prop:uparrow}.
    Applying similar reasoning, we obtain
    \[
        \biholant_{\prop{\fc}_X}(\Omega) = \biholant_{\prop{\fc}}\big(\Omega^\shortuparrow\big)
        = \biholant_{\prop{\gc}}\big(\Omega^\shortuparrow_{\prop{\fc}\to\prop{\gc}}\big)
        = \biholant_{\prop{\gc}_X}(\Omega_{\prop{\fc}_X\to\prop{\gc}_X}). \qedhere
    \]
\end{proof}

\begin{proposition}
    \label{prop:subdomain_nondegen}
    If $\fc$ is $(\ell,r)$-quantum-nonvanishing and $I_X^\shortuparrow \in \prop{\fc}$,
    then $\prop{\fc}_X$ is $(\ell,r)$-quantum-nonvanishing.
\end{proposition}
\begin{proof}
    Let $F \in \dprop{\prop{\fc}_X}{\ell}{r}$. By \autoref{prop:uparrow}, $F^\shortuparrow \in 
    \dprop{\fc}{\ell}{r}$.
    Since $\fc$ is $(\ell,r)$-quantum-nonvanishing, there is a $\widehat{F^{\shortuparrow}} \in 
    \dprop{\fc}{r}{\ell}$ 
    such that $\langle F^{\shortuparrow}, \widehat{F^{\shortuparrow}} \rangle \neq 0$. 
    But $F^{\shortuparrow}$ is only supported on $X$, so
    \[
        0 \neq \langle F^{\shortuparrow}, \widehat{F^{\shortuparrow}} \rangle
        = \langle F^{\shortuparrow}|_X, \widehat{F^{\shortuparrow}}|_X \rangle
        = \langle F, \widehat{F^{\shortuparrow}}|_X \rangle.
    \]
    Thus $\widehat{F^{\shortuparrow}}|_X \in \prop{\fc}_X$ witnesses that $F$ is
    $\prop{\fc}_X$-nonvanishing, so $\prop{\fc}_X$ is $(\ell,r)$-quantum-nonvanishing.
\end{proof}

\subsection{Simultaneous Similarity}
\label{sec:binary}
In this subsection, we consider $\fc \subset \kk^q \otimes (\kk^q)^*$, a set of
\emph{mixed binary} signatures with one left and one right input. Thinking of $\fc$
as generators of a wheeled PROP, we always assume $I \in \fc$. We also view
$\fc$ as a set of matrices in $\kk^{q \times q}$, and for $T \in \gl_q$,
$T \fc = \{TFT^{-1} \mid F \in \fc\}$ is simultaneous conjugation of the
matrices in $\fc$ by $T$. 

\begin{definition}[$\Gamma_{\fc}$]
    Let $\Gamma_{\fc}$ be the set of all finite products of matrices in $\fc$.
\end{definition}
Every Bi-Holant $\fc$-grid is a disjoint union of cycles, each of which defines a 
word $w \in \Gamma_{\fc}$ and has value $\tr(w)$. Note that bipartiteness prevents 
transposing matrices in $\fc$ when constructing $w$ (this would require connecting two left
or two right edges), as is allowed in non-bipartite Holant for a set of binary signatures.
If transpose is allowed and $\kk = \rr$ or $\kk = \cc$, then the 
indistinguishability relation is always simultaneous similarity by an 
real or complex orthogonal matrix, respectively 
\cite[Corollary 2.3, Theorem 2.4]{jing_unitary_2015}, \cite[Corollary 5.4]{orthogonal}. 
Instead, the only conclusion we can immediately draw from indistinguishability in the 
bipartite setting is that every $\fc \ni F \leftrightsquigarrow G \in \gc$ have the same
multiset of eigenvalues, as, by a standard argument using Newton's identities for symmetric
polynomials, this is 
equivalent to $\tr(F^k) = \tr(G^k)$ for every $k \geq 0$. Two matrices
are similar if and only if they have the same Jordan normal form, so any $\fc = \{F\}$ and 
$\gc = \{G\}$ for $F$ and $G$ with identical spectrum but different Jordan normal forms
provide a counterexample to the converse of the Bi-Holant theorem.
If $F$ is not diagonalizable, then put $F$ in Jordan normal
form and write $F = \widetilde{F} + N$ for diagonal $\widetilde{F}$ and nilpotent $N$. 
We first make the following well-known observation.
Since they have the same multiset of eigenvalues, $F$ and $\widetilde{F}$ are Bi-Holant 
indistinguishable.
Therefore, if $\kk = \cc$, their $\gl_q$-orbit closures intersect by \autoref{thm:main1}. 
Indeed, the
invertible matrices $\text{diag}(\epsilon^q,\epsilon,\ldots,1)$ transform $F$ arbitrarily 
close to $\widetilde{F}$ as $\epsilon \to 0$ -- e.g.
\[
    \begin{bmatrix} \epsilon^2 & 0 & 0 \\ 0 & \epsilon & 0 \\ 0 & 0 & 1 \end{bmatrix}
    \begin{bmatrix} \lambda & 1 & 0 \\ 0 & \lambda & 1 \\ 0 & 0 & \lambda \end{bmatrix}
    \begin{bmatrix} \epsilon^{-2} & 0 & 0 \\ 0 & \epsilon^{-1} & 0 \\ 0 & 0 & 1 \end{bmatrix}
    = \begin{bmatrix} \lambda & \epsilon & 0 \\ 0 & \lambda & \epsilon \\ 0 & 0 & \lambda \end{bmatrix}
    \xrightarrow[\epsilon \to 0]{}
    \begin{bmatrix} \lambda & 0 & 0 \\ 0 & \lambda & 0 \\ 0 & 0 & \lambda \end{bmatrix}.
\]
Second, the minimal polynomial $p$ of $\widetilde{F}$ divides but does not
equal the minimal polynomial of $F$, so $p(F) \neq 0 = p(\widetilde{F})$. Since
$\prop{F} \ni p(F) \leftrightsquigarrow p(\widetilde{F}) = 0 \in \langle \widetilde{F} \rangle$,
it follows from indistinguishability (as in the
proof of \autoref{prop:covanishing}) that $p(F)$ is $\{F\}$-vanishing, so $\{F\}$ is
$(1,1)$-quantum-vanishing. 
Thus any $(1,1)$-quantum-nonvanishing $\{F\}$ and $\{G\}$ are diagonalizable, so, for any such pair,
indistinguishability does imply similarity.
\autoref{thm:binary} below generalizes this statement to simultaneous similarity.
Note that $(1,1)$-quantum-nonvanishing does not necessarily imply
full quantum-nonvanishing at all arities.
So, instead of \autoref{thm:duality}, our proof uses the following
theorem of Kaplansky (see also \cite[Theorem 2.1]{kaplansky_revisited}).
Say that $F \in \kk^{q \times q}$ has \emph{singleton spectrum} if $F$ has (up to 
multiplicity) only one distinct eigenvalue.
\begin{theorem}[Kaplansky \cite{kaplansky}]
    Suppose $\ac \subset \kk^{q \times q}$ is closed under matrix product and
    every $A \in \ac$ has singleton spectrum. 
    Then $\ac$ is simultaneously triangularizable under some $T \in \gl_q$.
\end{theorem}
Say $\fc \subset \tc(\kk^q)$ is $(1,1)$-\emph{trivial} if $\dprop{\fc}{1}{1} \subset \spn(I)$ 
(i.e. is as small as possible, as the wire gadget is always present).
\begin{corollary}
    Let $\fc \subset \kk^{q \times q}$ be $(1,1)$-quantum-nonvanishing. If every 
    $F \in \dprop{\fc}{1}{1}$ has singleton spectrum, then $\fc$ is $(1,1)$-trivial.
    \label{cor:kaplansky}
\end{corollary}
\begin{proof}
    Applying Kaplansky's theorem to $\dprop{\fc}{1}{1}$, 
    we may transform $\fc$ so that every matrix in $\dprop{\fc}{1}{1}$ is
    upper triangular, with constant diagonal. This does not change whether $\fc$ is 
    quantum-vanishing. Suppose $\fc \ni F \not\in \spn(I)$, with constant $\lambda$ on the 
    diagonal. Then $F - \lambda I\in \prop{\fc}$ is nonzero and strictly upper triangular, 
    so $(F-\lambda I)F'$ is strictly upper triangular for every $F' \in \dprop{\fc}{1}{1}$.
    But every connected $\prop{\fc}$-grid containing $F - \lambda I$ is a cycle formed by a 
    contraction between $F - \lambda I$ and a path with signature $F' \in \dprop{\fc}{1}{1}$,
    with Holant value $\tr((F- \lambda I)F') = 0$. 
    Therefore $F - \lambda I$ is $\fc$-vanishing, contradicting $(1,1)$-quantum-nonvanishing.
\end{proof}

Any $\fc$ failing the condition of Kaplansky's theorem satisfies the condition of the 
following domain separation lemma, which we will apply similarly to the 
Vandermonde-interpolation-based \cite[Proposition 4.1]{orthogonal}.
\begin{lemma}
    \label{lem:subdomain}
    Let $\fc, \gc \subset \tc(\kk^q)$ be Bi-Holant-indistinguishable and $(1,1)$-quantum-nonvanishing. 
    Either $\fc$ and $\gc$ are $(1,1)$-trivial and
    $\dprop{\fc}{1}{1} = \dprop{\gc}{1}{1}$, or there is a nontrivial partition
    $(X,\overline{X})$ of $[q]$ and $T, U \in \gl_q$ such that 
    $\prop{T \fc} \ni I_X^\shortuparrow,I_{\overline{X}}^\shortuparrow \leftrightsquigarrow 
    I_X^\shortuparrow, I_{\overline{X}}^\shortuparrow \in \prop{U \gc}$.

    Furthermore, suppose there are $\dprop{\fc}{1}{1} \ni F \leftrightsquigarrow G \in 
    \dprop{\gc}{1}{1}$ that do not have singleton spectrum and
    have block forms $\left[\begin{smallmatrix} * & 0 \\ * & 0\end{smallmatrix}\right]$ and
    $\left[\begin{smallmatrix} * & * \\ 0 & 0\end{smallmatrix}\right]$ respectively, with
    the first block indexed by $\Delta \subset [q]$. Then we may choose, under the same blocks,
    $T = \left[\begin{smallmatrix} * & 0 \\ * & *\end{smallmatrix}\right]$ and 
    $U = \left[\begin{smallmatrix} * & * \\ 0 & *\end{smallmatrix}\right]$ so that
    $X = [x] \subset \Delta$.
\end{lemma}
\begin{proof}
    If $\fc$ and $\gc$ are $(1,1)$-trivial, then 
    for every $\dprop{\fc}{1}{1} \ni F = \lambda_F I \leftrightsquigarrow  \lambda_G I = G \in \dprop{\gc}{1}{1}$, we have $q \lambda_F = \tr(F) = \tr(G) = q \lambda_G$, hence
    $\lambda_F = \lambda_G$, so $F = G$. Thus $\dprop{\fc}{1}{1} = \dprop{\gc}{1}{1}$.

    Otherwise, \autoref{cor:kaplansky} asserts that
    there are $\dprop{\fc}{1}{1} \ni F \leftrightsquigarrow G \in \dprop{\gc}{1}{1}$ such that one 
    of $F$ or $G$ does not have singleton spectrum.
    By indistinguishability, $\tr(F^k) = \tr(G^k)$ for every $k \geq 0$. Thus
    $F$ and $G$ have the same multiset of eigenvalues.
    In particular, $F$ and $G$ share some eigenvalue $\lambda$ with the same (algebraic)
    multiplicity. We claim that $F$ and $G$ must have the same minimal
    polynomial. Otherwise, suppose WLOG that the minimal polynomial of $F$ does not
    divide the minimal polynomial $p_G$ of $G$. By \autoref{prop:covanishing}, $\fc$ and
    $\gc$ are $(1,1)$-covanishing, but
    $p_G(F) \neq 0 = p_G(G)$ and $\dprop{\fc}{1}{1} \ni p_G(F) \leftrightsquigarrow p_G(G) 
    \in \dprop{\gc}{1}{1}$, a contradiction.

    Choose $T$ and $U$ to be the bases under which $F$ and $G$ are in Jordan normal form,
    respectively. Then, since $\lambda$ has the same multiplicity in $F$ and $G$, we can
    define $X \subset [q]$ such that $F|_{\overline{X}}$ and $G|_{\overline{X}}$ are 
    the union of the $\lambda$-blocks of $F$ and $G$, respectively. Since $F$ and $G$ do
    not have singleton spectrum,
    $X \subset [q]$ is nontrivial. Then choose sufficiently large $r$ such that 
    \begin{equation}
        (F - \lambda I)^r|_{\overline{X}} = (G - \lambda I)^r|_{\overline{X}} = 0. 
        \label{eq:power_r}
    \end{equation}
    Then $\dprop{\fc}{1}{1} \ni (F - \lambda I)^r \leftrightsquigarrow (G - \lambda I)^r
    \in \dprop{\gc}{1}{1}$ are both supported only on $X$, so it follows as above from
    $(1,1)$-covanishing that
    $(F - \lambda I)^r|_{X}$ and $(G - \lambda I)^r|_{X} \in \kk^{X \times X}$
    have the same minimal polynomial $p$. Furthermore, $(F - \lambda I)^r|_{X}$ and 
    $(G - \lambda I)^r|_{X}$ have no $0$-eigenvalues, so $p$ has a nonzero constant 
    term $cI_X$.
    Expanding $p-cI_X$ removes all instances of $I_X$, so we can view $p-cI_X$ as a polynomial
    on full $q \times q$ matrices. Now, by \eqref{eq:power_r},
    \[
        \dprop{\fc}{1}{1} \ni I_X^\shortuparrow = -\frac{1}{c} (p-cI_X)\big((F - \lambda I)^r\big) \leftrightsquigarrow
        -\frac{1}{c} (p-cI_X)\big((G - \lambda I)^r\big) = I_X^\shortuparrow \in 
        \dprop{\gc}{1}{1}
    \]
    and $\dprop{\fc}{1}{1} \ni I_{\overline{X}}^\shortuparrow = I - I_X^\shortuparrow \leftrightsquigarrow I - I_X^\shortuparrow = I_{\overline{X}}^{\shortuparrow} \in \dprop{\gc}{1}{1}$.

    For the second claim, it suffices to show that $F = 
    \left[\begin{smallmatrix} * & 0 \\ * & 0\end{smallmatrix}\right]$ can be put in Jordan
    normal form by $T$ of the form 
    $\left[\begin{smallmatrix} * & 0 \\ * & *\end{smallmatrix}\right]$; the claim about
    $G$ and $U$ follows by transposed reasoning.
    Note that $\{e_{|\Delta|+1}, \ldots, e_q\}$ is a set of linearly independent 0-eigenvectors
    of $F$. We can always choose a $T$ whose final $q-|\Delta|$ columns are
    $\{e_{|\Delta|+1}, \ldots, e_q\}$, giving $T$ the desired block form and, with $\lambda = 0$
    in the proof above, ensuring that $X = [x] \subset \Delta$.
\end{proof}

For a word $w \in \Gamma_{\fc}$, construct $w_{\fc\to\gc} \in \Gamma_{\gc}$ 
by replacing every character $F \in \fc$ in $w$ by the corresponding $G \in \gc$. We now obtain our characterization of simultaneously similarity
of quantum-nonvanishing sets of matrices.
\begin{theorem}
    \label{thm:binary}
    Let $\fc, \gc \subset \kk^{q \times q}$ be $(1,1)$-quantum-nonvanishing.
    Then $\tr(w) = \tr(w_{\fc\to\gc})$ for every word $w \in \Gamma_{\fc}$ if and only if
    there is a $T \in \gl_q$ such that $T F T^{-1} = G$ for every 
    $\fc \ni F \leftrightsquigarrow G \in \gc$.
\end{theorem}
\begin{proof}
    We only need $(\Longrightarrow)$. The assumption is equivalent to
    Bi-Holant-indistinguishability of $\fc$ and $\gc$.
    So, unless we have $\fc \subset \dprop{\fc}{1}{1} = \dprop{\gc}{1}{1} \supset \gc$ 
    and are already done,
    \autoref{lem:subdomain} gives a nontrivial partition $(X, \overline{X})$
    of $[q]$ such that, after suitable transformations,
    $\prop{\fc} \ni I_X^\shortuparrow,I_{\overline{X}}^\shortuparrow 
    \leftrightsquigarrow I_X^\shortuparrow, I_{\overline{X}}^\shortuparrow \in \prop{\gc}$. 

    In general, suppose $\fc \ni I^\shortuparrow_{X_1}, \ldots I^\shortuparrow_{X_s}
    \leftrightsquigarrow I_{X_1}^\shortuparrow, \ldots I_{X_s}^\shortuparrow \in \gc$ 
    for a partition $(X_1,\ldots,X_s)$ of $[q]$. 
    We will show that every subdomain is either
    $(1,1)$-trivial or can be further decomposed into smaller subdomains.
    By Propositions \ref{prop:subdomain_indist}
    and \ref{prop:subdomain_nondegen}, each $\prop{\fc}_{X_i}$ and $\prop{\gc}_{X_i}$ are 
    Bi-Holant-indistinguishable and $(1,1)$-quantum-nonvanishing.
    If any $\prop{\fc}_{X_i}$ and $\prop{\gc}_{X_i}$ are $(1,1)$-nontrivial, then by \autoref{lem:subdomain}
    there are $T,U \in \gl(\kk^{X_i})$ and nontrivial $Y_i \subset X_i$ such that 
    \[
        \prop{T \prop{\fc}_{X_i}} \ni I_{Y_i}^{\shortuparrow X_i}, 
        I_{X_i \setminus Y_i}^{\shortuparrow X_i} \leftrightsquigarrow
        I_{Y_i}^{\shortuparrow X_i}, I_{X_i \setminus Y_i}^{\shortuparrow X_i} \in \prop{U \prop{\gc}_{X_i}}.
    \]
    Define $T^{\shortuparrow} := I_{X_1} \oplus \ldots \oplus I_{X_{i-1}} 
    \oplus T \oplus I_{X_{i+1}} \oplus \ldots \oplus I_{X_s} \in \gl_q$ and replace $\fc$
    with $T^{\shortuparrow} \fc$. This replaces $\prop{\fc}_{X_i}$ with 
    $\prop{T^\shortuparrow \fc}_{X_i} = T \prop{\fc}_{X_i}$ (by \autoref{prop:respect})
    while preserving $I_{X_1}^\shortuparrow,\ldots,I_{X_s}^\shortuparrow$.
    Now $I_{Y_i}^{\shortuparrow X_i} \in \prop{\prop{\fc}_{X_i}}$ and we still have
    $I_{X_i} \in \prop{\fc}$, so \autoref{prop:uparrow} gives
    $I_{Y_i}^\shortuparrow = (I_{Y_i}^{\shortuparrow X_i})^{\shortuparrow} \in \prop{\fc}$.
    Similarly, $I_{X_i\setminus Y_i}^\shortuparrow \in \prop{\fc}$ and, after transforming 
    $\gc$ by $U^{\shortuparrow}$, we obtain 
    $I_{Y_i}^\shortuparrow, I_{X_i\setminus Y_i}^\shortuparrow
    \in \prop{\gc}$, so we have refined the
    partition of $[q]$ to $(X_1,\ldots,Y_i,X_i\setminus Y_i,\ldots,X_s)$.

    Let this process stabilize at a maximal partition $(X_1,\ldots,X_m)$ of $[q]$.
    At this point, \autoref{lem:subdomain} asserts that every
    $\prop{F}_{X_i}$ and $\prop{G}_{X_i}$ are $(1,1)$-trivial and satisfy
    $\dprop{\prop{\fc}_{X_i}}{1}{1} =  \dprop{\prop{\gc}_{X_i}}{1}{1}$. We proceed to inductively transform
    $\fc$ into $\gc$. Suppose that $\bdprop{\prop{\fc}_{X_1 \cup \ldots \cup X_{p-1}}}{1}{1} 
    =  \bdprop{\prop{\gc}_{X_1 \cup \ldots \cup X_{p-1}}}{1}{1}$. Use 
    $I_{X_1 \cup \ldots \cup X_p}^\shortuparrow = \sum_{i=1}^p I_{X_i}^\shortuparrow$
    to isolate $\prop{\fc}_{X_1 \cup \ldots \cup X_{p}}$ and $\prop{\gc}_{X_1 \cup \ldots \cup X_{p}}$;
    by \autoref{prop:subdomain_indist}
    and \autoref{prop:subdomain_nondegen}, $\prop{\fc}_{X_1 \cup \ldots \cup X_{p}}$ and
    $\prop{\gc}_{X_1 \cup \ldots \cup X_{p}}$ are Bi-Holant-indistinguishable and $(1,1)$-quantum-nonvanishing.
    Every $\bdprop{\prop{\fc}_{X_1 \cup \ldots \cup X_p}}{1}{1} \ni F \leftrightsquigarrow G \in 
    \bdprop{\prop{\gc}_{X_1 \cup \ldots \cup X_p}}{1}{1}$ are of the form
    \begin{equation}
        \left[
        \begin{array}{cccc|c}
            \lambda_1 I_{X_1} & F_{1,2} & \hdots & F_{1,p-1} & F_{1,p} \\
            F_{2,1} & \lambda_2 I_{X_2} & \hdots & F_{2,p-1} & F_{2,p} \\
            \vdots & \vdots & \ddots & \vdots & \vdots \\
            F_{p-1,1} & F_{p-1,2} & \hdots & \lambda_{p-1} I_{X_{p-1}} & F_{p-1,p} \\
            \hline
            F_{p,1} & F_{p,2} & \hdots & F_{p,p-1} & \lambda_p I_{X_p}
        \end{array}
        \right]
        \leftrightsquigarrow
        \left[
        \begin{array}{cccc|c}
            \lambda_1 I_{X_1} & F_{1,2} & \hdots & F_{1,p-1} & G_{1,p} \\
            F_{2,1} & \lambda_2 I_{X_2} & \hdots & F_{2,p-1} & G_{2,p} \\
            \vdots & \vdots & \ddots & \vdots & \vdots \\
            F_{p-1,1} & F_{p-1,2} & \hdots & \lambda_{p-1} I_{X_{p-1}} & G_{p-1,p} \\
            \hline
            G_{p,1} & G_{p,2} & \hdots & G_{p,p-1} & \lambda_p I_{X_p}
        \end{array}
        \right]
        \label{eq:block_form}
    \end{equation} 
    where $F_{i,j} := F|_{X_i,X_j}$ and $G_{i,j} := G|_{X_i,X_j}$. 
    Extending \autoref{def:shortuparrowow}, for $i,j \in [m]$ with $i \neq j$, let
    \[
        F_{i,j}^{\shortuparrow} = I_{X_i}^\shortuparrow F I_{X_j}^\shortuparrow \in \dprop{\fc}{1}{1}
    \]
    be the matrix with $F_{i,j}$ in the $(X_i,X_j)$ block and 0 in the other blocks. 
    Since $\fc$ and $\gc$ are $(1,1)$-covanishing by \autoref{prop:covanishing} and
    $F_{i,j}^{\shortuparrow}  \leftrightsquigarrow G_{i,j}^{\shortuparrow}$, we have
    \begin{equation}
        F_{i,j} = 0 \iff G_{i,j} = 0.
        \label{eq:block_covanish}
    \end{equation}
    If any $F_{i,j} \neq 0$ then, by $(1,1)$-quantum-nonvanishing of $\fc$, there is a 
    $\widehat{F_{i,j}^{\shortuparrow}} =: 
    \big(\widehat{F_{k,\ell}}\big)_{k,\ell \in [p]} \in \dprop{\fc}{1}{1}$ such that 
    \begin{align*}
        &0 \neq \Big\langle F_{i,j}^{\shortuparrow}, \widehat{F_{i,j}^{\shortuparrow}} \Big\rangle
        = \tr\Big(F_{i,j}^{\shortuparrow} \widehat{F_{i,j}^{\shortuparrow}}\Big) \\
        &= \tr\left(
        \begin{bmatrix}
            0 & \hdots & 0 & \hdots\hdots & 0 \\
            \vspace{-0.25cm}
            \vdots & & \vdots & & \vdots \\
            \vdots & & \vdots & & \vdots \\
            0 & \hdots & F_{i,j} & \hdots\hdots & 0 \\
            \vdots &  & \vdots & & \vdots \\
            0 & \hdots & 0 & \hdots\hdots & 0
        \end{bmatrix}
        \begin{bmatrix}
            * & \hdots\hdots & * & \hdots & * \\
            \vdots & & \vdots & & \vdots \\
            * & \hdots\hdots & \widehat{F_{j,i}} & \hdots & * \\
            \vspace{-0.25cm}
            \vdots &  & \vdots & & \vdots \\
            \vdots & & \vdots & & \vdots \\
            * & \hdots\hdots & * & \hdots & *
        \end{bmatrix}       
         \right)
        = \tr\Big(F_{i,j} \widehat{F_{j,i}}\Big).
        \numberthis\label{eq:block_i}
    \end{align*}
    Note that $F_{i,j} \widehat{F_{j,i}}$ is the
    $(X_i,X_i)$-block of $F_{i,j}^{\shortuparrow} \widehat{F_{i,j}^{\shortuparrow}} \in 
    \prop{\fc}$
    in \eqref{eq:block_i}, so $F_{i,j} \widehat{F_{j,i}} \in \prop{\fc}_{X_i}$. But $\prop{\fc}_{X_i}$ is
    $(1,1)$-trivial, so $F_{i,j} \widehat{F_{j,i}} = \lambda_{F,i,j} I_{X_i}$ for some 
    $\lambda_{F,i,j} \neq 0$. We simultaneously have
    \begin{align*}
        &0 \neq 
        \tr\Big( F_{i,j}^{\shortuparrow} \widehat{F_{i,j}^{\shortuparrow}} \Big)
        = \tr\Big( \widehat{F_{i,j}^{\shortuparrow}} F_{i,j}^{\shortuparrow}\Big) \\
        &= \tr\left(
        \begin{bmatrix}
            * & \hdots\hdots & * & \hdots & * \\
            \vdots & & \vdots & & \vdots \\
            * & \hdots\hdots & \widehat{F_{j,i}} & \hdots & * \\
            \vspace{-0.25cm}
            \vdots &  & \vdots & & \vdots \\
            \vdots & & \vdots & & \vdots \\
            * & \hdots\hdots & * & \hdots & *
        \end{bmatrix}
        \begin{bmatrix}
            0 & \hdots & 0 & \hdots\hdots & 0 \\
            \vspace{-0.25cm}
            \vdots & & \vdots & & \vdots \\
            \vdots & & \vdots & & \vdots \\
            0 & \hdots & F_{i,j} & \hdots\hdots & 0 \\
            \vdots &  & \vdots & & \vdots \\
            0 & \hdots & 0 & \hdots\hdots & 0
        \end{bmatrix}
        \right)
        = \tr\Big(\widehat{F_{j,i}} F_{i,j}\Big),
    \end{align*}
    and the $(1,1)$-triviality of $\fc_{X_j}$ gives 
    $\widehat{F_{j,i}} F_{i,j}= \lambda_{F,i,j}' I_{X_j}$ for $\lambda_{F,i,j}' \neq 0$. Hence
    $F_{i,j}$ and $\widehat{F_{j,i}}$ are both left and right-invertible, so must be square,
    giving $|X_i| = |X_j|$ and $\lambda_{F,i,j} = \lambda_{F,i,j}'$. On the $\gc$ side, 
    \eqref{eq:block_covanish} gives $G_{i,j} \neq 0$ as well, so
    let $\widehat{F_{i,j}^{\shortuparrow}} \leftrightsquigarrow
    \widehat{G_{i,j}^{\shortuparrow}} = (\widehat{G_{k,\ell}})_{k,\ell \in [p]}$. 
    In general, if $F, \widetilde{F} \in \dprop{\fc}{1}{1}$, then $F_{i,j} \widetilde{F}_{j,k}$ is the
    $(X_i,X_k)$-block of $F_{i,j}^\shortuparrow \widetilde{F}_{j,k}^\shortuparrow \in 
    \dprop{\fc}{1}{1}$. In particular, $F_{i,j} \widetilde{F}_{j,i} \in 
    \dprop{\prop{\fc}_{X_i}}{1}{1}$. Then, since
    $\dprop{\prop{\fc}_{X_i}}{1}{1} = \dprop{\prop{\gc}_{X_i}}{1}{1}$, 
    \begin{equation}
        F_{i,j} \widetilde{F}_{j,i} = G_{i,j} \widetilde{G}_{j,i} \text{ for every }
        \dprop{\fc}{1}{1} \ni F, \widetilde{F} \leftrightsquigarrow G, \widetilde{G}
        \in \dprop{\gc}{1}{1}.
        \label{eq:i_block_equal}
    \end{equation}
    In particular,
    \begin{equation}
        F_{i,j} \widehat{F_{j,i}} = \widehat{F_{j,i}} F_{i,j} =
        G_{i,j} \widehat{G_{j,i}} = \widehat{G_{j,i}} G_{i,j} = \lambda_{F,i,j} I_{|X_i|}.
        \label{eq:fij}
    \end{equation}

    For $k < p$, fix $F^{(k)} \in \dprop{\fc}{1}{1}$ with $F^{(k)}_{k,p} \neq 0$, if any such $F^{(k)}$ 
    exists, and let $F^{(k)} \leftrightsquigarrow G^{(k)}$. Define
    \begin{equation}
        \label{eq:tp}
        T_p = I_{X_p} \text{ and } T_k = \begin{cases} 
            \lambda_{F^{(k)},k,p}^{-1} G^{(k)}_{k,p} \widehat{F^{(k)}_{p,k}}
            & \exists F' \in \dprop{\fc}{1}{1} \text{ such that } F'_{k,p} \neq 0 \\
            I_{X_k} & \text{otherwise}
        \end{cases}
        \in \kk^{X_k \times X_k}
    \end{equation}
    and $T := \bigoplus_{k=1}^p T_k$. 
    By \eqref{eq:fij}, $T$ is invertible and $T^{-1} = \bigoplus_{k=1}^p T_k^{-1}$, where
    \begin{equation}
        \label{eq:tp_inv}
        T_p^{-1} = I_{X_p} \text{ and }
        T_k^{-1} = \begin{cases} 
            \lambda_{F^{(k)},k,p}^{-1} F^{(k)}_{k,p} \widehat{G^{(k)}_{p,k}}
            &\exists F' \in \dprop{\fc}{1}{1} \text{ such that } F'_{k,p} \neq 0 \\
            I_{X_k} & \text{otherwise.}
        \end{cases}
    \end{equation}
    We claim that $T F T^{-1} = G$ for every $\bdprop{\prop{\fc}_{X_1 \cup \ldots \cup X_p}}{1}{1} 
    \ni F \leftrightsquigarrow
    G \in \bdprop{\prop{\gc}_{X_1 \cup \ldots \cup X_p}}{1}{1}$. This is equivalent to
    $T_i F_{i,j} T_j^{-1} = G_{i,j}$
    for arbitrary $i,j \leq p$. If $F_{i,j} = 0$ then $G_{i,j} = 0$ by 
    \eqref{eq:block_covanish} and we are done. Otherwise, we consider several cases.
    \begin{enumerate}
        \item If $i=j$, then $F_{i,i} = G_{i,i} = \lambda_i$ by \eqref{eq:block_form}, so
            $T_i F_{i,i} T_i^{-1} = G_{i,i}$.
        \item If $i \neq p = j$, then $F_{i,j} = F_{i,p} \neq 0$ implies that 
            $T_i = \lambda_{F^{(i)},i,p}^{-1} G^{(i)}_{i,p} \widehat{F^{(i)}_{p,i}}$, so, applying \eqref{eq:i_block_equal} followed by
            \eqref{eq:fij},
            \[
                T_i F_{i,p} T_p^{-1}
                = \lambda_{F^{(i)},i,p}^{-1} {G^{(i)}_{i,p}} \widehat{F^{(i)}_{p,i}} F_{i,p} I_{X_p}
                = \lambda_{F^{(i)},i,p}^{-1} {G^{(i)}_{i,p}} \widehat{G^{(i)}_{p,i}} G_{i,p}
                = G_{i,p}.
            \]
        \item If $i = p \neq j$, then $F_{i,j} = F_{p,j} \neq 0$ implies that
            $\widehat{F_{j,p}} \neq 0$ by \eqref{eq:fij}, so
            $T_j^{-1} = \lambda_{F^{(j)},j,p}^{-1} F^{(j)}_{j,p} \widehat{G^{(j)}_{p,j}}$ and,
            applying \eqref{eq:i_block_equal} followed by \eqref{eq:fij},
            \[
                T_p F_{p,j} T_j^{-1}
                = \lambda_{F^{(j)},j,p}^{-1} I_{X_p} F_{p,j} F^{(j)}_{j,p} \widehat{G^{(j)}_{p,j}}
                = \lambda_{F^{(j)},j,p}^{-1} G_{p,j} G^{(j)}_{j,p} \widehat{G^{(j)}_{p,j}}
                = G_{p,j}.
            \]
        \item If $i,j,p$ are all distinct, then $F_{i,j} = G_{i,j}$ by induction, so if 
            $T_i = T_j^{-1} = I$ then we are done. Otherwise, there are two possibilities.
            \begin{enumerate}
                \item If $T_i \neq I$ then by \eqref{eq:tp} there is a $F'_{i,p} \neq 0$.
                Now $\widehat{F_{j,i}}$ and $F'_{i,p}$ are 
                invertible by \eqref{eq:fij}, so $\widehat{F_{j,i}} F'_{i,p} \neq 0$
                is the $(X_j,X_p)$-block of
                $\widehat{F_{j,i}}^\shortuparrow (F'_{i,p})^\shortuparrow$. Thus
                $T_j^{-1} = \lambda_{F^{(j)},j,p} F^{(j)}_{j,p} \widehat{G^{(j)}_{p,j}}$ by \eqref{eq:tp_inv}.

                \item If $T_j^{-1} \neq I$ then by \eqref{eq:tp_inv}
                    there is a $F'_{j,p} \neq 0$. Now $F_{i,j}$
                    and $F'_{j,p}$ are invertible by \eqref{eq:fij}, so 
                    $F_{i,j} F'_{j,p} \neq 0$ is the $(X_i,X_p)$-block of
                    $F_{i,j}^\shortuparrow (F'_{j,p})^\shortuparrow$. Thus
                    $T_i = \lambda_{F^{(i)},i,p} G^{(i)}_{i,p} \widehat{F^{(i)}_{p,i}}$ by \eqref{eq:tp}.
            \end{enumerate}
            In either case, both $T_i$ and $T_j^{-1}$ fall into the respective first cases in
            \eqref{eq:tp} and \eqref{eq:tp_inv}. Therefore, by reasoning similar to
            \eqref{eq:i_block_equal}, followed by \eqref{eq:fij},
            \[
                T_i F_{i,j} T_j^{-1} 
                = \lambda_{F^{(i)},i,p}^{-1} \lambda_{F^{(j)},j,p}^{-1}
                G^{(i)}_{i,p} \widehat{F^{(i)}_{p,i}} F_{i,j} F^{(j)}_{j,p} \widehat{G^{(j)}_{p,j}}
                = \lambda_{F^{(i)},i,p}^{-1} \lambda_{F^{(j)},j,p}^{-1}
                G^{(i)}_{i,p} \widehat{G^{(i)}_{p,i}} G_{i,j} G^{(j)}_{j,p} \widehat{G^{(j)}_{p,j}}
                = G_{i,j}.
            \]
    \end{enumerate}
    Now, after transforming $\fc$ by $T \oplus I_{X_{p+1}} \oplus \ldots \oplus I_{X_m}$, we
    have $\bdprop{\prop{\fc}_{X_1 \cup \ldots \cup X_p}}{1}{1} = 
    \bdprop{\prop{\gc}_{X_1 \cup \ldots \cup X_p}}{1}{1}$. After similar transforms at each level
    of the induction, we obtain 
    $\fc \subset \dprop{\fc}{1}{1} = \dprop{\gc}{1}{1} \supset \gc$.
\end{proof}

\subsection{The Bipartite Case}
\label{sec:bipartite}
To prove \autoref{thm:main2}, our second main result, we need the following construction,
frequently employed in the study of counting indistinguishability 
\cite{lovasz,young2022equality,cai_planar_2023,orthogonal}.
\begin{definition}[$\oplus$]
    Let $\fc$ and $\gc$ be sets on domains $V(\fc)$ and $V(\gc)$, respectively. Define
    a set $\fc \oplus \gc = \{F \oplus G \mid \fc \ni F \leftrightsquigarrow G \in \gc\}$ on
    domain $V(\fc) \sqcup V(\gc)$ and bijective with $\fc$ and $\gc$, where
    \[
        (F \oplus G)(\va) = \begin{cases} F(\va) & \va \in V(\fc)^n \\ G(\va) 
        & \va \in V(\gc)^n \\ 0 & \text{otherwise} \end{cases}
    \]
    for $n$-ary $F$ and $G$ and $\va \in (V(\fc) \sqcup V(\gc))^n$.
\end{definition}
Providing any input from $V(\fc)$ to a connected $\prop{\fc} \oplus \prop{\gc}$-gadget forces
all edges in the gadget take values in $V(\fc)$ (all other edge assignments give 0).
Note the difference between $\prop{\fc} \oplus \prop{\gc}$ and $\prop{\fc \oplus \gc}$.
Every signature in $\prop{\fc} \oplus \prop{\gc}$, such as 
$(F_1 \otimes F_2) \oplus (G_1 \otimes G_2)$, is zero on mixed inputs from $\fc$ and $\gc$.
On the other hand,
$(F_1 \oplus G_1) \otimes (F_2 \oplus G_2) \in \prop{\fc \oplus \gc}$, being
disconnected, could be nonzero on
inputs from $V(\fc)$ to the first factor and $V(\gc)$ to the second and vice-versa. 

For
$K \in \prop{\prop{\fc} \oplus \prop{\gc}}$, use $K|_{\fc}$ as shorthand for $K|_{V(\fc)}$. 
\begin{proposition}
    \label{prop:oplus}
    If $K \in \prop{\prop{\fc} \oplus \prop{\gc}}$, then
    $\prop{\fc} \ni K|_{\fc} \leftrightsquigarrow K|_{\gc} \in \prop{\gc}$.
\end{proposition}
\begin{proof}
    By definition, $K$ is the signature of some quantum $\prop{\fc} \oplus \prop{\gc}$-gadget 
    $\vk$ with no connected components without a dangling edge. 
    To construct $K|_{\fc}$, restrict all inputs to $\vk$ to $V(\fc)$. As discussed above,
    this restricts all edges of all gadgets composing $\vk$ to $V(\fc)$.
    Thus $K|_{\fc}$ is the signature of $\vk_{\prop{\fc} \oplus \prop{\gc} \to \prop{\fc}}$.
    Similarly, $\vk_{\prop{\fc} \oplus \prop{\gc} \to \prop{\gc}}$ has signature $K|_{\gc}$,
    and the result follows.
\end{proof}

\begin{proposition}
    \label{prop:oplus2}
    Assume $\fc$ and $\gc$ are Bi-Holant-indistinguishable and let
    $\prop{\fc} \ni F \leftrightsquigarrow G \in \prop{\gc}$ and 
    $K \in \prop{\prop{\fc} \oplus \prop{\gc}}$. Then
    \[
        \langle K, F \oplus G \rangle 
        = \langle K|_{\fc}, F \rangle + \langle K|_{\gc}, G \rangle
        = 2 \langle K|_{\fc}, F \rangle.
    \]
\end{proposition}
\begin{proof}
    In each nonzero term of $\langle K, F \oplus G \rangle$, either all
    inputs to both $K$ and $F \oplus G$ are from $V(\fc)$, 
    or all inputs to both $K$ are $F \oplus G$ are from $V(\gc)$, giving the first
    equality. The second equality follows from indistinguishability and \autoref{prop:oplus}.
\end{proof}

\begin{lemma}
    \label{lem:oplus}
    Assume $\fc$ and $\gc$ are Bi-Holant-indistinguishable. Then
    $\prop{\fc} \oplus \prop{\gc}$ is quantum-nonvanishing if and only if $\fc$ and $\gc$ are 
    both quantum-nonvanishing.
\end{lemma}
\begin{proof}
    $(\Rightarrow)$: We will show that
    $\fc$ is quantum-nonvanishing; the proof for $\gc$ is similar. Let $F \in \prop{\fc}$ be nonzero,
    and $F \leftrightsquigarrow G$. Since $F \oplus G \in \prop{\fc} \oplus \prop{\gc}$, 
    the quantum-nonvanishing of $\prop{\fc} \oplus \prop{\gc}$ guarantees the existence of a
    $K \in \prop{\prop{\fc} \oplus \prop{\gc}}$ such that, by \autoref{prop:oplus2},
    $0 \neq \langle K, F \oplus G\rangle = 2 \langle K|_{\fc}, F \rangle$.
    \autoref{prop:oplus} asserts that $K|_{\fc} \in \prop{\fc}$, so $K|_{\fc}$ witnesses that
    $F$ is $\fc$-nonvanishing.

    $(\Leftarrow)$: Assume $\fc$ and $\gc$ are quantum-nonvanishing, and let 
    $0 \neq K \in \prop{\prop{\fc} \oplus \prop{\gc}}$ be the signature of a quantum 
    $\prop{\fc} \oplus \prop{\gc}$-gadget $\vk$.
    First suppose that $K|_{\fc} \neq 0$. By \autoref{prop:oplus}, 
    $K|_{\fc} \in \prop{\fc}$, so
    by the quantum-nonvanishing of $\fc$ there is a $\widehat{F} \in \prop{\fc}$ such that 
    $\langle K|_{\fc}, \widehat{F} \rangle \neq 0$. Then, letting 
    $\widehat{F} \leftrightsquigarrow \widehat{G}$, \autoref{prop:oplus2} gives
    $\langle K, \widehat{F} \oplus \widehat{G} \rangle = 
    2 \langle K|_{\fc}, \widehat{F} \rangle \neq 0$,
    so $\widehat{F} \oplus \widehat{G} \in \prop{\fc} \oplus \prop{\gc}$ witnesses that
    $K$ is $\prop{\fc} \oplus \prop{\gc}$-nonvanishing.

    If $K|_{\fc} = 0$ then $K|_{\gc} = 0$ as well by Propositions \ref{prop:oplus} and 
    \ref{prop:covanishing}.
    Since $K \neq 0$, there is a nontrivial partition of the inputs of $K$ into 
    $X_1 \sqcup X_2$ such that the block $K|_{X_1 \gets V(\fc), X_2 \gets V(\gc)}$ of $K$ 
    (in which inputs in $X_1$ are restricted to $V(\fc)$ and inputs in $X_2$ are
    restricted to $V(\gc)$) is nonzero. Let $\vk = \mathbf{M} + \sum_{i=1}^j c_i \vj_i$,
    where each $\vj_i$ is a $\prop{\fc} \oplus \prop{\gc}$-gadget composed of two components 
    $\vj_{i,1}$ and $\vj_{i,2}$, not necessarily themselves connected but disconnected from 
    each other, such that the dangling edges
    of $\vj_i$ indexed by $X_1$ (resp. $X_2$) are incident to $\vj_{i,1}$ (resp. $\vj_{i,2}$),
    and $\mathbf{M}$
    is the quantum-$\prop{\fc} \oplus \prop{\gc}$-gadget composed of all terms of $\vk$ in
    which there is a path between some input indexed by $X_1$ and some input indexed by
    $X_2$. Hence the signature $M$ of $\mathbf{M}$ satisfies
    \begin{equation}
        \label{eq:jm_blocks}
        M|_{X_1 \gets V(\fc), X_2 \gets V(\gc)} = M|_{X_1 \gets V(\gc), X_2 \gets V(\fc)} = 0.
    \end{equation}
    By reordering the left dangling edges and right dangling edges of $\vk$, 
    which does not change whether $K$ is
    $\prop{\fc} \oplus \prop{\gc}$-nonvanishing, we may assume 
    $\vj_i = \vj_{i,1} \otimes \vj_{i,2}$, so their signatures satisfy
    \begin{equation}
        J_i|_{X_1 \gets V(\fc), X_2 \gets V(\gc)} = J_{i,1}|_{\fc} \otimes J_{i,2}|_{\gc}
        \text{ and }
        J_i|_{X_1 \gets V(\gc), X_2 \gets V(\fc)} = J_{i,1}|_{\gc} \otimes J_{i,2}|_{\fc}.
        \label{eq:ji_decomp}
    \end{equation}

    \begin{figure}[ht!]
        \centering
        \begin{tikzpicture}[scale=.39]
\GraphInit[vstyle=Classic]
\SetUpEdge[style=-]

\def\wlen{1}
\def\wgap{0.4}
\def\la{3}
\def\lb{1}
\def\ra{3.5}
\def\rb{2.5}
\def\rc{\lb}
\def\boxd{\lb}
\def\boxy{4}
\def\boxx{3}
\def\boxgap{0.5}
\def\xsh{6.5}
\def\xxsh{12}

\draw[thin, color=gray] (-\wlen,\la) -- (0,\la);
\draw[thin, color=gray] (-\wlen,\lb) -- (0,\lb);
\draw[thin, color=gray] (\boxx,\ra) -- (\boxx+\wlen,\ra);
\draw[thin, color=gray] (\boxx,\rb) -- (\boxx+\wlen,\rb);
\draw[thin, color=gray] (\boxx,\rc) -- (\boxx+\wlen,\rc);
\filldraw[color=black!70, fill=blue!8] (0,0) rectangle (\boxx,\boxy);
\node at (\boxx/2,\boxy/2) {$M$};
\node at (\boxx+1.8,\boxy/2) {$+$};
\node at (\boxx+1.8,0) {\small $\textcolor{red}{\vk}$};
\draw[thick,color=red] (-\boxgap,-\boxgap) rectangle (\boxx+\xsh+\boxgap,\boxy+\boxgap);

\begin{scope}[xshift=\xsh cm]
    \draw[thin, color=gray] (-\wlen,\la) -- (0,\la);
    \draw[thin, color=gray] (-\wlen,\lb) -- (0,\lb);
    \draw[thin, color=gray] (\boxx,\ra) -- (\boxx+\wlen,\ra);
    \draw[thin, color=gray] (\boxx,\rb) -- (\boxx+\wlen,\rb);
    \draw[thin, color=gray] (\boxx,\rc) -- (\boxx+\wlen,\rc);
    \filldraw[color=black!70, fill=blue!8] (0,0) rectangle (\boxx,\boxy/2-\boxgap/2);
    \node at (\boxx/2,\boxy/4-\boxgap/4) {$J_{1,2}$};
    \filldraw[color=black!70, fill=blue!8] (0,\boxy/2+\boxgap/2) rectangle (\boxx,\boxy);
    \node at (\boxx/2,\boxy/2 + \boxy/4+\boxgap/4) {$J_{1,1}$};
\end{scope}

\begin{scope}[xshift=\xxsh cm]
    \draw[thin, color=gray] (-\wlen,\ra) .. controls +(0.5,0) .. (0,\ra/2+\rb/2);
    \draw[thin, color=gray] (-\wlen,\rb) .. controls +(0.5,0) .. (0,\ra/2+\rb/2);
    \draw[thin, color=gray] (0,\ra/2+\rb/2) -- (\wlen,\ra/2+\rb/2);
    \draw[thin, color=gray] (-\wlen,\rc) -- (\wlen,\rc);
    \node[draw, fill=black, regular polygon, regular polygon sides=4, minimum size = 7pt, inner sep = 1pt] at (0,\rc) {};
    \node at (\wlen*0.7,\rc-0.8) {\scriptsize $\widehat{F_2} \oplus \widehat{G_2}$};
    \node[draw, fill=black, regular polygon, regular polygon sides=5, minimum size = 7pt, inner sep = 1pt] at (0,\ra/2+\rb/2) {};
    \node at (\wlen,\ra/2+\rb/2-0.8) {\scriptsize $\widehat{F_1} \oplus \widehat{G_1}$};
\end{scope}

\draw[thin, color=gray] 
(-\wlen-\wgap,\lb) .. controls 
+(-1.5,0) and +(-3,0) ..
(0,-1) -- (\xxsh,-1) .. controls
+(3,0) and +(1.5,0) ..
(\xxsh+\wlen+\wgap,\lb);

\draw[thin, color=gray] 
(-\wlen-\wgap,\la) .. controls 
+(-1.5,0) and +(-3,0) ..
(0,\boxy+1) -- (\xxsh,\boxy+1) .. controls
+(3,0) and +(1.5,0) ..
(\xxsh+\wlen+\wgap,\la);

\begin{scope}[xshift=19 cm]
    \node at (-3.4,\la/2+\lb/2) {$=$};
    \draw[thin, color=gray] (-\wlen,\la) -- (0,\la);
    \draw[thin, color=gray] (-\wlen,\lb) -- (0,\lb);
    \draw[thin, color=gray] (\boxx,\ra) -- (\boxx+\wlen,\ra);
    \draw[thin, color=gray] (\boxx,\rb) -- (\boxx+\wlen,\rb);
    \draw[thin, color=gray] (\boxx,\rc) -- (\boxx+\wlen,\rc);
    \filldraw[color=black!70, fill=blue!8] (0,0) rectangle (\boxx,\boxy/2-\boxgap/2);
    \node at (\boxx/2,\boxy/4-\boxgap/4) {$J_{1,2}|_{\gc}$};
    \filldraw[color=black!70, fill=blue!8] (0,\boxy/2+\boxgap/2) rectangle (\boxx,\boxy);
    \node at (\boxx/2,\boxy/2 + \boxy/4+\boxgap/4) {$J_{1,1}|_{\fc}$};

    \begin{scope}[xshift=5.4 cm]
    \draw[thin, color=gray] (-\wlen,\ra) .. controls +(0.5,0) .. (0,\ra/2+\rb/2);
    \draw[thin, color=gray] (-\wlen,\rb) .. controls +(0.5,0) .. (0,\ra/2+\rb/2);
    \draw[thin, color=gray] (0,\ra/2+\rb/2) -- (\wlen,\ra/2+\rb/2);
    \draw[thin, color=gray] (-\wlen,\rc) -- (\wlen,\rc);
    \node[draw, fill=black, circle, minimum size = 5pt, inner sep = 1pt] at (0,\rc) {};
    \node at (0,\rc-0.8) {\scriptsize $\widehat{G_2}$};
    \node[draw, fill=black, circle, minimum size = 5pt, inner sep = 1pt] at (0,\ra/2+\rb/2) {};
    \node at (0,\ra/2+\rb/2-0.8) {\scriptsize $\widehat{F_1}$};
    \end{scope}

    \draw[thin, color=gray] 
    (-\wlen-\wgap,\lb) .. controls 
    +(-1,0) and +(-3,0) ..
    (0,-0.6) -- (5.4,-0.6) .. controls
    +(3,0) and +(1,0) ..
    (\boxx + 3*\wlen + 2*\wgap,\lb);

    \draw[thin, color=gray] 
    (-\wlen-\wgap,\la) .. controls 
    +(-1,0) and +(-3,0) ..
    (0,\boxy+0.6) -- (5.4,\boxy+0.6) .. controls
    +(3,0) and +(1,0) ..
    (\boxx + 3*\wlen + 2*\wgap,\la);
\end{scope}

\begin{scope}[xshift=30 cm]
    \node at (-2.8,\la/2+\lb/2) {$+$};
    \draw[thin, color=gray] (-\wlen,\la) -- (0,\la);
    \draw[thin, color=gray] (-\wlen,\lb) -- (0,\lb);
    \draw[thin, color=gray] (\boxx,\ra) -- (\boxx+\wlen,\ra);
    \draw[thin, color=gray] (\boxx,\rb) -- (\boxx+\wlen,\rb);
    \draw[thin, color=gray] (\boxx,\rc) -- (\boxx+\wlen,\rc);
    \filldraw[color=black!70, fill=blue!8] (0,0) rectangle (\boxx,\boxy/2-\boxgap/2);
    \node at (\boxx/2,\boxy/4-\boxgap/4) {$J_{1,2}|_{\fc}$};
    \filldraw[color=black!70, fill=blue!8] (0,\boxy/2+\boxgap/2) rectangle (\boxx,\boxy);
    \node at (\boxx/2,\boxy/2 + \boxy/4+\boxgap/4) {$J_{1,1}|_{\gc}$};

    \begin{scope}[xshift=5.4 cm]
    \draw[thin, color=gray] (-\wlen,\ra) .. controls +(0.5,0) .. (0,\ra/2+\rb/2);
    \draw[thin, color=gray] (-\wlen,\rb) .. controls +(0.5,0) .. (0,\ra/2+\rb/2);
    \draw[thin, color=gray] (0,\ra/2+\rb/2) -- (\wlen,\ra/2+\rb/2);
    \draw[thin, color=gray] (-\wlen,\rc) -- (\wlen,\rc);
    \node[draw, fill=black, circle, minimum size = 5pt, inner sep = 1pt] at (0,\rc) {};
    \node at (0,\rc-0.8) {\scriptsize $\widehat{F_2}$};
    \node[draw, fill=black, circle, minimum size = 5pt, inner sep = 1pt] at (0,\ra/2+\rb/2) {};
    \node at (0,\ra/2+\rb/2-0.8) {\scriptsize $\widehat{G_1}$};
    \end{scope}

    \draw[thin, color=gray] 
    (-\wlen-\wgap,\lb) .. controls 
    +(-1,0) and +(-3,0) ..
    (0,-0.6) -- (5.4,-0.6) .. controls
    +(3,0) and +(1,0) ..
    (\boxx + 3*\wlen + 2*\wgap,\lb);

    \draw[thin, color=gray] 
    (-\wlen-\wgap,\la) .. controls 
    +(-1,0) and +(-3,0) ..
    (0,\boxy+0.6) -- (5.4,\boxy+0.6) .. controls
    +(3,0) and +(1,0) ..
    (\boxx + 3*\wlen + 2*\wgap,\la);
\end{scope}
\end{tikzpicture}
        \caption{Illustrating \eqref{eq:tensor_inner_prod} for $\vk = \mathbf{M} + 
        \vj_{1,1} \otimes \vj_{1,2}$.}
        \label{fig:disconnected}
    \end{figure}
    For any $\prop{\fc} \ni \widehat{F_1}, \widehat{F_2} \leftrightsquigarrow \widehat{G_1}, 
    \widehat{G_2} \in \prop{\gc}$ of appropriate shape (see \autoref{fig:disconnected}), 
    reasoning similar to \autoref{prop:oplus2}, with the assumption that
    $K|_{\fc} = K|_{\gc} = 0$ and \eqref{eq:jm_blocks} and \eqref{eq:ji_decomp}, gives
    \begin{align*}
        &\langle K, (\widehat{F_1} \oplus \widehat{G_1}) \otimes (\widehat{F_2} \oplus \widehat{G_2}) \rangle \\
        &\quad = \langle K|_{\fc}, \widehat{F_1} \otimes \widehat{F_2} \rangle
        + \langle K|_{\gc}, \widehat{G_1} \otimes \widehat{G_2} \rangle \\
        &\qquad+ \langle K|_{X_1 \gets V(\fc), X_2 \gets V(\gc)}, 
                \widehat{F_1} \otimes \widehat{G_2} \rangle
        + \langle K|_{X_1 \gets V(\gc), X_2 \gets V(\fc)}, 
                \widehat{G_1} \otimes \widehat{F_2} \rangle \\
        &\quad= \sum_{i=1}^j c_i \langle J_{i,1}|_{\fc}, \widehat{F_1} \rangle
        \langle J_{i,2}|_{\gc}, \widehat{G_2} \rangle
        + \sum_{i=1}^j c_i \langle J_{i,1}|_{\gc}, \widehat{G_1} \rangle
        \langle J_{i,2}|_{\fc}, \widehat{F_2} \rangle \\
        &\quad= 2 \left\langle \sum_{i=1}^j c_i
            J_{i,1}|_{\fc} \otimes J_{i,2}|_{\gc},~
            \widehat{F_1} \otimes \widehat{G_2}\right\rangle.
            \numberthis\label{eq:tensor_inner_prod}
    \end{align*}
    Each $J_{i,2}|_{\gc} \in \prop{\gc}$, which is closed under linear combinations, so we 
    may successively eliminate any
    $J_{i,1}|_{\fc}$ which is linearly dependent on the other $J_{i',1}|_{\fc}$ to obtain
    \begin{equation}
        0 \neq K|_{X_1 \gets V(\fc), X_2 \gets V(\gc)}
        = \sum_{i=1}^j c_i J_{i,1}|_{\fc} \otimes J_{i,2}|_{\gc}
        = \sum_{i=1}^{j'} c'_i E_i \otimes H_i
        \label{eq:lin_ind}
    \end{equation}
    for $H_1, \ldots H_{j'} \in \prop{\gc}$ and linearly independent 
    $E_1,\ldots,E_{j'} \in \prop{\fc}$. Substituting into \eqref{eq:tensor_inner_prod} gives
    \begin{equation}
        \langle K, (\widehat{F_1} \oplus \widehat{G_1}) \otimes (\widehat{F_2} \oplus \widehat{G_2}) \rangle
        = 2 \left\langle \sum_{i=1}^{j'} c'_i E_i \otimes H_i, ~\widehat{F_1} \otimes 
        \widehat{G_2} \right\rangle
        = \left\langle 2\sum_{i=1}^{j'} c'_i \langle H_i, \widehat{G_2} \rangle E_i,
        ~\widehat{F_1} \right\rangle.
        \label{eq:kfg_inner_prod}
    \end{equation}
    Some $c_i' H_i \neq 0$ by \eqref{eq:lin_ind}, so
    quantum-nonvanishing of $\gc$ gives a $\widehat{G_2}$ such that
    $c'_i \langle H_i, \widehat{G_2} \rangle \neq 0$. Hence, by linear independence,
    $0 \neq 2\sum_{i=1}^{j'} c'_i \langle H_i, \widehat{G_2} \rangle E_i \in \prop{\fc}$,
    so by \eqref{eq:kfg_inner_prod} and quantum-nonvanishing of $\fc$,
    there is an $\widehat{F_1}$ such that $\langle K, (\widehat{F_1} \oplus \widehat{G_1}) 
    \otimes (\widehat{F_2} \oplus \widehat{G_2}) \rangle \neq 0$.
    This $(\widehat{F_1} \oplus \widehat{G_1}) \otimes (\widehat{F_2} \oplus \widehat{G_2}) \in \prop{\prop{\fc} \oplus \prop{\gc}}$ witnesses that $K$ is
    $\prop{\fc} \oplus \prop{\gc}$-nonvanishing.
\end{proof}

Next we have the following analogue of \cite[Lemma 3.2]{orthogonal}, with a similar proof.
\begin{lemma}
    \label{lem:nonconstructive}
    If $\fc$ and $\gc$ are Bi-Holant-indistinguishable and quantum-nonvanishing, then there exists an
    $H \in \stab(\prop{\prop{\fc} \oplus \prop{\gc}})$ with $H|_{\fc,\gc} \neq 0$ or
    $H|_{\gc,\fc} \neq 0$.
\end{lemma}
\begin{proof}
    First observe that \autoref{thm:duality} applies to the wheeled PROP 
    $\prop{\prop{\fc} \oplus \prop{\gc}}$, which is quantum-nonvanishing by \autoref{lem:oplus}
    (hence the claimed $\stab(\prop{\prop{\fc} \oplus \prop{\gc}})$ exists). Assume that
    every $H \in \stab(\prop{\prop{\fc} \oplus \prop{\gc}})$ satisfies 
    $H|_{\fc,\gc} = H|_{\gc,\fc} = 0$ (i.e. is block-diagonal). Then 
    \[
        I_{\fc} \oplus 2I_{\gc} = \begin{bmatrix} I & 0 \\ 0 & 2I \end{bmatrix} 
        \in \tc(\kk^{2q})
    \]
    satisfies $H(I_{\fc} \oplus 2I_{\gc}) H^{-1} = I_{\fc} \oplus 2I_{\gc}$ 
    for every $H \in 
    \stab(\prop{\prop{\fc} \oplus \prop{\gc}})$, so, by \autoref{thm:duality},
    $I_{\fc} \oplus 2I_{\gc} \in \prop{\prop{\fc} \oplus \prop{\gc}}$. But
    \autoref{prop:oplus} gives
    \[
        \prop{\fc} \ni I_{\fc} = (I_{\fc} \oplus 2I_{\gc})|_{\fc}
        \leftrightsquigarrow (I_{\fc} \oplus 2I_{\gc})|_{\gc} = 2I_{\gc} \in 
        \prop{\gc},
    \]
    violating indistinguishability, as $\tr(I_{\fc}) = q \neq 2q = \tr(2I_{\gc})$.
\end{proof}

\begin{lemma}
    \label{lem:intertwine}
    If $\fc|\fc'$ and $\gc|\gc'$ are Bi-Holant-indistinguishable and quantum-nonvanishing, then there 
    exist $\varnothing \neq Z \subset [q]$ and
    $T_1,T_2 \in \gl_q$ such that, after transforming $\fc|\fc'$ by $T_1$ and
    $\gc|\gc'$ by $T_2$, 
    every $\dprop{\fc|\fc'}{n}{0} \ni F \leftrightsquigarrow G \in \dprop{\gc|\gc'}{n}{0}$
    and $\dprop{\fc|\fc'}{0}{n} \ni F' \leftrightsquigarrow G' \in \dprop{\gc|\gc'}{0}{n}$ 
    satisfy
    \begin{equation}
        \label{eq:ix_intertwine}
        (I_Z^\shortuparrow)^{\otimes n} F = G \text{ and } F' = G' (I_Z^\shortuparrow)^{\otimes n}.
    \end{equation}
\end{lemma}
\begin{proof}
    \autoref{lem:nonconstructive} gives an
    $H \in \stab(\prop{\prop{\fc|\fc'} \oplus \prop{\gc|\gc'}})$ with, WLOG, $H|_{\gc,\fc}
    \neq 0$. Choose $T_1,T_2 \in \gl_q$ so that $T_2 H_{\gc,\fc} T_1^{-1} = 
    I_Z^\shortuparrow \in \kk^{q \times q}$ 
    for some $Z \subset [q]$ with $|Z| = \text{rank}(H_{\gc,\fc}) > 0$. Transform
    $\fc|\fc'$ by $T_1$ and $\gc|\gc'$ by $T_2$. By \autoref{prop:respect}, this transforms
    $\prop{\prop{\fc|\fc'} \oplus \prop{\gc|\gc'}}$ to
    \[
        \prop{(T_1 \oplus T_2) \big(\prop{\fc|\fc'} \oplus \prop{\gc|\gc'}\big)}
        = (T_1 \oplus T_2) \prop{\prop{\fc|\fc'} \oplus \prop{\gc|\gc'}}.
        \label{eq:oplus_transform}
    \]
    By \autoref{thm:duality}, $H$ satisfied $H \cdot K = K$ for every 
    $K \in \prop{\prop{\fc|\fc'} \oplus \prop{\gc|\gc'}}$. Hence
    \[
        \widetilde{H} := (T_1 \oplus T_2) H (T_1 \oplus T_2)^{-1} =
        \begin{bmatrix}T_1 & 0 \\ 0 & T_2 \end{bmatrix}
        \begin{bmatrix} * & * \\ H_{\gc,\fc} & *\end{bmatrix}
        \begin{bmatrix}T_1^{-1} & 0 \\ 0 & T_2^{-1} \end{bmatrix}
        = \begin{bmatrix} * & * \\ I_Z^\shortuparrow & * \end{bmatrix}
    \]
    stabilizes every signature in $\prop{\prop{\fc|\fc'} \oplus \prop{\gc|\gc'}}$ after the
    transformation by $(T_1 \oplus T_2)$.

    Let $\fc \ni F \leftrightsquigarrow G \in \gc$ have
    arity $n$. Then $(F \otimes I) \oplus (G \otimes I) \in 
    \dprop{\prop{\fc|\fc'} \oplus \prop{\gc|\gc'}}{n+1}1$, so 
    $\widetilde{H}^{\otimes n+1} ((F \otimes I) \oplus (G \otimes I)) = 
    \big((F \otimes I) \oplus (G \otimes I)\big)
    \widetilde{H}$, which in $(V(\fc),V(\gc))$-block matrix form (see e.g. \cite[Appendix A]{orthogonal}) is
    \begin{equation}
        \begin{bmatrix} 
            * & * & \hdots & * \\
            \vdots & \vdots & \iddots & \vdots \\
            * & * & \hdots & *\\
            (I_Z^\shortuparrow)^{\otimes n+1} & * & \hdots & *
        \end{bmatrix}
        \begin{bmatrix} F \otimes I & 0 \\ 0 & 0 \\ \vdots & \vdots \\ 0 & 0 \\ 0 & G \otimes I
        \end{bmatrix} =
            \begin{bmatrix} F \otimes I & 0 \\ 0 & 0 \\ \vdots & \vdots \\ 0 & 0 \\ 0 & G \otimes I
        \end{bmatrix} \begin{bmatrix} * & * \\ I_Z^\shortuparrow & * \end{bmatrix}.
        \label{eq:block_matrix}
    \end{equation}
    The bottom left block of \eqref{eq:block_matrix} gives
    \begin{equation}
        \big((I_Z^\shortuparrow)^{\otimes n} F\big) \otimes I_Z^\shortuparrow
        = (I_Z^\shortuparrow)^{\otimes n+1} (F \otimes I)
        = (G \otimes I) I_Z^\shortuparrow = G \otimes I_Z^\shortuparrow
        \label{eq:tensor_i}
    \end{equation}
    (see \autoref{fig:tensor_i}), which implies that $(I_Z^\shortuparrow)^{\otimes n} F = G$. 
    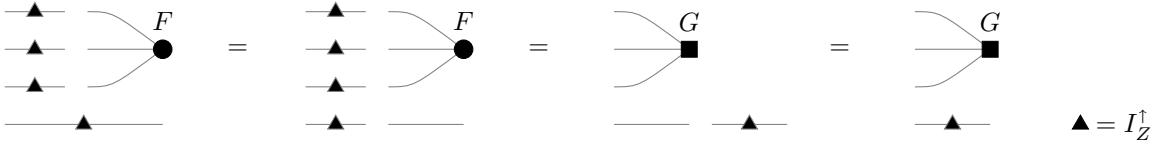
\begin{figure}[ht!]
        \centering
        \begin{tikzpicture}[scale=1]
\tikzstyle{every node}=[font=\small]
\GraphInit[vstyle=Classic]
\SetUpEdge[style=-]
\SetVertexMath

\def\ysh{0.5}
\def\wlen{1}
\def\hwlen{0.8}
\def\wgap{0.3}

\foreach \y in {-1,0,1} {
    \draw[thin, color=gray] (0-\wlen, \y*\ysh) 
    .. controls (0-\wlen+\wlen/3,\y*\ysh) .. (0,0);
    \draw[thin, color=gray] (0-\hwlen-\wlen-\wgap,\y*\ysh) -- (0-\wlen-\wgap,\y*\ysh)
    node[draw, fill=black, regular polygon, regular polygon sides=3, minimum size = 7pt, inner sep = 0pt, pos=0.5] {};
}

\draw[thin, color=gray] (-\hwlen-\wlen-\wgap,-2*\ysh) -- (0,-2*\ysh)
node[draw, fill=black, regular polygon, regular polygon sides=3, minimum size = 7pt, inner sep = 0pt, pos=0.5] {};

\tikzset{VertexStyle/.style = {shape=circle, fill=black, minimum size=7pt, inner sep=1pt, draw}}
\Vertex[x=0,y=0,L=F,Lpos=90]{f}
\node at (1,0) {$=$};

\begin{scope}[xshift=4cm]
    \foreach \y in {-1,0,1} {
        \draw[thin, color=gray] (0-\wlen, \y*\ysh) 
        .. controls (0-\wlen+\wlen/3,\y*\ysh) .. (0,0);
        \draw[thin, color=gray] (0-\hwlen-\wlen-\wgap,\y*\ysh) -- (0-\wlen-\wgap,\y*\ysh)
        node[draw, fill=black, regular polygon, regular polygon sides=3, minimum size = 7pt, inner sep = 0pt, pos=0.5] {};
    }

    \draw[thin, color=gray] (-\hwlen-\wlen-\wgap,-2*\ysh) -- (-\wlen-\wgap,-2*\ysh)
    node[draw, fill=black, regular polygon, regular polygon sides=3, minimum size = 7pt, inner sep = 0pt, pos=0.5] {};
    \draw[thin, color=gray] (0,-2*\ysh) -- (-\wlen,-2*\ysh);
    \Vertex[x=0,y=0,L=F,Lpos=90]{f2}
    \node at (1,0) {$=$};
\end{scope}

\tikzset{VertexStyle/.style = {shape=rectangle, fill=black, minimum size=6pt, inner sep=1pt, draw}}
\begin{scope}[xshift=7cm]
    \foreach \y in {-1,0,1} {
        \draw[thin, color=gray] (0-\wlen, \y*\ysh) 
        .. controls (0-\wlen+\wlen/3,\y*\ysh) .. (0,0);
    }

    \draw[thin, color=gray] (0,-2*\ysh) -- (-\wlen,-2*\ysh);
    \draw[thin, color=gray] (\wgap,-2*\ysh) -- (\wgap+\wlen,-2*\ysh)
    node[draw, fill=black, regular polygon, regular polygon sides=3, minimum size = 7pt, inner sep = 0pt, pos=0.5] {};

    \Vertex[x=0,y=0,L=G,Lpos=90]{g2}
    \node at (1+\wlen,0) {$=$};
\end{scope}

\begin{scope}[xshift=11cm]
    \foreach \y in {-1,0,1} {
        \draw[thin, color=gray] (0-\wlen, \y*\ysh) 
        .. controls (0-\wlen+\wlen/3,\y*\ysh) .. (0,0);
    }
    \draw[thin, color=gray] (0,-2*\ysh) -- (-\wlen,-2*\ysh)
    node[draw, fill=black, regular polygon, regular polygon sides=3, minimum size = 7pt, inner sep = 0pt, pos=0.5] {};
    \Vertex[x=0,y=0,L=G,Lpos=90]{g}

    \tikzset{VertexStyle/.style = {shape=regular polygon, regular polygon sides=3, fill=black, minimum size=7pt, inner sep=1pt, draw}}
    \Vertex[x=1.2,y=-2*\ysh,L={=I_Z^\shortuparrow}]{ax7S}

\end{scope}
\end{tikzpicture}
        \caption{Illustrating \eqref{eq:tensor_i} for $n=3$.}
        \label{fig:tensor_i}
    \end{figure}

    Similarly, if $\fc' \ni F' \leftrightsquigarrow G' \in \gc'$, then
    $(F' \otimes I) \oplus (G' \otimes I) \in 
    \dprop{\prop{\fc|\fc'} \oplus \prop{\gc|\gc'}}{1}{n+1}$, so 
    $\widetilde{H} \left((F' \otimes I) \oplus (G' \otimes I)\right) = 
    \big((F' \otimes I) \oplus (G' \otimes I)\big) \widetilde{H}^{\otimes n+1}$, or equivalently
    \[
        \begin{bmatrix} * & * \\ I_Z^\shortuparrow & * \end{bmatrix}
        \begin{bmatrix} F' \otimes I & 0 & \hdots & 0 & 0 \\ 0 & 0 & \hdots & 0 & G' \otimes I
        \end{bmatrix} 
        =  
        \begin{bmatrix} F' \otimes I & 0 & \hdots & 0 & 0 \\ 0 & 0 & \hdots & 0 & G' \otimes I
        \end{bmatrix} 
        \begin{bmatrix} 
            * & * & \hdots & * \\
            \vdots & \vdots & \iddots & \vdots \\
            * & * & \hdots & *\\
            (I_Z^\shortuparrow)^{\otimes n+1} & * & \hdots & *
        \end{bmatrix},
        \label{eq:block_matrix2}
    \]
    and the bottom left block of \eqref{eq:block_matrix2} gives
    \[
        F' \otimes I_Z^\shortuparrow = I_Z^\shortuparrow (F' \otimes I)
        = (G' \otimes I) (I_Z^\shortuparrow)^{\otimes n+1}
        = \big(G' (I_Z^\shortuparrow)^{\otimes n}\big) \otimes I_Z^\shortuparrow,
    \]
    and it follows that $F' = G' (I_Z^\shortuparrow)^{\otimes n}$.
\end{proof}

If $Z = [q]$ in \autoref{lem:intertwine}, then, since 
$\fc \subset \dprop{\fc|\fc'}{n}{0}$
and $\fc' \subset \dprop{\fc|\fc'}{0}{n}$, we already have $T_1 (\fc|\fc') = 
T_2 (\gc|\gc')$ by \eqref{eq:ix_intertwine}, hence $T_2^{-1} T_1 (\fc|\fc')
= (\gc|\gc')$. Otherwise, we must diverge from the proof strategy of 
\cite{orthogonal}. The natural continuation along those lines would
be to use \autoref{lem:intertwine} to add $I_Z^\uparrow$ to
to $\fc$ and $\gc$ while preserving indistinguishability, then split into subdomains and
apply induction. However, we cannot guarantee that these subdomains are
quantum-nonvanishing. Instead, we use \autoref{lem:intertwine}
to heavily restrict the form of $\fc|\fc'$ and $\gc|\gc'$, 
then use \autoref{lem:subdomain} to either
split into subdomains or place further restrictions on $\dprop{\fc|\fc'}{1}{1}$ and 
$\dprop{\gc|\gc'}{1}{1}$.

\begin{proof}[Proof of \autoref{thm:main2}]
    \autoref{lem:intertwine} gives $\varnothing \neq Z \subset [q]$ and $T_1,T_2$ such 
    that, after replacing $\fc|\fc'$ with $T_1(\fc|\fc')$ and $\gc|\gc'$ with 
    $T_2 (\gc|\gc')$ (which preserves indistinguishability,
    quantum-nonvanishing and $\gl_q$-orbits), \eqref{eq:ix_intertwine} is satisfied.
    As mentioned in the previous paragraph, if $Z = [q]$ then we are done. Otherwise,
    \eqref{eq:ix_intertwine} is equivalent to the statement that
    every $F' \in \dprop{\fc|\fc'}{0}{n}$ and $G \in \dprop{\gc|\gc'}{n}{0}$
    are supported only on $Z$, and furthermore $G|_Z = F|_Z$ for $F \leftrightsquigarrow G$
    and $F'|_Z = G'|_Z$ for $F' \leftrightsquigarrow G'$.
    Or, assuming WLOG that $Z = [z] \subset [q]$, every
    $\dprop{\fc|\fc'}{n}{0} \ni F \leftrightsquigarrow G \in \dprop{\gc|\gc'}{n}{0}$ 
    and $\dprop{\fc|\fc'}{0}{n} \ni F' \leftrightsquigarrow G' \in \dprop{\gc|\gc'}{0}{n}$ 
    have $(Z,\overline{Z})$-block form (with $\overline{Z} := [q] \setminus Z$)
    \begin{equation}
        \label{eq:block_summary}
        \begin{aligned}[c]
            F = \begin{bmatrix} F|_Z \\ * \\ \vdots \\ *\end{bmatrix},
            G = \begin{bmatrix} F|_Z \\ 0 \\ \vdots \\ 0\end{bmatrix},
        \end{aligned}
        \quad
        \begin{aligned}[c]
            &F' = \begin{bmatrix} G'|_Z & 0 & \hdots & 0\end{bmatrix}, \\
            &G' = \begin{bmatrix} G'|_Z & * & \hdots & *\end{bmatrix}.
        \end{aligned}
    \end{equation}

    All generators (signatures in $\fc|\fc'$ and $\gc|\gc'$) are purely 
    covariant or contravariant, so are subject to \eqref{eq:block_summary}. Say that
    $\fc|\fc'$ and $\gc|\gc'$ have \emph{skew blocks} if the purely covariant/contravariant
    signatures in $\prop{\fc|\fc'}$ and $\prop{\gc|\gc'}$ have zero blocks matching 
    \eqref{eq:block_summary}. We will use quantum-nonvanishing to force the $*$ blocks
    in \eqref{eq:block_summary} to be 0, at which point $\fc|\fc' = \gc|\gc'$.
    \renewcommand{\qedsymbol}{$\blacksquare$}
    \begin{claim}
        \label{cl:k_blocks}
        Let $\vk$ be a nontrivial (not just a wire) $(1,1)$-$\fc|\fc'$-gadget with signature
        $K$ and let $\widetilde{K}$ be the signature of $\vk_{\fc|\fc'\to\gc|\gc'}$. 
        If $\fc|\fc'$ and $\gc|\gc'$ have skew blocks, then
        \begin{equation}
            \label{eq:k_blocks}
            K = \begin{bmatrix} K|_Z & 0 \\ * & 0\end{bmatrix} \text{ and }
            \widetilde{K} = \begin{bmatrix} \widetilde{K}|_Z & * \\ 0 & 0\end{bmatrix}.
        \end{equation}
    \end{claim}
    \begin{proof}
        Since $\vk$ is nontrivial, it must contain at least one signature in both $\fc$ and
        $\fc'$ to preserve covariant/contravariant balance.
        The right input to $\vk$ is incident to an $F' \in \fc'$, which by 
        \eqref{eq:block_summary} is only supported on $Z$. Similarly, the left input to 
        $\vk_{\fc|\fc'\to\gc|\gc'}$ is incident to a $G \in \gc$, which is only supported
        on $Z$. This completes the proof of \autoref{cl:k_blocks}.
    \end{proof}

    Say that $T \in \gl_q$ is $(Z,\overline{Z})$\emph{-lower-triangular} if it has block form
    $T = \left[\begin{smallmatrix} T|_Z & 0 \\ T|_{\overline{Z},Z} & T|_{\overline{Z}}\end{smallmatrix}\right]$. 
    Define $(Z,\overline{Z})$\emph{-upper-triangular} similarly.
    \begin{claim}
        \label{cl:skew}
        If $\fc|\fc'$ and $\gc|\gc'$ have skew blocks and $T$ and $U$ are 
        $(Z,\overline{Z})$-lower- and upper-triangular, respectively, then $T(\fc|\fc')$ and 
        $U(\gc|\gc')$ have skew blocks and $(T \cdot K)|_Z = T|_Z \cdot K|_Z$ for every 
        purely covariant or contravariant $K \in \prop{\fc|\fc'} \cup \prop{\gc|\gc'}$.
    \end{claim}
    \begin{proof}
        The transformations $T$ and $U$ act on every 
        $F \in \dprop{\fc|\fc'}{n}{0}$, $F' \in \dprop{\fc|\fc'}{0}{n}$, 
        $G \in \dprop{\gc|\gc'}{n}{0}$, and $G' \in \dprop{\gc|\gc'}{0}{n}$ as
        \begin{align*}
            &F \mapsto T^{\otimes n} F 
            = \begin{bmatrix} 
                (T|_Z)^{\otimes n} & 0 & \hdots & 0 \\
                * & * & \hdots & 0\\
                \vdots & \vdots & \ddots & \vdots \\
                * & * & \hdots & *
            \end{bmatrix}
            \begin{bmatrix} F|_Z \\ * \\ \vdots \\ * \end{bmatrix}
            = \begin{bmatrix} (T|_Z)^{\otimes n} F|_Z \\ * \\ \vdots \\ * \end{bmatrix}, \\
            &F' \mapsto F' (T^{-1})^{\otimes n}
            = \begin{bmatrix} F'|_Z & 0 & \hdots & 0 \end{bmatrix}
            \begin{bmatrix} 
                (T|_Z^{-1})^{\otimes n} & 0 & \hdots & 0 \\
                * & * & \hdots & 0\\
                \vdots & \vdots & \ddots & \vdots \\
                * & * & \hdots & *
            \end{bmatrix}
            = \begin{bmatrix} F'|_Z (T|_Z^{-1})^{\otimes n} & 0 & \hdots & 0 \end{bmatrix}, \\
            &G \mapsto U^{\otimes n} G
            = \begin{bmatrix} 
                (U|_Z)^{\otimes n} & * & \hdots & * \\
                0 & * & \hdots & *\\
                \vdots & \vdots & \ddots & \vdots \\
                0 & 0 & \hdots & *
            \end{bmatrix}
            \begin{bmatrix} G|_Z \\ 0 \\ \vdots \\ 0 \end{bmatrix}
            = \begin{bmatrix} (U|_Z)^{\otimes n} G|_Z \\ 0 \\ \vdots \\ 0 \end{bmatrix}, \\
            &G' \mapsto G' (U^{-1})^{\otimes n}
            = \begin{bmatrix} G'|_Z & * & \hdots & * \end{bmatrix}
            \begin{bmatrix} 
                (U|_Z^{-1})^{\otimes n} & * & \hdots & * \\
                0 & * & \hdots & *\\
                \vdots & \vdots & \ddots & \vdots \\
                0 & 0 & \hdots & *
            \end{bmatrix}
            = \begin{bmatrix} G'|_Z (U|_Z^{-1})^{\otimes n} & * & \hdots & * \end{bmatrix}.
            \qedhere
        \end{align*}
    \end{proof}

    \begin{claim}
        \label{cl:delta}
        Suppose $\fc|\fc'$ and $\gc|\gc'$ have skew blocks and
        $\prop{\fc|\fc'} \ni I_X^\shortuparrow
        \leftrightsquigarrow I_X^\shortuparrow \in
        \prop{\gc|\gc'}$, where $X = [x] \subset [z] = Z$ and $Z = X \cup \Delta$
        with $|\Delta| = \delta < |Z|$. Then there are
        $(Z,\overline{Z})$-lower-triangular $T_\delta \in \gl_q$ and 
        $(Z,\overline{Z})$-upper-triangular $U_\delta \in \gl_q$ such that, after transforming
        $\fc|\fc'$ by $T_\delta$ and $\gc|\gc'$ by $U_\delta$, $F|_Z^\shortuparrow \in
        \langle \fc|\fc' \rangle$ for every $F \in \fc$ and $G'|_Z^\shortuparrow \in \langle 
        \gc|\gc' \rangle$ for every $G' \in \gc'$.
    \end{claim}
    \begin{proof}
    We prove the claim by induction on $\delta$. If $\delta = 0$, then $X = Z$, so we 
    already have $F|_Z^\shortuparrow = (I_Z^\shortuparrow)^{\otimes n} F \in 
    \dprop{\fc|\fc'}{n}{0}$
    and $G'|_Z^\shortuparrow = G' (I_Z^\shortuparrow)^{\otimes n} \in \dprop{\gc|\gc'}{0}{n}$.

    Otherwise $\delta > 0$. First suppose every $F \in \fc$ and $G' \in \gc'$ is supported
    on only $\overline{\Delta} = X \cup \overline{Z}$ (that is, $F$ and $G'$ are 0 when given
    any input from $\Delta$). Then $F|_Z^\shortuparrow = 
    F|_X^\shortuparrow = (I_X^\shortuparrow)^{\otimes n} F \in \dprop{\fc|\fc'}{n}{0}$, 
    and similarly $G'|_Z^\shortuparrow \in \dprop{\gc|\gc'}{0}{n}$, so we are done. 
    Otherwise, there is an $F \in \fc$ or $G' \in \gc'$ that is supported on $\Delta$.
    We give a proof for the former case; the latter follows by transposed reasoning.
    There is an $\va \in [q]^n$ satisfying $F_{\va} \neq 0$ and $a_i \in \Delta$ for some $i$.
    Then, since $\Delta \subset \overline{X}$,
    \[
        0 \neq F^\Delta := (I^{\otimes i-1} \otimes I_{\overline{X}}^\shortuparrow \otimes 
        I^{\otimes n-i}) F =
        (I^{\otimes i-1} \otimes (I-I_{X}^\shortuparrow) \otimes I^{\otimes n-i})F
        \in \dprop{\fc|\fc'}{n}{0}.
    \]
    \begin{figure}[ht!]
        \centering
        \begin{tikzpicture}[scale=1]
\GraphInit[vstyle=Classic]
\SetUpEdge[style=-]
\SetVertexMath
\tikzset{VertexStyle/.style = {shape=circle, fill=black, minimum size=7pt, inner sep=1pt, draw}}

\def\ysh{0.5}
\def\wlen{1.5}
\def\boxgap{0.2}
\def\boxx{1.5}

\draw[thick, color=black] (0,0) .. controls +(0,0.3) and +(\wlen/3,0) .. (-\wlen,\ysh)
node[draw, fill=black, regular polygon, regular polygon sides=3, minimum size = 8pt, inner sep = 0pt, pos=0.6, sloped=true] {};
\draw[thick, color=black] (0,0) .. controls +(0,-0.3) and +(\wlen/3,0) .. (-\wlen,-\ysh);
\draw[thick, color=black] (0,0) -- (-\wlen,0);
\node at (-\wlen/2+0.3,\ysh+\boxgap) {\small $I_{\overline{X}}^\uparrow$};

\Vertex[x=0,y=0,L=F]{f}

\filldraw[color=black!70, fill=blue!8] (-\wlen,\ysh+\boxgap) rectangle (-\wlen-\boxx,-\ysh-\boxgap);
\node at (-\wlen-\boxx/2,0) {$\mathbf{F}'$};

\node at (-\wlen-\boxx-1.5,0) {$\langle F',F^\Delta \rangle:$};

\begin{scope}[xshift=7cm]
\draw[thin, color=gray] (0,0) .. controls +(0,0.6) and +(\wlen,0) .. 
(-\wlen-\boxx-0.5,\ysh+3*\boxgap)
node[draw, fill=black, regular polygon, regular polygon sides=3, minimum size = 8pt, inner sep = 0pt, pos=0.8] {};
\draw[thin, color=gray] (-\wlen-0.5,\ysh) -- (0.5,\ysh);
\draw[thick, color=black] (0,0) .. controls +(0,-0.3) and +(\wlen/3,0) .. (-\wlen,-\ysh);
\draw[thick, color=black] (0,0) -- (-\wlen,0);
\node at (-\wlen-\boxx/2,\ysh+4*\boxgap) {\small $I_{\overline{X}}^\uparrow$};

\Vertex[x=0,y=0,L=F]{f}

\filldraw[color=black!70, fill=blue!8] (-\wlen,\ysh+\boxgap) rectangle (-\wlen-\boxx,-\ysh-\boxgap);
\node at (-\wlen-\boxx/2,0) {$\mathbf{F}'$};

\node at (-\wlen-\boxx-1,0) {$\vk:$};
\end{scope}

\end{tikzpicture}
        \caption{Breaking an edge of the grid $\langle F', F^{\Delta}\rangle$ to produce $\vk$, with $n=3$ and $i=1$.}
        \label{fig:break_edge}
    \end{figure}
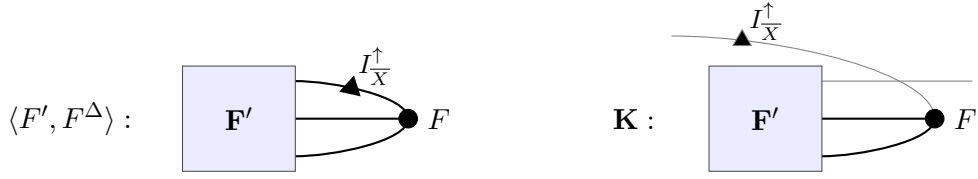
    Therefore, by the quantum-nonvanishing of $\fc|\fc'$, there is an
    quantum-$\fc|\fc'$-gadget $\mathbf{F}'$ with signature 
    $F' \in \dprop{\fc|\fc'}{0}{n}$ such that
    $\langle F', F^{\Delta}\rangle \neq 0$. 
    View $\langle F', F^\Delta\rangle$ as a $\prop{\fc|\fc'}$-grid composed of $F$, 
    $I_{\overline{X}}^\shortuparrow$, and $\mathbf{F}'$ (see \autoref{fig:break_edge}). 
    Breaking the edge between
    $I_{\overline{X}}^\shortuparrow$ and $\mathbf{F}'$ produces a
    $(1,1)$-quantum-$\fc|\fc'$-gadget $\vk$ with signature $K \in \dprop{\fc|\fc'}{1}{1}$
    such that $\tr(K) \neq 0$.
    The left input to $\vk$ is incident to
    $I_{\overline{X}}^\shortuparrow$ and the right input to $\vk$ is incident to 
    $\mathbf{F}'$, whose signature $F'$ is only supported on $Z$ by skew blocks. On the
    $\gc|\gc'$ side, every term of $\vk_{\fc|\fc'\to\gc|\gc'}$ is a nontrivial
    $\gc|\gc'$-gadget (it contains e.g. the generator $G \in \gc$ such that 
    $F \leftrightsquigarrow G$), so satisfies the condition of \autoref{cl:k_blocks}.
    Thus, letting $\widetilde{K}$ be the signature of
    $\vk_{\fc|\fc'\to\gc|\gc'}$ (so $K \leftrightsquigarrow \widetilde{K}$), 
    \[
        K = \left[\begin{array}{c:c|c} 
            0 & 0 & 0 \\ \hdashline K|_{\Delta,X} & K|_\Delta & 0 \\
            \hline && \\ K|_{\overline{Z}, X} & K|_{\overline{Z}, \Delta} &
            \phantom{---}0\phantom{---} \\ &\phantom{x}&
        \end{array}\right] \text{ and }
        \widetilde{K} = \left[\begin{array}{c:c|c} 
            \widetilde{K}|_X & \widetilde{K}|_{X,\Delta} & \widetilde{K}|_{X,\overline{Z}} \\
            \hdashline \widetilde{K}|_{\Delta,X} & \widetilde{K}|_\Delta & \widetilde{K}|_{\Delta,\overline{Z}} \\ 
            \hline && \\ 0 & 0 & 
            \phantom{---}0\phantom{---} \\ &\phantom{x}&
        \end{array}\right].
    \]
    Now $\tr(K) \neq 0$ implies that $\tr(K|_\Delta) \neq 0$ and
    \[
        \prop{\fc|\fc'}_{\overline{X}} \ni K|_{\overline{X}}
        = \left[\begin{array}{c|c} 
            K|_{\Delta} & \phantom{-} 0 \phantom{-} \\ 
            \hline K|_{\overline{Z}, \Delta} & \phantom{-} 0 \phantom{-}
        \end{array}\right]
        \leftrightsquigarrow
        \left[\begin{array}{c|c} 
            \widetilde{K}|_{\Delta} &  \widetilde{K}|_{\Delta,\overline{Z}} \\ 
            \hline 0 & 0
        \end{array}\right]
        = \widetilde{K}|_{\overline{X}} \in \prop{\gc|\gc'}_{\overline{X}}.
    \]
    By \autoref{prop:subdomain_indist}
    and \autoref{prop:subdomain_nondegen}, $\prop{\fc|\fc'}_{\overline{X}}$ and
    $\prop{\gc|\gc'}_{\overline{X}}$ are Bi-Holant-indistinguishable and quantum-nonvanishing.
    Hence $\tr(\widetilde{K}|_\Delta) = \tr(K|_\Delta) \neq 0$. Thus $K|_{\overline{X}}$
    and $\widetilde{K}|_{\overline{X}}$ do not have singleton spectrum, so, by
    \autoref{lem:subdomain}, there are 
    $T = \left[\begin{smallmatrix} 
    T|_\Delta & 0 \\ T|_{\overline{Z},\Delta} & T|_{\overline{Z}} \end{smallmatrix}\right] 
    \in \gl\big(\kk^{\overline{X}}\big)$
    and $U = \left[\begin{smallmatrix} 
    U|_\Delta & U|_{\Delta,\overline{Z}} \\ 0 & U|_{\overline{Z}} \end{smallmatrix}\right] 
    \in \gl\big(\kk^{\overline{X}}\big)$ 
    such that, after transforming $\prop{\fc|\fc'}_{\overline{X}}$ by $T$ and 
    $\prop{\gc|\gc'}_{\overline{X}}$ by $U$, we obtain $\prop{\prop{\fc|\fc'}_{\overline{X}}} \ni
    I_{X'}^{\shortuparrow \overline{X}} \leftrightsquigarrow I_{X'}^{\shortuparrow \overline{X}} \in
    \prop{\prop{\gc|\gc'}_{\overline{X}}}$ for some $\varnothing \neq X' \subset \Delta$. 
    Apply the $(Z,\overline{Z})$-lower-triangular and $(Z,\overline{Z})$-upper-triangular 
    transformations
    \[
        I_X \oplus T = \left[\begin{array}{c:c|c} 
            I_X & 0 & 0 \\ \hdashline 0 & T|_\Delta & 0 \\
            \hline && \\ 0 & T|_{\overline{Z},\Delta} & \phantom{--} T|_{\overline{Z}} 
            \phantom{--} \\ &\phantom{x}&
        \end{array}\right] \in \gl_q 
        \text{ and }
        I_X \oplus U = \left[\begin{array}{c:c|c} 
            I_X & 0 & 0 \\ \hdashline 0 & U|_\Delta & U|_{\Delta,\overline{Z}} \\
            \hline && \\ 0 & 0 & \phantom{--} U|_{\overline{Z}} 
            \phantom{--} \\ &\phantom{x}&
        \end{array}\right]
        \in \gl_q
    \]
    to $\fc|\fc'$ and $\gc|\gc'$, respectively.
    This preserves $I_X^\shortuparrow$ and $I_{\overline{X}}^\shortuparrow$ in 
    $\prop{\fc|\fc'}$ and $\prop{\gc|\gc'}$, preserves skew blocks by \autoref{cl:skew},
    and, by the above, realizes $\prop{\prop{\fc|\fc'}_{\overline{X}}} \ni
    I_{X'}^{\shortuparrow \overline{X}} \leftrightsquigarrow I_{X'}^{\shortuparrow \overline{X}} \in
    \prop{\prop{\gc|\gc'}_{\overline{X}}}$. Now, by \autoref{prop:uparrow},
    \[
        \prop{\fc|\fc'} \ni \big(I_{X'}^{\shortuparrow \overline{X}}\big)^\shortuparrow
        = I_{X'}^\shortuparrow \leftrightsquigarrow I_{X'}^\shortuparrow =
        \big(I_{X'}^{\shortuparrow \overline{X}}\big)^\shortuparrow \in \prop{\gc|\gc'}.
    \]
    Hence $\prop{\fc|\fc'} \ni I_{X \cup X'}^\shortuparrow = I_X^\shortuparrow + 
    I_{X'}^\shortuparrow \leftrightsquigarrow 
    I_{X \cup X'}^\shortuparrow \in \prop{\gc|\gc'}$. We have $Z = X \cup \Delta
    = (X \cup X') \cup \Delta'$, where $\Delta' = \Delta \setminus X'$. This 
    $\delta' := |\Delta'| < |\Delta| = \delta$, so, by induction (with $X := X \cup X'$),
    there exist $(Z,\overline{Z})$-lower- and upper-triangular transformations $T_{\delta'}$ 
    and $U_{\delta'}$ after which $\prop{\fc|\fc'}$ and $\prop{\gc|\gc'}$ contain the desired 
    $F|_Z^\shortuparrow$ and $G'|_Z^\shortuparrow$. In total, we have applied
    $T_{\delta} := T_{\delta'} \circ (I_X \oplus T)$ 
    and $U_{\delta} := U_{\delta'} \circ (I_X \oplus U)$, which, since both components are
    $(Z,\overline{Z})$-lower (resp. upper)-triangular, 
    are $(Z,\overline{Z})$-lower (resp. upper)-triangular.
    This completes the proof of \autoref{cl:delta}.
    \end{proof}

    Unless (by quantum-nonvanishing and covanishing) 
    $\fc|\fc'$ and $\gc|\gc'$ consist only of zero signatures, there is a nonzero
    $F' \in \fc'$, and there is an $\fc|\fc'$-grid $\Omega$ containing $F'$
    with $\holant(\Omega) \neq 0$. Breaking an edge of $\Omega$ yields a nontrivial
    binary $\fc|\fc'$-gadget
    $\vk$ whose signature $K \in \dprop{\fc|\fc'}{1}{1}$ satisfies $\tr(K) 
    = \holant(\Omega) \neq 0$. 
    By indistinguishability, the signature $\widetilde{K} \in \dprop{\gc|\gc'}{1}{1}$ of 
    $\vk_{\fc|\fc'\to\gc'|\gc'}$ has the same nonzero trace. \autoref{cl:k_blocks} asserts
    that $K$ and $\widetilde{K}$ have the form \eqref{eq:k_blocks}, so
    $\tr(K|_Z) = \tr(\widetilde{K}|_Z) \neq 0$. Therefore $K$ and $\widetilde{K}$ do not have 
    singleton spectrum, so, by \autoref{lem:subdomain}, we may transform $\fc$ by $T = 
    \left[\begin{smallmatrix} T|_Z & 0 \\ T|_{\overline{Z},Z} & T|_{\overline{Z}} 
    \end{smallmatrix}\right]$ and $\gc$ by
            $U = \left[\begin{smallmatrix} U|_Z & U|_{Z,\overline{Z}} \\ 0 & U|_{\overline{Z}} \end{smallmatrix}\right]$ to obtain
    $\prop{\fc|\fc'} \ni I_X^\shortuparrow \leftrightsquigarrow
    I_X^\shortuparrow  \in \prop{\gc|\gc'}$
    for some $\varnothing \neq X = [x] \subset Z$. By \autoref{cl:skew}, these transformations
    preserve skew blocks. Thus \autoref{cl:delta} applies and we obtain
    $T_\delta$ and $U_\delta$ under which $F|_Z^\shortuparrow \in
    \langle \fc|\fc' \rangle$ for every $F \in \fc$ and $G'|_Z^\shortuparrow \in \langle 
    \gc|\gc' \rangle$ for every $G' \in \gc'$. After the combined transformations 
    $T_\delta \circ T = \left[\begin{smallmatrix} (T_\delta T)|_Z & * \\ 0 & *\end{smallmatrix}\right]$ and 
    $U_\delta \circ U = \left[\begin{smallmatrix} (U_\delta U)|_Z & 0 \\ * & *\end{smallmatrix}\right]$,
    \eqref{eq:block_summary} becomes, by \autoref{cl:skew},
    \begin{equation}
        \label{eq:block_summary2}
        \begin{aligned}[c]
            F = \begin{bmatrix} ((T_\delta T)|_Z)^{\otimes n} \widetilde{F}|_Z \\ * \\ \vdots \\ *\end{bmatrix},
            G = \begin{bmatrix} ((U_\delta U)|_Z)^{\otimes n} \widetilde{F}|_Z \\ 0 \\ \vdots \\ 0\end{bmatrix},
        \end{aligned}
        \quad
        \begin{aligned}[c]
            &F' = \begin{bmatrix} \widetilde{G}'|_Z ((T_\delta T)|_Z^{-1})^{\otimes n} & 0 & \hdots & 0\end{bmatrix}, \\
            &G' = \begin{bmatrix} \widetilde{G}'|_Z ((U_\delta U)|_Z^{-1})^{\otimes n} & * & \hdots & *\end{bmatrix}.
        \end{aligned}
    \end{equation}
    for every $\dprop{\fc|\fc'}{n}{0} \ni F \leftrightsquigarrow G \in \dprop{\gc|\gc'}{n}{0}$
    and $\dprop{\fc|\fc'}{0}{n} \ni F' \leftrightsquigarrow G' \in \dprop{\gc|\gc'}{0}{n}$
    (where $\widetilde{F}$ and $\widetilde{G}'$ are the pre-transformation $F$ and $G'$).
    For $F \in \fc$, we now have $F - F|_Z^\shortuparrow \in \dprop{\fc|\fc'}{n}{0}$, and
    \eqref{eq:block_summary2} gives
    \[
        \langle (F-F|_Z^\shortuparrow), F'\rangle = 
        \langle (F-F|_Z^\shortuparrow)|_Z, F'|_Z \rangle
        = \langle 0, F'|_Z \rangle = 0
    \]
    for every $F' \in \dprop{\fc|\fc'}{0}{n}$, 
    so the quantum-nonvanishing of $\fc|\fc'$ implies that $F - F|_Z^\shortuparrow = 0$.
    Similarly, every $G' - G'|_Z^\shortuparrow = 0$. So \eqref{eq:block_summary2} is
    \[
        \begin{aligned}[c]
            F = \begin{bmatrix} ((T_\delta T)|_Z)^{\otimes n} \widetilde{F}|_Z \\ 0 \\ \vdots \\ 0\end{bmatrix},
            G = \begin{bmatrix} (U_\delta U)|_Z^{\otimes n} \widetilde{F}|_Z \\ 0 \\ \vdots \\ 0\end{bmatrix},
        \end{aligned} ~
        \begin{aligned}[c]
            &F' = \begin{bmatrix} \widetilde{F}'|_Z ((T_\delta T)|_Z^{-1})^{\otimes n} & 0 & \hdots & 0\end{bmatrix}, \\
            &G' = \begin{bmatrix} \widetilde{F}'|_Z ((U_\delta U)|_Z^{-1})^{\otimes n} & 0 & \hdots & 0\end{bmatrix}
        \end{aligned}
    \]
    for every $\fc \ni F \leftrightsquigarrow G \in \gc$ and
    $\fc' \ni F' \leftrightsquigarrow G' \in \gc'$. After a final transformation of
    $\fc|\fc'$ by $(T_\delta T)|_Z^{-1} \oplus I_{\overline{Z}} \in \gl_q$ and 
    $\gc|\gc'$ by $(U_\delta U)|_Z^{-1} \oplus I_{\overline{Z}} \in \gl_q$, we obtain
    $\fc|\fc' = \gc|\gc'$.
    \renewcommand{\qedsymbol}{$\Box$}
\end{proof}

We conclude this section by noting that \autoref{thm:duality} applies to any field $\kk$ of
characteristic 0. However, the
multitude of Jordan decompositions performed -- via \autoref{lem:subdomain} -- in the proof 
of \autoref{thm:main2} necessitate the extra assumption that $\kk$ is algebraically closed. 
Indeed, \autoref{thm:main2} does not hold without this assumption. For example, let 
$\kk = \rr$ and consider $\fc = (=_2|=_2)$ and $\gc = (-(=_2)|-(=_2))$. Every
$\gc$-grid must contain an equal number of covariant and contravariant $-(=_2)$ signatures,
hence an even number of total signatures. Therefore $\fc$ and $\gc$ are 
(Bi-)Holant-indistinguishable. Furthermore, if $K \in \dprop{\gc}{\ell}{r}$, then construct
$\pm K^\top \in \dprop{\gc}{r}{\ell}$ by connecting a left-facing $-(=_2)$ to every right 
input of $K$ and connecting a right-facing $-(=_2)$ to each left input of $K$ (this exchanges the
left and right inputs of $K$ while preserving the underlying signature up to a global $\pm$).
Now $\langle K, \pm K^\top \rangle \neq 0$ (as this is effectively a
contraction of $K$ with itself), so $K$ is $\gc$-nonvanishing. 
Thus $\gc$ and, similarly, $\fc$, are quantum-nonvanishing. \autoref{thm:main2} guarantees the
existence of a complex $T$ (in this case $T = iI$) transforming $\fc$ to $\gc$, but any such
$T$ must satisfy $TT^\top = T (=_2)^{1,1} T^\top = (-(=_2))^{1,1} = -I$, which is impossible 
for real-valued $T$.


\section{More Corollaries of the Main Theorems}
\label{sec:corollaries}
In this section, we exploit the expressive power of Holant and Bi-Holant to derive 
novel consequences of Theorems \ref{thm:main1} and \ref{thm:main2}. 
We begin with a complex generalization of \autoref{thm:orthogonal}, which does not hold as 
stated for complex-valued signatures -- for example, consider $\fc = \{0\}$ and the vanishing
set $\gc$ containing the single unary signature $[1,i]$. Say that 
$\fc \subset \tc(\cc^q)$ is \emph{conjugate-closed} if $F \in \fc \iff \overline{F} \in \fc$,
where $\overline{F}$ is the entrywise complex conjugate of $F$ (note that any real-valued
set is conjugate-closed). Young \cite[Section 6.1]{orthogonal} conjectured the following
extension of \autoref{thm:orthogonal}, which we
we confirm using \autoref{thm:main2}.
\begin{corollary}
    \label{cor:complex}
    Suppose $\fc,\gc \subset \tc(\cc^q)$ are conjugate-closed. Then $\fc$ and $\gc$
    are Holant-indistinguishable if and only if there is a complex orthogonal matrix $T$
    such that $T \fc = \gc$.
\end{corollary}
\begin{proof}
    By \autoref{prop:shapeless}, $\fc$ and $\gc$ are Holant-indistinguishable if and only if
    $=_2|\fc , =_2$ and $=_2|\gc , =_2$
    are Holant-indistinguishable. Now the
    $(\Leftarrow)$ direction follows from \autoref{prop:orthogonal} and \autoref{thm:holant}.
    We will show that $=_2|\fc , =_2$ is quantum-nonvanishing (the $\gc$ argument
    is similar); then the $(\Rightarrow)$ direction follows from
    \autoref{prop:orthogonal} and \autoref{thm:main2} with $\kk = \cc$. Let 
    $0 \neq K \in \dprop{=_2|\fc , =_2}{\ell}{r}$. 
    By definition, $K = \sum_{i=1}^m c_i K_i$, where
    each $c_i \in \cc$ and each $K_i$ is the signature of a $(=_2|\fc , =_2)$-gadget 
    $\vk_i$. Since $\fc$ is conjugate closed, each entrywise conjugate $\overline{K_i}$ is
    the signature of the $(=_2|\fc , =_2)$-gadget constructed from $\vk_i$ with 
    replacing every $F \in \fc$ in $\vk_i$ by $\overline{F} \in \fc$. Therefore
    $\overline{K} = \sum_{i=1}^m \overline{c_i} \overline{K_i} \in \prop{=_2|\fc , =_2}$.
    Now construct the dual $K^* \in \dprop{=_2|\fc , =_2}{r}{\ell}$ 
    by connecting a
    left-facing $=_2$ to each right input of $\overline{K}$ and connecting a 
    right-facing $=_2$ to each left input of $\overline{K}$. Then $\langle K,
    K^* \rangle \neq 0$, so $K^*$ witnesses that $K$ is $(=_2|\fc , =_2)$-nonvanishing.
    See \autoref{fig:complex}.
\end{proof}

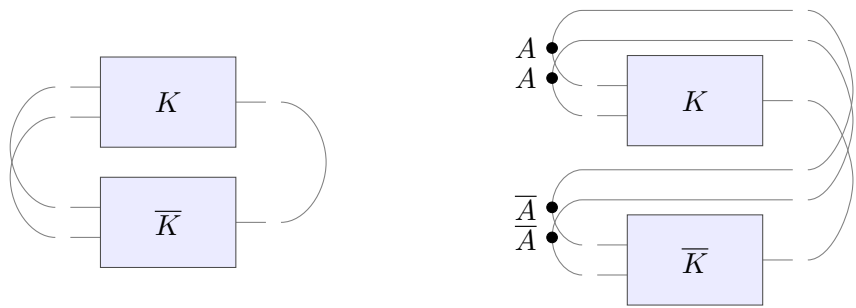
\begin{figure}[ht!]
    \begin{center}
    \begin{subfigure}{0.4\textwidth}
        \centering
        \begin{tikzpicture}[scale=.4]
\GraphInit[vstyle=Classic]
\SetUpEdge[style=-]
\SetVertexMath

\def\ysh{4}
\def\wlen{1}
\def\wgap{0.5}
\def\boxlen{4.5}

\draw[thin, color=gray] (-\wlen,1) -- (0,1);
\draw[thin, color=gray] (-\wlen,2) -- (0,2);
\draw[thin, color=gray] (\boxlen,1.5) -- (\boxlen+\wlen,1.5);
\filldraw[color=black!70, fill=blue!8] (0,0) rectangle (\boxlen,3);
\node at (\boxlen/2,1.5) {$\overline{K}$};

\begin{scope}[yshift=\ysh cm]
\draw[thin, color=gray] (-\wlen,1) -- (0,1);
\draw[thin, color=gray] (-\wlen,2) -- (0,2);
\draw[thin, color=gray] (\boxlen,1.5) -- (\boxlen+\wlen,1.5);
\filldraw[color=black!70, fill=blue!8] (0,0) rectangle (4.5,3);
\node at (\boxlen/2,1.5) {$K$};
\end{scope}

\foreach \yy in {1,2} {
\draw[thin, color=gray] (-\wlen-\wgap,\yy) .. controls
+(-0.8,0) and +(0,-1) ..
(-2.5*\wlen-\wgap,\yy+\ysh/2) .. controls
+(0,1) and +(-0.8,0) ..
(-\wlen-\wgap,\yy+\ysh);
}

\draw[thin, color=gray] (\boxlen+\wlen+\wgap,1.5) .. controls
+(0.8,0) and +(0,-1) ..
(\boxlen+2.5*\wlen+\wgap,1.5+\ysh/2) .. controls
+(0,1) and +(0.8,0) ..
(\boxlen+\wlen+\wgap,1.5+\ysh);

\end{tikzpicture}
        \vspace{0.5cm}
        \caption{The construction in \autoref{cor:complex}}
        \label{fig:complex}
    \end{subfigure}
    \begin{subfigure}{0.4\textwidth}
        \centering
        \begin{tikzpicture}[scale=.4]
\GraphInit[vstyle=Classic]
\SetUpEdge[style=-]
\SetVertexMath
\tikzset{VertexStyle/.style = {shape=circle, fill=black, minimum size=4pt, inner sep=1pt, draw}}

\def\ysh{5.3}
\def\wlen{1}
\def\wgap{0.5}
\def\boxlen{4.5}
\def\ay{2.5}

\draw[thin, color=gray] (-\wlen,1) -- (0,1);
\draw[thin, color=gray] (-\wlen,2) -- (0,2);
\draw[thin, color=gray] (\boxlen,1.5) -- (\boxlen+\wlen,1.5);
\filldraw[color=black!70, fill=blue!8] (0,0) rectangle (\boxlen,3);
\node at (\boxlen/2,1.5) {$\overline{K}$};

\foreach \yy in {1,2} {
\draw[thin, color=gray] (-\wlen-\wgap,\yy) .. controls
+(-0.5,0) and +(0,-0.7) ..
(-2*\wlen-\wgap,\yy+\ay/2) .. controls
+(0,0.7) and +(-0.5,0) ..
(-\wlen-\wgap,\yy+\ay) .. controls 
+(0.2,0) ..
(\boxlen+\wlen,\yy+\ay);
\Vertex[x=-2*\wlen-\wgap,y=\yy+\ay/2,L={\overline{A}},Lpos=180]{b\yy}
}

\begin{scope}[yshift=\ysh cm]
\draw[thin, color=gray] (-\wlen,1) -- (0,1);
\draw[thin, color=gray] (-\wlen,2) -- (0,2);
\draw[thin, color=gray] (\boxlen,1.5) -- (\boxlen+\wlen,1.5);
\filldraw[color=black!70, fill=blue!8] (0,0) rectangle (4.5,3);
\node at (\boxlen/2,1.5) {$K$};

\foreach \yy in {1,2} {
\draw[thin, color=gray] (-\wlen-\wgap,\yy) .. controls
+(-0.5,0) and +(0,-0.7) ..
(-2*\wlen-\wgap,\yy+\ay/2) .. controls
+(0,0.7) and +(-0.5,0) ..
(-\wlen-\wgap,\yy+\ay) .. controls 
+(0.2,0) ..
(\boxlen+\wlen,\yy+\ay);
\Vertex[x=-2*\wlen-\wgap,y=\yy+\ay/2,L={A},Lpos=180]{b\yy}
}

\end{scope}

\foreach \yy in {1+\ay,2+\ay,1.5} {
\draw[thin, color=gray] (\boxlen+\wlen+\wgap,\yy) .. controls
+(0.8,0) and +(0,-1) ..
(\boxlen+2.5*\wlen+\wgap,\yy+\ysh/2) .. controls
+(0,1) and +(0.8,0) ..
(\boxlen+\wlen+\wgap,\yy+\ysh);
}

\end{tikzpicture}
        \caption{The construction in \autoref{lem:nonsingular}}
        \label{fig:invertible}
    \end{subfigure}
    \end{center}
    \vspace{-0.4cm}
    \caption{Connecting quantum gadgets with their conjugates
    for quantum-nonvanishing.}
\end{figure}

\subsection{Bounded-Degree Graph Homomorphisms and \#CSP}
\label{sec:hom}
Graphs $F$ and $G$ are \emph{homomorphism indistinguishable} over
a graph class $\mathfrak{G}$ if $\hom(X,F) = \hom(X,G)$ for every
$X \in \mathfrak{G}$.
It follows from the discussion around \autoref{fig:hom} that
two graphs $F$ and $G$
are homomorphism-indistinguishable over graphs of maximum degree at most $d$ iff
$\eq_{\leq d}|\{A_F\}$ and $\eq_{\leq d}|\{A_G\}$ are Holant-indistinguishable.
More generally, for any $\fc$ consider $\csp(\fc) = \holant_{\eq\cup\fc}$ and 
$\csp^{(d)}(\fc) = \holant_{\eq_{\leq d}|\fc}$. The \emph{counting constraint 
satisfaction problem} \#CSP is a well-studied problem in counting complexity, itself the 
subject of broad dichotomy theorems 
\cite{bulatov_2013,dyer_richerby,cai-chen-lu,cai-chen-complexity}. 
In $\csp^{(d)}(\fc)$, every variable appears at most $d$ times
across all constraints \cite{bounded_csp}. In general, $\fc$
and $\gc$ are $\csp$-indistinguishable (i.e. $\eq\cup\fc$ and $\eq\cup\gc$ are 
Holant-indistinguishable) if and only if $\fc$ and $\gc$ are isomorphic \cite{young2022equality}. Putting 
$\fc = \{A_F\}$ and $\gc = \{A_G\}$, we recover the classical result of Lovász that
homomorphism-indistinguishability is equivalent to isomorphism \cite{lovasz_operations}.
This also follows from from \autoref{thm:orthogonal} and the fact that
$T$ is a permutation matrix if and only if $T \eq = \eq$ (viewing the equalities as 
contravariant) \cite{xia,orthogonal}. In fact, the following sharper characterization holds.
\begin{proposition}
    \label{prop:bounded}
    $T \in \gl_q$ is a permutation matrix if and only if 
    $T \{=_2,=_3\} = \{=_2,=_3\}$.
\end{proposition}
\begin{proof}
    Assume $T \{=_2,=_3\} = \{=_2,=_3\}$. By \autoref{prop:orthogonal}, $T$ is orthogonal and
    preserves the covariant $=_2$. Therefore, by \autoref{prop:respect},
    $T$ preserves the signature of every $(=_3|=_2)$-gadget. Every $=_n$ for $n \geq 4$
    is the signature of the $(=_3|=_2)$-gadget constructed by chaining together
    $n-2$ copies of $=_3$ using the covariant $=_2$ (and $=_1$ is realized by connecting two 
    inputs of a single $=_3$ with $=_2$). Therefore $T \eq = \eq$, so $T$ is a permutation 
    matrix.
\end{proof}
However, also recall the crucial fact that
$\holant_{\eq_{\leq d} \mid \{A_G\}}$, and, more generally, 
$\holant_{\eq_{\leq d} \mid \fc}$ are strictly bipartite problems, so \autoref{thm:orthogonal} does not apply. Instead, we
must apply the conditional \autoref{thm:main2} to obtain the following (and its extension
to $\csp^{(d)}$):
\begin{corollary}
    \label{cor:bounded}
    For $d \geq 3$,
    define $\mathfrak{N}_d$ to be the set of all graphs $G$ such that $\eq_{\leq d}|\{A_G\}$ 
    is quantum-nonvanishing. For any pair of graphs in $\mathfrak{N}_d$, 
    homomorphism-indistinguishability
    over graphs of maximum degree at most $d$ is equivalent to isomorphism.
\end{corollary}
\begin{corollary}
    \label{cor:csp_bounded}
    For $d \geq 3$, if $\eq_{\leq d}|\fc$ and $\eq_{\leq d}|\gc$ are quantum-nonvanishing, then
    $\fc$ and $\gc$ are $\csp^{(d)}$-indistinguishable if and only if $\fc$ and $\gc$ are
    isomorphic.
\end{corollary}
\autoref{cor:bounded} raises the interesting problem of characterizing when a graph is
in $\mathfrak{N}_d$. The next lemma, which generalizes the quantum-nonvanishing argument in 
\autoref{cor:complex}, implies that, if $A_G$ is invertible, then $G \in \mathfrak{N}_d$
for every $d \geq 2$. 
\begin{lemma}
    \label{lem:nonsingular}
    If $\fc \subset \tc(\cc^q)$ is conjugate-closed and satisfies 
    $(=_2) \in \dprop{\fc}{2}{0}$ and $A \in \dprop{\fc}{0}{2}$ (or vice-versa) for some
    $A$ whose matrix form $A^{1,1}$ is nonsingular, then $\fc$ is quantum-nonvanishing.
\end{lemma}
\begin{proof}
    Let $0 \neq K \in \dprop{\fc}{\ell}{r}$. If $\ell = 0$, then, as in \autoref{cor:complex},
    construct the dual $K^* \in \dprop{\fc}{r}{0}$ by connecting each right input of $K$ 
    with a copy of $(=_2) \in \dprop{\fc}{2}{0}$ and conjugating all coefficients and
    signatures composing $K$. Then $\langle K, K^*\rangle \neq 0$, so $K$ is
    $\fc$-nonvanishing. Otherwise, $(A^{1,1})^{\otimes \ell} K \neq 0$ by nonsingularity of 
    $A^{1,1}$, and therefore the signature $K' \in \dprop{\fc}{0}{\ell+r}$ formed by 
    connecting $\ell$ copies of $A$ with the $\ell$ left inputs of $K$ (equivalently, 
    connecting $\ell$ copies of right-facing $=_2$ with the $\ell$ left inputs of 
    $(A^{1,1})^{\otimes \ell} K$) is nonzero. Again, since $K'$ is now fully covariant, its
    dual $(K')^*$ is in $\dprop{\fc}{\ell+r}{0}$. 
    The $\prop{\fc}$-grid formed by contracting $K'$ and $(K')^*$ contains $K$ and has 
    nonzero value, so $K$ is $\fc$-nonvanishing. See \autoref{fig:invertible}.
\end{proof}
\begin{corollary}
    Let $\fc, \gc \subset \tc(\cc^q)$ be conjugate-closed. For $d \geq 3$,
    if there exist $A_1 \in \dprop{\eq_{\leq d}|\fc}{0}{2}$ 
    and $A_2 \in \dprop{\eq_{\leq d}|\gc}{0}{2}$ whose matrix forms are nonsingular, 
    then $\fc$ and $\gc$ are
    $\csp^{(d)}$-indistinguishable if and only if $\fc$ and $\gc$ are isomorphic.
\end{corollary}
Specializing to $\fc = \{A_F\}$ and $\gc = \{A_G\}$ for nonsingular $A_F$ and $A_G$, we obtain
\begin{corollary}
    Graphs $F$ and $G$ with nonsingular adjacency matrices are
    homomorphism-indistinguishable over graphs of maximum degree at most 3 if and only if
    they are isomorphic.
\end{corollary}
 
We may also apply \autoref{thm:main1}, which applies to all
graphs $F$ and $G$, to
obtain the first characterization of homomorphism-indistinguishability over graphs of bounded degree.
\begin{corollary}
    \label{cor:bounded_orbit_cl}
    Graphs $F$ and $G$ on $q$ vertices are homomorphism-indistinguishable over all graphs of 
    maximum degree at most $d$ if and only if $\overline{\gl_q(\eq_{\leq d}|A_F)}
    \cap \overline{\gl_q(\eq_{\leq d}|A_G)} \neq \varnothing$.

    In particular, by \autoref{cor:decidable}, the problem of deciding whether two graphs
    are homomorphism-indistinguishable over all graphs of maximum degree at most $d$ is
    decidable.
\end{corollary}

The decidability result in \autoref{cor:bounded_orbit_cl} answers another open question, and
is interesting because homomorphism-indistinguishability over some graph classes (e.g. planar
graphs \cite{planar}) is known to be undecidable. While \autoref{thm:main1} and 
\autoref{cor:decidable}-- of which 
\autoref{cor:bounded_orbit_cl} is just a special case --
do not immediatly yield a polynomial-time algorithm for testing bounded-degree 
homomorphism-indistinguishability, it is possible that efficient algorithms -- and a
more specific, algebraic characterization than that given by \autoref{cor:bounded_orbit_cl}
-- exist for this case. 

\subsection{Indistinguishability, \textbf{TOCI}, and \textbf{GI}}
Lysikov and Walter \cite{lysikov} define the class \textbf{TOCI} of (problems reducible to) orbit closure intersection problems for actions of general linear groups on finite subsets of $\bigcup_{i=1}^m \tc(\cc^{q_i})$ (sets are allowed to contain
signatures with different domains). They show that $\textbf{GI} \subset \textbf{TOCI}$ by reducing isomorphism of $q$-vertex
graphs $F$ and $G$ to $\gl_q$-orbit-intersection of $(A_F,=_3|=_2)$ and
$(A_G,=_3|=_2)$ \cite[Lemma 5.26 and Proposition 5.28]{lysikov}. Our framework gives a short alternative proof of this reduction.
First, if $F \cong G$, then, since every permutation matrix preserves $\eq$, the
$\gl_q$-orbits of $(A_F,=_3|=_2)$ and $(A_G,=_3|=_2)$ intersect. Conversely,
the `easy' $(\Longleftarrow)$ direction of \autoref{thm:main1} asserts that $(A_F,=_3|=_2)$ and $(A_G,=_3|=_2)$
are Holant-indistinguishable. As in the proof of \autoref{prop:bounded}, contravariant
$=_3$ and covariant $=_2$ together construct all of $\eq$, so 
$(A_F|\eq)$ and $(A_G|\eq)$ are Holant-indistinguishable -- that is, $F$ and $G$ are
homomorphism-indistinguishable. Then, by Lovász's theorem, $F \cong G$.
Lysikov and Walter also show that the orbit closure intersection
problem for $\fc$ containing two contravariant ternary signatures and one
covariant binary signature, all on the same domain, is \textbf{TOCI}-complete \cite[Corollary 1.3]{lysikov}.
Combining these results with \autoref{thm:main1} gives the following.
\begin{corollary}
    The following problem is $\textbf{TOCI}$-complete: Given ternary $F_3,F_3',G_3,G_3'$ and
    binary $F_2,G_2$, decide whether $(F_3,F_3'\mid F_2)$ and $(G_3,G_3'\mid G_2)$ are Holant-indistinguishable.
\end{corollary}
\begin{corollary}
    The following problem is $\textbf{GI}$-hard: Given ternary $F_3,G_3$ and binary
    $F_2,F_2',G_2,G_2'$, decide whether  $(F_2,F_3\mid F_2')$ and $(G_2,G_3\mid G_2')$ 
    are Holant-indistinguishable.
\end{corollary}

\subsection*{Acknowledgements}
The second author thanks Tim Seppelt for helpful discussions on 
homomorphism indistinguishability.

\printbibliography

\end{document}